\documentclass{amsart}
\usepackage[utf8]{inputenc}
\usepackage{amsmath}
\usepackage{array, amsfonts}
\usepackage{amsthm}
\usepackage{amssymb}
\usepackage{enumerate}
\usepackage{lineno}
\usepackage[bookmarks=true]{hyperref}
\usepackage{amsaddr}
\usepackage{enumitem}
\usepackage{mathrsfs}
\usepackage{stackengine}
\usepackage{bookmark}
\usepackage{amsrefs}
\usepackage[draft]{todonotes}
\usepackage{mathtools}
\usepackage{fancyhdr}

\usepackage{color}
\definecolor{darkgreen}{rgb}{0,0.5,0}

\DeclareMathOperator*{\myDelta}{\Delta}

\newcommand{\comment}[1]

\newcommand\myfunc[5]{%
	\begingroup
	\setlength\arraycolsep{0pt}
	#1\colon\begin{array}[t]{c >{{}}c<{{}} c}
		#2 & \to & #3 \\ #4 & \mapsto & #5 
	\end{array}%
	\endgroup}

\title{ Algorithmic Networks: \\ central time to trigger \\ expected emergent open-endedness  }

\thanks{ This work is a revised version of the Research Report no. 4/2018 at the National 
Laboratory for Scientific Computing (LNCC) that is available at: 
\url{http://www.lncc.br/departamentos/producaocientificageral.php?vMenu=2&vTipo=13&vCabecalho=pesq&vTitulo=lncc&vDepto=&idt_responsavel=&vAn}
 }

\thanks{ A preliminary version of some of the results in this article was orally presented in the the 8th International Workshop on Guided Self-Organization at the Fifteenth International Conference on the Synthesis and Simulation of Living Systems (ALIFE) under the title ``Emergent algorithmic creativity on networked Turing machines'' with an extended abstract available at \url{http://guided-self.org/gso8/program/index.html} }

\author[Felipe S. Abrah\~{a}o]{Felipe S. Abrah\~{a}o}
\email[A1]{fsa@lncc.br}

\author[Klaus Wehmuth]{Klaus Wehmuth}
\email[A2]{klaus@lncc.br}

\author[Artur Ziviani]{Artur Ziviani}
\email[A3]{ziviani@lncc.br}

\address[A1,A2,A3]{National Laboratory for Scientific Computing (LNCC)
	\\ 25651-075 – Petropolis, RJ – Brazil}

\thanks{Authors acknowledge the partial support from CNPq through their individual grants: F. S. Abrah\~{a}o (313.043/2016-7), K. Wehmuth (312599/2016-1), and A. Ziviani~(308.729/2015-3). Authors also acknowledge the INCT in Data Science -- INCT-CiD (CNPq 465.560/2014-8).}

\fancypagestyle{plain}{
	\fancyhead{}
	\fancyhead[C]{ \textit{Revised version } \\ National Laboratory for Scientific Computing (LNCC)  
	\\ \textbf{ Research Report 4/2018} }
}
\begin{document}

\maketitle\thispagestyle{plain}


\newtheorem{thm}{Theorem}[subsection]
\newtheorem{amslemma}{Lemma}[subsection]
\newtheorem{amscorollary}{Corollary}[thm]
\newtheorem{corollaryundersubsection}{Corollary}[subsection]

\newtheorem{notation}{Notation}[subsection]

\theoremstyle{definition}
\newtheorem{amsproposition}{Proposition}[subsection]
\newtheorem{amsdefinition}{Definition}[subsection]
\newtheorem{subdefinition}{Definition}[amsdefinition]
\newtheorem{subsubdefinition}{Definition}[subdefinition]
\newtheorem{autodefinition}{Definition}

\theoremstyle{remark}
\newtheorem{subnotation}{Notation}[amsdefinition]

\theoremstyle{remark}
\newtheorem{remark}{Remark}[amsdefinition]
\newtheorem{remarknote}{Note}[amsdefinition]
\newtheorem{noteunderlemma}{Note}[amslemma]
\newtheorem{noteunderthm}{Note}[thm]
\newtheorem{subremarknote}{Note}[subdefinition]
\newtheorem{subsubremarknote}{Note}[subsubdefinition]
\newtheorem{subremarknote2}{Note}[remarknote]
\newtheorem{note}{\textbf{Note}}[subsection]

\newtheorem{Bthm}{Theorem}[section]
\newtheorem{Blemma}{Lemma}[section]
\newtheorem{Bcorollary}{Corollary}[Bthm]
\newtheorem{Bcorollaryundersection}{Corollary}[section]

\theoremstyle{definition}
\newtheorem{Bproposition}{Proposition}[section]
\newtheorem{Bdefinition}{Definition}[section]
\newtheorem{Bsubdefinition}{Definition}[Bdefinition]
\newtheorem{Bsubsubdefinition}{Definition}[Bsubdefinition]
\newtheorem{Bautodefinition}{Definition}
\newtheorem{Bnotation}{Notation}[section]
\newtheorem{BnotationunderBnotation}{Notation}[Bnotation]

\theoremstyle{remark}
\newtheorem{Bsubnotation}{Notation}[Bdefinition]

\theoremstyle{remark}
\newtheorem{Bremark}{Remark}[Bdefinition]
\newtheorem{Bremarknote}{Note}[Bdefinition]
\newtheorem{BnoteunderBlemma}{Note}[Blemma]
\newtheorem{BnoteunderBthm}{Note}[Bthm]
\newtheorem{Bsubremarknote}{Note}[Bsubdefinition]
\newtheorem{Bsubsubremarknote}{Note}[Bsubsubdefinition]
\newtheorem{Bsubremarknote2}{Note}[Bremarknote]
\newtheorem{Bnote}{\textbf{Note}}[section]
\newtheorem{BnoteunderBnotation}{Note}[Bnotation]

\begin{abstract}\label{abstract}
		This article investigates emergence and complexity in complex systems that can share 
		information on a network. To this end, we use a theoretical approach from information 
		theory, computability theory, and complex networks. One key studied question is how much 
		emergent complexity (or information) arises when a population of computable systems is 
		networked compared with when this population is isolated. First, we define a general model 
		for networked theoretical machines, which we call \emph{algorithmic networks}. Then, we 
		narrow our scope to investigate algorithmic networks that optimize the average fitnesses of 
		nodes in a scenario in which each node imitates the fittest neighbor and the randomly 
		generated population is networked by a time-varying graph. We show that there are 
		graph-topological conditions that cause these algorithmic networks to have the property of 
		expected emergent open-endedness for large enough populations. In other words, the 
		expected emergent algorithmic complexity of a node tends to infinity as the population size 
		tends to infinity. Given a dynamic network, we show that these conditions imply the 
		existence of a central time to trigger expected emergent open-endedness. Moreover, we 
		show that networks with small diameter compared to the network size meet these 
		conditions. We also discuss future research based on how our results are related to some 
		problems in network science, information theory, computability theory, distributed 
		computing, game theory, evolutionary biology, and synergy in complex systems.   
\end{abstract}



\pagestyle{myheadings}
\markboth{Algorithmic Networks}{Central Time to Trigger Expected Emergent Open-Endedness}

\section{Introduction}\label{sectionIntro}

The general context of our research is to mathematically investigate emergence and complexity when a population of complex systems is networked. That is, we will study emergence of complexity (or information) in a system composed of interacting complex systems. As supported by \cite{Barabasi2009a,Michail2018}, the pursuit of a universal framework for the problem of complexity in complex networks is paramount for a wide range of topics in complex systems. The present work relies on an intersection between information theory, computability theory, evolutionary game theory, complex networks, distributed computing, multi-agent systems, communication complexity, adaptive complex systems, and biology. In fact, we discuss in further detail in Section~\ref{sectionRelatedandFuture} that the investigated problem is connected to questions ranging, for instance, from the problem of symbiosis~\cite{MargulisLynn1981}, cooperation~\cite{Hammerstein1994, Axelrod1997, Axelrod2006}, integration \cite{Feltz2006, Oizumi2014,Maguire2014}, and synergy \cite{Lizier2018} to biological \cite{Rashevsky1955,Walker2016,Kim2015}, economic \cite{Schweitzer2009, Axelrod1997, Shoham2008}, social \cite{Miller2007} networks, or artificial \cite{Lee2011,DeLillo2018} networks.

In particular, we study the problem of how a group (or population) of randomly generated computable systems give rise to emergent phenomena through the exchange of information and how it would affect the overall performance, fitness, or payoff. We create a theoretical toy model from which one can investigate how much emergent complexity the whole group has on average when they are networked compared with when they are isolated. As we will formally show later in this paper, the randomly generated algorithmic network $\mathfrak{N}_{BB}$ that plays the Busy Beaver imitation game (BBIG) --- see Section~\ref{sectionBusyBeaverGame} --- is a mathematical object useful to prove fruitful theorems using well-known results from statistics, computability theory, information theory, and graph theory. 

Thus, how does one define a measure of expected emergent complexity (or information) for a randomly generated population of Turing machines? In the context of this abstract toy model, we answer this question by generalizing the results for open-ended evolutionary systems under algorithmic information theory (AIT), as presented in \cite{Chaitin2012, Chaitin2013, Chaitin2014, Abrahao2015, Abrahao2016, Abrahao2016b, Abrahao2016c, Hernandez-Orozco2016}, in order to present a new mathematical phenomenon of \emph{open-endedness}. We will show that there are graph-topological conditions that trigger \emph{expected emergent open-endedness}~(EEOE), i.e., that trigger an unlimited increase of expected emergent complexity as the population size goes to infinity~(see Section~\ref{sectionAEOE}). It is an akin --- but different --- phenomenon to evolutionary open-endedness \cite{Bedau1998,Channon2001,Maley1999,Standish2003,Ruiz-Mirazo2004}: Instead of achieving an unbounded quantity of complexity over time~(or successive mutations), an unbounded quantity of emergent complexity is achieved as the population size increases when this population is networked (see Section \ref{subsectionOE}). In other words, the interaction (i.e., the exchange of information) among the population of randomly generated computable systems induces an endless increase in the expected (or average) emergent complexity as the population size grows toward infinity.

We start by defining a general encompassing mathematical model that we call algorithmic networks in Section~\ref{sectionAN}. This definition relies on a population of arbitrarily chosen theoretical machines and relies on a MultiAspect Graph (MAG)~\cite{Wehmuth2016b}, making the vertices in the respective MultiAspect Graph (MAG) to correspond to this population of systems/programs ~(see Definition \ref{BdefFunctionbinAN}) such that edges are communication channels that nodes/systems/programs can use to send and receive information. Thus, as it was our intention described in the first paragraph, an algorithmic network is a network composed of algorithms as nodes, with each node representing a computable system.

Then, we introduce in Section \ref{sectionBusyBeaverGame} a particular model of synchronous dynamic algorithmic networks $ \mathfrak{N}_{BB} $ that is based on simple imitation of the fittest neighbor: a type of algorithmic network that plays the BBIG. A network Busy Beaver game is a game in which each player is trying to calculate the largest integer\footnote{ As established as our measure of fitness or payoff. See Section~\ref{sectionBusyBeaverGame}.} it can using the information shared by its neighbors. The BBIG is a particular case of the Busy Beaver game in which every node can only propagate the largest integer, taking into account the one produced by itself and the ones from their neighbors. It configures a simple imitation-of-the-fittest procedure. Thus, these algorithmic networks $ \mathfrak{N}_{BB} $ can be seen as playing an optimization procedure where the whole pursues the increase of the average fitness/payoff through diffusing on the network the best randomly generated solution~(see discussion on Section~\ref{subsectionDiscussiononBBIM...}).

We present our main Theorem \ref{BthmMain} proving that there is a lower bound for the expected emergent algorithmic complexity in algorithmic networks $ \mathfrak{N}_{BB} $. Additionally, we prove Corollary \ref{BcorMain} showing that this lower bound can be calculated from a diffusion measure like cover time \cite{Costa2015a}. Further, from this corollary, we also prove in Theorem~\ref{BthmMainCentralTime} that there are asymptotic conditions on the increasing diffusion power of the cover time (as a function of the population size) such that they ensure that there is a central time to trigger EEOE. Then, we introduce in Section~\ref{sectionDiameter} a small modification on the family of MultiAspect Graphs of $\mathfrak{N}_{BB} $ with the purpose of investigating what would happen if the time-varying networks present a small diameter compared to the network size, i.e., $D(G_t,t) = \mathbf{ O }(\lg(N))$. Indeed, in this case, we show in Corollary \ref{BcorDiameter} that a small diameter is sufficient to ensure the existence of a central time\footnote{ Or just to trigger EEOE, if the network is static. See Note~\ref{BnoteDiameterandstaticgraphs} and Definition~\ref{BdefStaticNetwork}.} to trigger EEOE for algorithmic networks $ \mathfrak{N}_{BB} $ --- even in a computably larger number of communication rounds compared to the temporal diffusion diameter.

Additionally, we discuss in Section~\ref{sectionRelatedandFuture} future research from how our results are related to problems in network science, statistical (or probabilistic) information theory, computability theory, distributed systems, automata theory, game theory, evolutionary biology, and synergy in complex systems. We also choose to add discussions in the Sections in order to improve the readability and explanation of the new definitions and new models that we will introduce in the present article.
Finally, Section~\ref{sectionConclusion} concludes the paper. Complementarily, the appendix shows extended versions of the proofs of Lemmas \ref{BlemmaSLLNandAIT}, \ref{BlemmaComplexityp_i}, \ref{BlemmaGibbsandalgorithmicentropy}, \ref{BlemmaComplexityonHaltp_i}, \ref{BlemmaComplexityonBarHalt}, and \ref{BlemmaMinComplexityonDiffusion}, as well as of Theorems~\ref{BthmMain} and~\ref{BthmMainCentralTime}.

\section{Algorithmic Networks}\label{sectionAN}

\subsection{Discussion on algorithmic networks} \label{subsectionDiscussionundersectionAN}

In this section, we will define a general mathematical model for the study of networked machines which can share information with each other across their respective network while performing their computations. We want to define it in a general sense in order to allow future variations, to add specificities and to extend the model presented in Section \ref{sectionBusyBeaverGame}, while still being able to formally grasp a mathematical analysis of systemic features like the emergence of information and complexity along with its related phenomena: for example, the expected emergent open-endedness in our case (see discussion \ref{subsectionDiscussiononBBIM...} and Section \ref{sectionAEOE}). 

Since, following this general approach, one can see these mathematical models as a merger of algorithmic (and statistical) information theory and complex networks theoretically combining both distributed computing (or multiagent systems) and game theory, we refer to it as `\emph{algorithmic networks}'. The main idea is that a population of formal theoretical machines can use communication channels over the graph's edges. Thus, the graph topology causes this population to be networked. Once the elements of the population start to exchange information, it forms a overarching model for a system composed of interacting subsystems. So, note that algorithmic networks will be networks of algorithms --- which is the reason of its chosen name ---, should each theoretical machine represent a computable system. Indeed, in the present article we will consider each node as a program of a universal Turing machine (see Definitions \ref{BdefU'} and \ref{BdefN_BB}), which justifies calling either the nodes or the elements of the population of an algorithmic network as \emph{nodes/programs} hereafter.

The term `algorithmic network' is also employed to visually represent systems or processes with circuits of algorithms \cite{Marley1997,Ivanishev1997,Korolev2016} in order to tackle the problem of computationally modeling these processes. Differently, we employ the expression with the purpose of mathematically representing a model for networked populations of computable systems from which one can investigate systemic properties and prove theorems. In this way, a possible disambiguation may be using the expression `\emph{circuit-modeling} algorithmic network' for the former case and `\emph{population-systemic} algorithmic network' for our present approach --- on which, for the sake of simplifying our notation, we will use only the expression `algorithmic network' in this article.
Although the usage in the present article shares the general goal of mathematically representing models for systems through networked algorithms, it may be seen as a generalization\footnote{ If one accepts a formalization of algorithms using theoretical machines (e.g., Turing machines).} of the circuit-modeling approach. Note that, as we will discuss in Section~\ref{subsectionDiscussiononGM}, there may be several different aspects of the graph corresponding to different properties of the population such that a ``circuit of algorithms'' would not grasp. For example, there may be different nodes being the same Turing machine (that is, the population of programs may contain repeated elements), there may be no need of a `delay operator' for multi-step loops, and the graph may not be static. In fact, as we will present in Section \ref{sectionBusyBeaverGame}, these and other features will occur in $\mathfrak{N}_{BB}$.

\subsection{Discussion on graphs and complex networks}

One can have populations with very different properties and several different graphs linking them. Therefore, a general mathematical representation for graphs is paramount --- see Definition \ref{defMAG}. Additionally, we need a way to make abstract aspects of these graphs correspond to properties of the population of theoretical machines, which will be formalized in Definition \ref{BdefAN}. 

Aiming at a wider range of different network configurations, as mentioned in Sections \ref{sectionIntro} and \ref{subsectionDiscussionundersectionAN}, we ground our formalism about graph representations on MultiAspects Graphs (MAG) as presented in \cite{Wehmuth2016b}. In this way, one can mathematically represent abstract aspects that could appear in complex networks. For example, dynamical (or time-varying) networks, multicolored nodes and multilayer networks. Moreover, it facilitates network analysis by showing that their aspects can be isomorphically mapped into a classical directed graph. Thus, MAG abstraction has proven to be crucial to establish connections between the characteristics of the network and the properties of the population composed of theoretical machines (see Definition \ref{BdefFunctionbinAN}).

\subsection{Definition of MultiAspect Graphs}

\begin{Bdefinition}\label{defMAG}
	As in \cite{Wehmuth2017,Wehmuth2016b}, let $ G=(\mathscr{A},\mathscr{E}) $ be a 
	MultiAspect Graph (MAG), where $\mathscr{E}$ is the set of existing composite edges of 
	the 
	MAG and $\mathscr{A}$ is a class of sets (or a space), each of which is an \emph{aspect}. 
	Each 
	aspect $ \mathbf{ \sigma } \in \mathscr{A} $ is a finite set and the number of aspects $ p 
	= | 
	\mathscr{A} | $ is called the \emph{order} of $ G $. By an immediate convention, we call a 
	MAG with only one aspect as a \emph{first order} MAG, a MAG with two aspects as a 
	\emph{second order} MAG and so on. Each composite edge (or arrow) $ e \in \mathscr{E} 
	$ 
	may be denoted by an ordered $2p$-tuple $ ( a_1,\dots,a_p, b_1, \dots, b_p ) $, where $ 
	a_i, b_i 
	$ are elements of the $i$-th aspect with $ 1 \leq i \leq p = | \mathscr{A} | $.

	\begin{Bsubnotation}\label{notationAspectsandEdgesofMAGS}
		$ \mathscr{A}( G ) $ denotes the set of aspects of $ G $ and $ \mathscr{E}( G ) $ denotes the \emph{composite edge set} of $ G $. 
	\end{Bsubnotation}
	
	\begin{Bsubnotation}\label{notationithaspectofaMAG}
		We denote the $i$-th aspect of $ G $ as $ \mathscr{A}( G )[i] $. So, $ | \mathscr{A}( G )[i] |$ denotes the number of elements in $ \mathscr{A}( G )[i] $. In order to match the classical graph case, we adopt the convention of calling the elements of the first aspect of a MAG as \emph{vertices}. Therefore, we denote the set $ \mathscr{A}( G )[1] $ of elements of the first aspect of a MAG $ G $ as $ \mathrm{V}(G) $. Thus, a vertex should not be confused with a composite vertex (see Notation~\ref{notationCompositeverticesandedges}).
	\end{Bsubnotation}
	
	\begin{Bsubnotation}\label{notationCompositeverticesandedges}
		The set of all \emph{composite vertices} $ \mathbf{v} $ of $ G $ is denoted by
		\[
		\mathbb{V}( G ) = \bigtimes_{i=1}^{p} \mathscr{A}( G )[i] 
		\]
		\noindent  and the set of all \emph{composite edges} $ e $ of $ G $ is denoted by
		\[
		\mathbb{E}(G) = \bigtimes_{n=1}^{2p}  \mathscr{A}(G)[ { (n-1)\pmod{p} } + 1 )] \text{ ,}
		\]
		\noindent so that, for every ordered pair $ ( \mathbf{u} , \mathbf{v} ) $ with $ \mathbf{u} , \mathbf{v} \in \mathbb{V}( G )  $, we have $ ( \mathbf{u} , \mathbf{v} ) = e \in \mathbb{E}( G )  $. Also, for every $ e \in \mathbb{E}( G )  $ we have $ ( \mathbf{u} , \mathbf{v} ) = e $ such that $ \mathbf{u} , \mathbf{v} \in \mathbb{V}( G )  $. Thus,
		\[
		\mathscr{E}( G ) \subseteq \mathbb{E}(G)
		\] 
	
	\end{Bsubnotation}	
		
	\begin{Bremarknote}\label{noteNodesand vertices}
			The terms \emph{vertex} and \emph{node} may be employed interchangeably in this article. However, we choose to use the term \emph{node} preferentially in the context of networks, where nodes may realize operations, computations or would have some kind of agency, like in real networks. Thus, we choose to use the term \emph{vertex} preferentially in the mathematical context of graph theory.  
	\end{Bremarknote}

	\begin{Bremarknote}
		Note that $\mathscr{E}$ determines the (dynamic or not) topology of $G$.
	\end{Bremarknote}
	
	\begin{Bremarknote}
		Each aspect in $\mathscr{A}$ determines which variant a graph $G$ will be (and how the set $\mathscr{E}$ will be defined). As in \cite{Costa2015a}, we will deal only with time-varying graphs $G_t$ hereafter, so that there will be only two aspects ($|\mathscr{A}|=2$): the set of vertices (or nodes) $\mathrm{V}(G_t)$ and the set of time instants $\mathrm{T}(G_t)$. An element in $ \mathrm{V}(G_t) \times \mathrm{T}(G_t) $ is a \emph{composite vertex} (or composite node). The dynamic graphs $G_t$ will be better explained in Section \ref{sectionBusyBeaverGame} . 
	\end{Bremarknote}
	
\end{Bdefinition} 

\subsection{Discussion on the general model}\label{subsectionDiscussiononGM}
In a broad sense, one can think of an algorithmic network as a theoretical distributed computing representation model in which each node (or vertex) computes using network's shared information, returning a final output\footnote{ Which is our current case. If the maximum number of node cycles is finite, then every node/program must return a final output. } or not\footnote{ An algorithmic network may have a non limited number of node cycles, so that each node may remain returning an endless number of partial outputs. This may be also the case for circuit-modeling algorithmic networks. However, we do not tackle this problem in the present article.} . The computation of each node may be seen in a combined point of view or taken as individuals. Respectively, nodes/programs may be computing using network's shared information to solve a common purpose --- as the classical approach in distributed computing --- or, for example, nodes may be competing\footnote{ This game-theoretical approach will be discussed further in Section \ref{sectionRelatedandFuture}.} with each other. For the present purposes, we are interested in the average fitness (or payoff), and its related emergent complexity that may arise from a process that increases the average fitness. 

An algorithmic network may have several different configurations, so that the following Definition \ref{BdefAN} will not specify how a particular algorithmic network works. Instead, our formalism enables one to represent every (or most) variation of algorithmic networks with the purpose of modeling a particular problem that may arise from a networked complex system. For example, the networked population may be synchronous\footnote{ See Definition~\ref{BdefSynchronous}.} or asynchronous, have a set of information-sharing protocols\footnote{ See Definition~\ref{BdefIFP}.} or none, a randomly generated population\footnote{ See Definition~\ref{BdefRandompopulation}.} or a fixed one, with communication costs or without them\footnote{ The current model of this article does not consider any cost of communication or of using the communication channels to send information.} etc. In addition, the network topology that determines the communication channels may be dynamical\footnote{ See Definition~\ref{BdefTVG}.} , with weighted edges, multilayer etc. In fact, in accordance with the goal of this article, we will prove theorems on a particular model of algorithmic networks which will be defined in Section \ref{sectionBusyBeaverGame} . However, in this section we limit ourselves to present a general definition of algorithmic networks. Also the reader is invited to note that a circuit-modeling algorithmic network \cite{Marley1997,Ivanishev1997,Korolev2016} is one of these possible different configurations of population-systemic algorithmic network (see Section~\ref{subsectionDiscussionundersectionAN}). 

\subsection{Definitions of general algorithmic networks}

\begin{Bdefinition}\label{BdefAN}
	We define an \emph{algorithmic network} $ \mathfrak{N} = (G, \mathfrak{P}, b)$ upon a population of theoretical machines $\mathfrak{P}$, a MultiAspect Graph $G=(\mathscr{A},\mathscr{E})$ and a function $b$ that causes aspects of $G$ to be mapped\footnote{ See Definition \ref{BdefFunctionbinAN} .} into properties of $\mathfrak{ P }$, so that a vertex in $\mathrm{V}(G) $ corresponds one-to-one to a theoretical machine in $\mathfrak{ P }$. The MAG $G$ was previously defined in Definition~\ref{defMAG}, and we will define $\mathfrak{ P }$ and $b$ in Definitions \ref{BdefPopulation} and \ref{BdefFunctionbinAN}, respectively.  
	
	\begin{Bremarknote}
		In the model studied in this article, we will only deal with second order MultiAspect Graphs (MAGs), so that they represent time-varying graphs. See Definitions~\ref{BdefTVG} and \ref{BdefN_BB}. Thus, we choose to refer to MAGs $ G=(\mathscr{A},\mathscr{E}) $ as just \emph{graphs} in this article. 
	\end{Bremarknote}
	
	\begin{Bsubdefinition}\label{BdefPopulation}
		
		Let the \emph{population} $\mathfrak{P}$ be a ``subset'' of $L$ in which repetitions\footnote{ See also Definitions~\ref{BdefHalt} and~\ref{BdefL_BB}. } are allowed, where $L$ is the language on which the chosen theoretical machine $U$ are running. More formally, a population $\mathfrak{P}$ is a \emph{multiset} that has a set  $ X \subseteq L $  as its support set. Each member of this population may receive inputs and return outputs through communication channels.
		
		\begin{Bsubremarknote}
			One may also interpret a population $ \mathfrak{ P } $ as an ordered sequence\footnote{ In which repetitions are allowed.} $  \left( o_1 , \dots , o_i , \dots , o_{ \left| \mathfrak{P} \right| } \right) $, where $ X $ is the support set of the population and $f_o$ is a labeling surjective function
			\[
			\myfunc{ f_o }{ \mathfrak{ P } = \left( o_1 , \dots , o_i , \dots , o_{ \left| \mathfrak{P} \right| } \right)  }{ X \subseteq L }{ o_i }{ f_o( o_i ) = w }
			\]
			\noindent Therefore, it also justifies an equivalent definition of a population $ \mathfrak{ P } $ as the ordered set of labels
			\[
			\mathfrak{ P } = \{ o_i \mid f_o( o_i ) \in X \subseteq L \}
			\]
			\noindent For example, we may assume a labeling defined by an ordered pair such that
			\[
			o_i = ( i , w )
			\]
			\noindent where $ w = f_o( o_i ) $. This way, by informing $ o_i $ as input, one can promptly retrieve all the information about $ f_o( o_i ) $.
			\begin{Bsubsubdefinition}\label{BdefPopulationasaset}
				For the present purposes, since all populations will be finite, randomly generated (see Definition~\ref{BdefP_BB}), and the ordering of elements in this sequence does not have any impact on our lemmas and theorems, we will assume that our populations are random sequences (see Definition~\ref{BdefRandompopulation}) in the form $ \left( o_1 , \dots , o_i , \dots , o_{ \left| \mathfrak{P} \right| } \right)  $. Therefore, we choose preferentially the interpretation of a population as a set $ \{ o_i \mid f_o( o_i ) \in X \subseteq L \} $ with the labeling $ o_i = ( i , w ) $. However, except for changes in notation and in function's domains (e.g., see Definition~\ref{BnoteFunctionBinAN}), all these three definitions are equivalent in the present article.
			\end{Bsubsubdefinition} 
		\end{Bsubremarknote}
		
	\end{Bsubdefinition}
	
	\begin{Bsubremarknote}
		The choice of $L$ and $U$ determines the class of nodes/systems. For example, one may allow only time-bounded Turing machines in the population. In the present work, $ \mathbf{L_U} $ will be a self-delimiting universal programming language for a simple extended universal Turing machine $\mathbf{U'}$ (see Definition \ref{BdefU'} ) --- i.e., an oracle Turing machine --- that returns zero whenever (and only if) a non-halting computation occur.  
	\end{Bsubremarknote}
	
	\begin{Bsubsubdefinition}\label{BdefCycle}
		A \emph{node cycle} in a population $ \mathfrak{ P } $ is defined as a node/program returning an \emph{output} (which, depending on the language and the theoretical machine the nodes are running on, is equivalent to a node completing a halting computation).
		\begin{enumerate}
			\item If this node cycle is not the last node cycle, then its respective output is called a \emph{partial output}, and this partial output is shared (or not, which depends on whether the population is networked or isolated) with the node's neighbors, accordingly to a specific information-sharing protocol (or not);
			\item If this node cycle is the last one, then this output is called a \emph{final output} such that no more information is shared through the network; 
			\item If every node/program in $ \mathfrak{ P } $ has completed its last node cycle, returning its final outputs, and the population $ \mathfrak{ P } $ is running networked by an algorithmic network $ \mathfrak{N} = (G, \mathfrak{P}, b) $, then we say the algorithmic network $ \mathfrak{N} $ as a whole completed an \emph{algorithmic network cycle}.
		\end{enumerate}
		In addition, let 
		\[
		\mathfrak{C} = \bigcup_{ o_i \in \mathfrak{ P } }  \mathfrak{C}( o_i )
		\]
		be the set of the maximum number of node cycles that any node/program $o_i$ in the population $ \mathfrak{P} $ would be able to perform in order to return a final output, where $\mathfrak{C}(o_i)$ is the set of all node cycles that node/program $o_i$ can perform.
	\end{Bsubsubdefinition}

	\begin{Bsubsubdefinition}\label{BdefSynchronousinthelastcycle}
		In the \emph{last node cycle}, every node only returns its \emph{final output}. 
	\end{Bsubsubdefinition}
	
	\begin{Bsubsubremarknote} 
		So, if the algorithmic network is asynchronous, a node cycle can be seen as an individual communication round, not depending on whether its neighbors are still running or not, whereas, if the network is synchronous, a node cycle can be seen as the usual communication round in synchronous distributed computing. 
		Note that one also has a \emph{graph cycle} in graph theory \cite{Diestel2017,Brandes2005a}, i.e., a path with the same starting and ending vertex. Analogously, in \cite{Wehmuth2016b}, a \emph{MAG cycle} is a path between composite vertices that starts and ends on the same vertex. Thus, algorithmic network cycle, node cycle, and graph (or MAG) cycle are three distinct concepts. 
		In the present article, since the studied algorithmic network model in Definition~\ref{BdefN_BB} is synchronous\footnote{ See Definition~\ref{BdefSynchronous}. }, a node cycle will be equivalent to a communication round (as in synchronous distributed computing) and, for every $ o_i , o_j \in \mathfrak{P} $,  we will have that $  \mathfrak{C}( o_i ) =  \mathfrak{C}( o_j ) $ and that the final output of every node is returned in node cycle $ n = \max\{ \mathfrak{C} \}   \in \mathbb{N} $.
		Moreover, an algorithmic network cycle will be equivalent to a halting distributed (or parallel) computation in $ n = \max\{ \mathfrak{C} \}  $ communication rounds. 
		For the sake of simplicity, in this article, if the word ``cycle'' appears alone, then it refers to a node cycle.
	\end{Bsubsubremarknote}
	
	\begin{Bsubdefinition}
		A \emph{communication channel} between a pair of elements from $\mathfrak{P}$ is defined in $\mathscr{E}$ by an edge (whether directed or not) linking this pair of nodes/programs.
		
		\begin{Bsubremarknote}
			A directed edge (or arrow) determines which node/program sends an output to another node/program that takes this information as input. An undirected edge (or line) may be interpreted as two opposing arrows (i.e., a symmetric adjacency matrix). 
		\end{Bsubremarknote}
		
	\end{Bsubdefinition}

	\begin{Bsubdefinition}\label{BdefFunctionbinAN}
		Let 
		\[ \myfunc{b}{ Y \subseteq \mathscr{A}(G) } { X \subseteq Pr(\mathfrak{P}) } { \mathbf{\overline{a}}  } { b( \mathbf{\overline{a}} ) = \mathbf{\overline{p_r}} } \]
		be a function that \emph{maps} a subspace of aspects $Y$ in $\mathscr{A}$ into a subspace of properties $X$ in the set of properties $Pr(\mathfrak{P})$ of the respective population such that there is an surjective function $ f_{V\mathfrak{ P }} $ such that, for every $ (v,\mathbf{ \overline{x} }) \in Y \subseteq \mathscr{A}(G) $ where $ b( v,\mathbf{ \overline{x} } ) = ( o, b_{ dim( Y ) - 1 }( \mathbf{ \overline{x} } ) ) \in X$, we have that
		\[ \myfunc{f_{V\mathfrak{ P }}}{ \mathrm{V}(G)  } { Support( \mathfrak{P} ) } { v } { f_{V\mathfrak{ P }}(v) = w } \]
		\noindent where $v$ is a vertex (or node), $w$ is an element of $ Support( \mathfrak{P} ) $, and $ Support( \mathfrak{P} ) $ is the support set of the multiset/population $ \mathfrak{P} $.
		
		\begin{Bremarknote}
			If $ Y = \mathscr{A}(G)  $, then we just denote
			\[
			\myfunc{b}{ \mathbb{V}( G ) } { X \subseteq Pr(\mathfrak{P}) } { 
			\mathbf{v}  } { b( \mathbf{v} ) = \mathbf{\overline{p_r}} } 
			\]
		\end{Bremarknote}
		
		\begin{Bsubremarknote}\label{BnoteFunctionBinAN}
			Equivalently, in the case of considering a population as an ordered set of labels from a sequence (see Definition~\ref{BdefPopulationasaset}), we have function $ f_{V\mathfrak{ P }} $ as a bijective function such that, for every $ (v,\mathbf{ \overline{x} }) \in Y \subseteq \mathscr{A}(G) $ where $ b( v,\mathbf{ \overline{x} } ) = ( o_i, b_{ dim( Y ) - 1 }( \mathbf{ \overline{x} } ) ) \in X$,
			\[ \myfunc{f_{V\mathfrak{ P }}}{ \mathrm{V}(G)  } { \mathfrak{P} = \{ o_i \mid f_o( o_i ) = w \in L \} }   { v } { f_{V\mathfrak{ P }}(v) = o_i } \]
			\noindent where $v$ is a vertex (or node) and $o_i$ is an element of sequence $\mathfrak{ P }$.
		\end{Bsubremarknote}

	\end{Bsubdefinition}

\end{Bdefinition}

\begin{Bdefinition}
	We say an element $ o_i \in \mathfrak{P} $ is \emph{networked} \textit{iff} there is $ \mathfrak{N} = (G, \mathfrak{P}, b) $, where $ \mathscr{E}(G) $ is non-empty\footnote{ That is, there must be at least one edge connecting two elements of the algorithmic network. }, such that $o_i$ is running on it.
	
	\begin{Bsubdefinition}
		We say $o_i$ is \emph{isolated} otherwise. That is, it is only functioning as an element of $ \mathfrak{P} $ and not of $ \mathfrak{N} = (G, \mathfrak{P}, b) $.
	\end{Bsubdefinition}
\end{Bdefinition}

\begin{Bdefinition}\label{BdefNetworkinput}
	We say that an input $ w \in L $ is a \emph{network input} \textit{iff} it is the only external source of information every node/program receives and it is given to every node/program before the algorithmic network begins any computation. 
	
	\begin{Bremarknote}
		Note that letter $w$ may also appear across the text as denoting an arbitrary element of a language. It will be specified in the assumptions before it appears or in statement of the definition, lemma, theorem or corollary.
	\end{Bremarknote}
\end{Bdefinition}

\section{The Busy Beaver Imitation Game (BBIG)}\label{sectionBusyBeaverGame}

\subsection{Discussion on the Busy Beaver imitation model, emergence, complexity, optimization, and synchronicity} \label{subsectionDiscussiononBBIM...}

In this section, we will formalize a toy model for comparing the emergence of algorithmic complexity from one of the simplest forms of information sharing that percolates through the network: imitation of the fittest.
As a toy model, our theoretical simple object of investigation must be general and abstract, whereas enabling further variations and extensions with the purpose of studying more properties of networked complex systems for example (see Section \ref{sectionDiameter}). In the present section we choose to give a full discussion on how one characterizes the algorithmic network $ \mathfrak{N}_{BB} (N, f, t, \tau, j) $ (formerly denoted only as $ \mathfrak{N}_{BB} $). This way, the reader might get a full picture of underlying ideas and motivations around this model of algorithmic network that plays the Busy Beaver Imitation Game (BBIG).  

Take a randomly generated set of programs. They are linked, constituting a network which is represented by a graph. Each node/program is trying to return the ``best solution'' it can. And eventually one of these nodes/programs end up being generated carrying beforehand a ``best solution'' for the problem in question. This ``best solution'' is spread through the network by a diffusion process in which each node imitates the fittest neighbor if, and only if, its shared information is ``better'' than what the very node can produce. The question is: how much more complexity can this diffusion process generates on the average compared with the best nodes/programs could do if isolated?

A comparison between the complexity of what a node/program can do when networked and the complexity of the best a node/program can do when isolated will give the \emph{emergent algorithmic complexity} of the algorithmic network. In the present case, the networked ``side of the equation'' relies only on the simple imitation of the fittest neighbor. Since this kind of imitating procedure is one of the simplest or ``worst'' ways to use neighbors' partial outputs to get closer to a best solution\footnote{ Not necessarily given by the randomly generated node/program with the higher fitness, since a node/program could use its neighbor's partial output to calculate a even larger integer for example. Thus, this is the reason we say colloquially that simply imitating the fittest neighbor is one of the ``worst''. Nevertheless, we leave the actual mathematical investigation of what would be the less effective way to get averagely closer to a better solution for future research. }, we are interested in obtaining the emergent algorithmic complexity that arises from a ``worst'' networked case compared with the best isolated case. 

Indeed, a possible interpretation of the diffusion described to the above is \emph{average optimization through diffusion} in a random sampling. Whereas optimization through selection in a random sampling may refer to evolutionary computation or genetic algorithms for example \cite{Fogel2005} (in which the best solution eventually appears and remains sufficiently stable over time), optimization is obtained in our model in a manner that a best solution also eventually appears, but is diffused over time in order to make every individual as averagely closer to the best solution as they can. Therefore, the underlying goal of this process would be to \emph{optimize the average fitness} of the population using the least amount of diffusion time --- we will come back to this in Section \ref{sectionRelatedandFuture}. This type of optimization would particularly be better suited for the cases when adding new nodes/programs (or subparts) is cheaper than adding new cycles (or computational resources).     

But how does one measure complexity of what a node/program computes? We ground our complexity analysis on algorithmic information theory (AIT). As a well-established mathematical field in theoretical computer science and information science, it has proven to be a powerful tool to achieve analytical results, proving sound lemmas and theorems, in order to investigate and build theories on how complexity changes over time. See \cite{Chaitin2012,Chaitin2013,Chaitin2014,Abrahao2015,Abrahao2016,Abrahao2016c, Hernandez-Orozco2016,Hernandez-Orozco2017,Hernandez-Orozco2018}. These works present mathematical results on evolutionary \emph{open-endedness} for computable complex systems. That is, a process in which systems that could be fully simulated on a theoretical Turing machine gain a unlimited amount of complexity over time as random mutations and natural selection\footnote{ Which can be mathematically described by a fitness function and a selection procedure.} apply on them. 

Therefore, as mentioned before, a different form of open-endedness plays the central role in the fundamental characteristics and consequences of the results we will present here (see Section \ref{sectionAEOE} ): the \emph{expected emergent open-endedness}. Instead of asking about how complex systems become over time, as in evolutionary open-endedness, we are focusing another akin question: how complex systems (in fact, systems composed of interacting systems \cite{Shoham2008} \cite{Axelrod1997}) become when the number of its subparts increases? Or, more specifically in our case, how much more emergent complexity arises on the average when the number of networked systems increases? In other words, we are interested in how synergy \cite{Griffith2012,Griffith2014,Quax2017,Bertschinger2014} (see also discussion \ref{sectionRelatedandFuture}) among interacting systems may have an impact on the emergence of complexity. 

In order to tackle this problem we choose to first work under the framework of algorithmic complexity, as in \cite{Chaitin2012,Chaitin2013,Chaitin2014,Abrahao2015,Abrahao2016,Abrahao2016b,Abrahao2016c,Hernandez-Orozco2016}. Furthermore, it will give a direct way to measure fitness of each node/program as explained below. Despite being a theoretical and abstract model for complex systems, it gives the solid foundation and fruitful framework to develop further quasi-isomorphic extensions from resource-bounded --- with a closer application to computer simulations --- and hipercomputable versions, as analogously did in \cite{Abrahao2015,Abrahao2016,Abrahao2016c}. However, in this article the population of nodes/programs will be composed of arbitrary programs running on an arbitrarily chosen universal Turing machine. It makes the Definition~\ref{BdefAlgComp} of algorithmic complexity straightforwardly applicable. But, analogously in \cite{Chaitin2012}, it forces the need of a simple and restrict hipercomputable procedure to deal with eventual non-halting programs (see Definitions \ref{BdefU'} and \ref{Bdefsensitivetooracles}).    

The second toy-modeling property that our algorithmic network $ \mathfrak{N}_{BB} (N, f, t, \tau, j) $ will have is \emph{synchronicity}, which will be defined as a population property in Definition~\ref{BdefSynchronous} . As in synchronous distributed computing, communications are forced to happen at the same time, so that only at the end of each node cycle (see Definition~\ref{BdefCycle}) each node/program is allowed to exchange information by sending partial outputs and receiving partial outputs from its neighbors as input for the next cycle. Indeed, an asynchronous version of our results and how it relates to one presented here are paramount for future research, but it is not in the present scope.

\subsection{Discussion on dynamic graphs and measures of diffusion}\label{subsectionDiscussiononDGMD}

As in Section \ref{sectionAN}, we will start defining topological properties of networks. Following a pursuit of overarching mathematical theorems, we choose to deal with \emph{time-varying} (or dynamical) directed graphs \cite{Pan2011,Guimaraes2013,Costa2015a}. The static case is covered by a particular case of dynamical networks in which topology does not change over time --- see Definition~\ref{BdefStaticNetwork}. And the undirected case can be seen as a graph in which each undirected edge (or line) represents two opposing directed edges (or arrows). Moreover, the dynamical case sets proper theoretical foundations for future research, should one be interested in studying systemic emergent properties of algorithmic networks with a cycle-varying population for example (i.e., an algorithmic network that changes its population size as the number of cycles increases).

Since we aim a formalization of a model for optimizing diffusion with the purpose of investigating emergent complexity, we need a diffusion measure that can be applied on dynamical networks. \emph{Cover time} is a diffusion measure that gives the average time intervals in which a fraction $\tau$ of the nodes is ``infected'' \cite{Guimaraes2013,Costa2015a}. Besides being useful in order to measure diffusion on dynamical networks, it offers some advantages on domain conditions that other measures like the average geodesic distance (or average shortest path length) does not. See Note \ref{BnoteCTconditions1}.
The algorithmic networks that we will define below get their graph topologies from a family of dynamical graphs that has a certain cover time function as a common feature. This function has the domain on population sizes, time instants and fractions $\tau$ of nodes (see Definition \ref{BdefFamilyG}).  

Next we will define the properties of the population composed of nodes/programs. Remember in Section \ref{sectionAN} that both aspects of a graph and properties of a population on a theoretical machine (as well as a function $b_j$ that maps aspects into properties) are necessary in order to properly formalize an algorithmic network. As mentioned in the beginning of the present section, the diffusion process must rely on an imitation of the fittest. So one does not only need to define how fitness\footnote{ Or payoff in a game-theoretical approach.} is measured but also to mathematically state the exact procedure from which each diffusion step is determined.  

\subsection{Definitions on networks and graphs}

\begin{Bdefinition}\label{BdefTVG} 
	As defined in \cite{Costa2015a}, let $ G_t=(\mathrm{V},\mathscr{E},\mathrm{T}) $ be a \emph{time-varying graph} (TVG), where $\mathrm{V}$ is the set of vertices (or nodes), $\mathrm{T}$ is the set of time instants, and $\mathscr{E} \subseteq \mathrm{V} \times \mathrm{T} \times \mathrm{V} \times \mathrm{T}$ is the set of edges\footnote{ That is, the set of existent (second order) composite edges. }.

	\begin{Bsubdefinition}
		We define the set of time instants of the graph $G_t=(\mathrm{V},\mathscr{E},\mathrm{T})$ as $ \mathrm{T}(G_t)=\{t_0, t_1, \dotsc, t_{|\mathrm{T}(G_t)|-1} \} $.
	\end{Bsubdefinition}
	
	\begin{Bsubnotation}
		Let $\mathrm{V}(G_t)$ denote the set of vertices of $G_t$.
	\end{Bsubnotation}
	
	\begin{Bsubnotation}
		Let $|\mathrm{V}(G_t)|$ be the size of the set of vertices in $G_t$.
	\end{Bsubnotation}

	\begin{Bsubnotation}
		Let $ G_t( t ) $ denote the TVG $G_t$ at time instant $ t \in \mathrm{T}(G_t) $, if $G_t$ is a snapshot dynamic network.\footnote{ In any event, our results in this article applies to dynamic networks in general.}
	\end{Bsubnotation}

	\begin{Bremarknote}
		Therefore, a TVG  $G_t=(\mathrm{V},\mathscr{E},\mathrm{T})$ is a particular case of a second order MAG (see Definition~\ref{defMAG} and \cite{Costa2015a,Wehmuth2018}). The first aspect is the set of vertices\footnote{ See Notation~\ref{notationithaspectofaMAG}.  } and the second aspect is the set of time instants.
	\end{Bremarknote}

	\begin{Bsubremarknote}
		For the sake of simplifying our notations in the theorems below, one can take a natural ordering for $ \mathrm{T}(G_t) $ such that
		\[ 
		\forall i \in \mathbb{N} \; \left( \, 0 \leq i \leq | \mathrm{T}(G_t) | - 1 \implies t_i = i + 1 \, \right) 
		\]
	\end{Bsubremarknote}
	
\end{Bdefinition}

\begin{Bdefinition}
	Let $d_t(G_t, t_i, u, \tau)$ be the minimum number of time instants\footnote{ Not taking into account the starting time instant $ t_i $. } (steps, time intervals \cite{Pan2011} or, specially in the present article, cycles) for a diffusion starting on vertex $u$ at time instant $t_i$ to reach a fraction $\tau$ of vertices in the graph $G_t$.
\end{Bdefinition}

\begin{Bdefinition}\label{BdefCovertime}
	As in \cite{Costa2015a}, we define the \emph{cover time} for time-varying graphs as
	
	\begin{equation*}
	CT(G_t,t_i, \tau)=
	\begin{cases}
	\frac{1}{ | \mathrm{V}(G_t) | } \displaystyle\sum_{ \substack {u \in \mathrm{V}(G_t)} } d_t(G_t, t_i, u, \tau) & \quad \text{if fraction } \tau \text{ is reached} ; \\
	\infty & \quad \text{otherwise} ; 
	\end{cases}
	\end{equation*}

	\begin{Bsubnotation}\label{BdefTemporaldiameter}
		Let $D(G_t,t)$ denote the \emph{temporal diffusion diameter} of the graph $G_t$ taking time instant $t$ as the starting time instant of the diffusion process. That is,
		\[
		D(G_t,t) =
		\begin{cases}
		max \{ x \mid \, x=d_t(G_t,t,u,1) \, \land \, u \in \mathrm{V}(G_t) \} \quad\quad  \text{otherwise}; \\
		\infty \quad \quad if \, \,  \exists u \in \mathrm{V}(G_t) \forall x \in \mathbb{N}  \big( x \neq d_t(G_t,t,u,1) \big);
		\end{cases}
		\]
	\end{Bsubnotation}
\end{Bdefinition}

\begin{Bdefinition}\label{BdefStaticNetwork}
	Let $G_s=(\mathrm{V},\mathscr{E}, \mathrm{T})$ be a \emph{static network}, where $G_s$ is a TVG in which, for every $ t_i, t_j, t_k, t_h \in \mathrm{T} $, 
	\[
	\{ ( v_i, v_j ) \mid ( v_i, t_ i , v_j , t_j ) \in \mathscr{E} \}
	=
	\{ ( v_i, v_j ) \mid ( v_i, t_ k , v_j , t_h ) \in \mathscr{E} \}
	\]
	
	\begin{Bremarknote}
		A general way to define a classical static graph is from collapsing all the aspects in $ \mathscr{A} $ into just one aspect (i.e., into the set of vertices/nodes $ \mathrm{V} $) where the set of edges of this MAG is invariant under any relation other than the set of vertices/nodes --- see also sub-determination in~\cite{Wehmuth2016b}. Thus, a static network is a static graph $ G=( \mathrm{V}, \mathrm{E} ) $ for all relations depending only on its set of edges $ E $. 
	\end{Bremarknote}
	
\end{Bdefinition}

\begin{Bdefinition}\label{BdefFamilyG}
	Let
	\begin{align*}
	\mathbb{G}(f,t,\tau) = \{ G_t \mid  & \, i = |\mathrm{V}(G_t)| \,\land \, f(i,t,\tau)=CT(G_t,t,\tau) \, \land \\
	& \land \, \forall i \in \mathbb{N^*} \exists!G_t \in \mathbb{G}( f,t,\tau )  (\, |\mathrm{V}(G_t)|=i \,) \}
	\end{align*}
	\noindent where 
	\[
	\myfunc{f}{ \mathbb{N^*} \times X \subseteq \mathrm{T}(G_t) \times Y \subseteq \; \left]0,1\right] } { \mathbb{N} } { (x,t,\tau) } { y }
	\]

	\noindent be a \emph{family} of unique sized time-varying graphs which shares $f(i,t,\tau)=CT(G_t,t,\tau)$, where $i$ is the number of vertices, as a common property. The finite number of vertices of each graph in this family may vary from $1$ to $\infty$.
	
	\begin{Bremarknote}\label{BnoteWeakeningfamilyG}
		The results in this article can be equally achieved with a weaker assumption within this family: one can instead define a family of graphs upon an arbitrary function $ f $ such that $ f(i,t,\tau)=CT(G_t,t,\tau) $ \emph{only} when the network size $ N $ goes to $ \infty $. 
	\end{Bremarknote}

	\begin{Bremarknote}\label{BnoteCTconditions1}
		Note that a well defined family $ \mathbb{G}(f,t,\tau) $	for every $ t \in \mathrm{T}(G_t) $ 
		and $ \tau \in \left] 0,1 \right] $ implies that $ CT(G_t,t,\tau) \neq \infty $. Hence, in this case, 
		the graphs $ G_t \in \mathbb{G}(f,t,\tau) $ must be strongly temporal-connected for every 
		time instant (and, therefore, also cyclic) --- see \cite{Wehmuth2016b} ---, so that every 
		vertex can reach any other vertex given a sufficient amount of time intervals. In fact, for our 
		main results in Section~\ref{sectionProofs} to hold, one may weaken these conditions on 
		$\tau$ and $t$, because our theorems assume arbitrary values as long as they belong to 
		their respective intervals (or set) and the cover time has a function $f$ well defined for that 
		respective domain. In fact, this is one of the advantages of using the cover time as a 
		diffusion metric for dynamic graphs.
	\end{Bremarknote}

\end{Bdefinition}

\subsection{Discussion on the population playing the BBIG}\label{subsectionDiscussiononPPBBIG}
As in \cite{Chaitin2012,Chaitin2013,Chaitin2014,Abrahao2015,Abrahao2016,Abrahao2016c,Abrahao2016b}, we use the \emph{Busy Beaver function} as our complexity measure of \emph{fitness}. Naming larger integers relates directly to increasing algorithmic complexity \cite{Chaitin2012}, which will allow us to establish crucial probabilistic and statistical properties of a randomly generated population in Lemma \ref{BlemmaSLLNandAIT}. Remember that there is a diffusion of the ``best solutions'' during the cycles. Now the ``best solution'' assumes a formal interpretation of fittest final output (or payoff). The choice of the word ``solution'' for naming larger integers now strictly means a solution for the Busy Beaver problem. Also note that several uncomputable problems are equivalently reduced to the Busy Beaver, including the halting problem. Thus, these mathematical features supports the Busy Beaver function as a sound and meaningful choice for a fitness function for a toy model \cite{Chaitin2013,Abrahao2015,Prokopenko2017}. Not only on a resource-boundless case (e.g., Turing machines) in which finding the best solution might be reducible to a first order uncomputable problem (in the Turing hierarchy) --- which is the case presented in this article ---, but also for more realistic resource-bounded versions of the Busy Beaver, as shown in \cite{Abrahao2015,Abrahao2016,Abrahao2016c,Zenil2016,Delahaye2012}. Such resource-bounded versions may be useful to model optimization problems in which finding the best solution falls under a higher time complexity class\footnote{ Note that time complexity refers to computational complexity in which computation is limited by a certain amount of time. Algorithmic complexity should not be confused with time/space complexity.} and, as we will mention in Section \ref{sectionRelatedandFuture}, we leave for future research.

Thus, with a fixed fitness function that works as a universal\footnote{ Note that, as a universal parameter, this fitness function may not work as a measure of how well adapted a system is in respect to its respective environment or in respect to generate more offsprings. In this manner, it is not in our present scope to discuss the problem of measuring fitness and adaptation in adaptive complex systems.} parameter for every node/program's final (and partial) output it makes sense to have an interpretation of these running algorithmic networks $ \mathfrak{N}_{BB} (N, f, t, \tau, j) $ as playing a \emph{network Busy Beaver game}: during the cycles each node is trying to use the information shared by its neighbors to return the largest integer it can. The larger the final output integer the better the payoff (or fitness). However, only after we define the language and the theoretical machine in \ref{BdefU'} the notions of partial output and algorithmic complexity --- necessary for a fitness measure --- will be well-defined. 

Furthermore, a definition of the \textit{imitation} part of the game is still required. Under what circumstances should a node/program imitate a neighbor? How the best partial outputs are diffused through the network? We will give a formal definition of the \emph{information-sharing protocols} that each node/program must follow when networked in \ref{BdefIFP}. But before we will define the language and universal Turing machine for a population of self-delimiting programs. Doing so, it makes the \emph{algorithmic complexity} --- as mentioned in open-endedness --- well-defined as a direct consequence and, moreover, its correspondent algorithmic probabilities ground our definition of \emph{randomly generated population}.  

We use the term \emph{protocol} as an abstraction of its usage in distributed computing and telecommunications. A protocol is understood as a set of rules or algorithmic procedures that nodes/program must follow at the end of each cycle when communicating. For example, it can be seen as the ``rules for the communications'' under a game-theoretical perspective. So there must be a computable procedure determining what each node/program do when receiving inputs and sending outputs to its neighbors. In fact, as our model of algorithmic network works under a simple imitation of the fittest neighbor, the protocol which must be followed by every node/program determines the exact procedure in doing this imitation. Hence, we call this global information-sharing protocol as \emph{imitation-of-the-fittest protocol} (IFP). We will define these simple algorithmic procedures in \ref{BdefIFP} . The main idea is that each node/program $o_i$ compares its neighbors' partial output (that is, the integer they have calculated in the respective cycle) and runs the program of the neighbor that have output the largest integer if, and only if, this integer is larger than the one that the node/program $o_i$ has output. Since $ \mathfrak{N}_{BB} (N, f, t, \tau, j) $ is playing the Busy Beaver game on a network while limited to simple imitation performed by a randomly generated population of programs, we say it is playing a \emph{Busy Beaver imitation game}.

Finally, we end this section by using these previous definitions  and the definition of the function $b_j$, which ``binds together'' network and programs, in order to define $ \mathfrak{N}_{BB} (N, f, t, \tau, j) $ in \ref{BdefN_BB} . It is an algorithmic network populated by $N$ nodes/programs (that constitutes population $ \mathfrak{P}_{BB} (N) $) such that, after the first (or $c_0$ cycles) cycle, it starts a diffusion\footnote{ There are other diffusions too, since it is possible that two randomly generated neighbors are not close to the node/program with the highest fitness and have different integers as first partial outputs. However, only the one from the biggest partial output is independent of neighbor's partial outputs, so that it disseminates in any situation.} process of the biggest partial output (given at the end of the first cycle) determined by a time-varying graph $G_t$ that belongs to a family of graphs $ \mathbb{G}( f, t, \tau ) $. Then, at the last time instant diffusion stops and one cycle (or more) is spent in order to cause each node to return a final output.

\subsection{Definitions on languages, Turing machines, and algorithmic information theory}

\begin{Bnotation}
	Let $\lg(x)$ denote the binary logarithm $\log_{2}(x)$.
\end{Bnotation}

\begin{Bnotation}
	Let $ |x| $ denote the length of a finite string, if $ x \in  \{ 0 , 1 \}^*  $. In addition, let $ | X | $ denote the number of elements in a set, if $ X $ is a finite set.
\end{Bnotation}

\begin{Bnotation}\label{Review this}
	Let $ (x)_2 $ denote the binary representation of the number $ x \in \mathbb{N} $. In addition, let $ (x)_{L} $ denote the representation of the number $ x \in \mathbb{N} $ in language $ L $.
\end{Bnotation}

\begin{Bnotation}\label{BdefFunctionU}
	Let $ \mathbf{U}(x) $ denote the output of a universal Turing machine $\mathbf{U}$ when $x$ is given as input in its tape. Thus, $ \mathbf{U}(x) $ denotes a \emph{partial recursive} function
	\[
	\myfunc{ \varphi_{\mathbf{U}} }{ L }{ L }{ x }{ y = \varphi_{\mathbf{U}}(x) } \text{ ,}
	\]  
	\noindent where $L$ is a language. In particular, $ \varphi_{\mathbf{U}}(x) $ is a \emph{universal} partial function \cite{Rogers1987,Li1997}. Note that, if $x$ is a non-halting program on $\mathbf{U}$, then this function $\mathbf{U}(x)$ is undefined for $x$.
	
	\begin{BnotationunderBnotation}\label{BdefFunctionphi}
		Wherever number $ n \in \mathbb{N} $ appears in the domain or in the codomain of a partial (or total) function
		\[
		\myfunc{ \varphi_{ \mathcal{U} } }{ L }{ L }{ x }{ y = \varphi_{ \mathcal{U} }(x) } \text{ ,}
		\]
		\noindent where $ \mathcal{U} $ is a Turing machine, or an oracle Turing machine, running on language $L$, it actually denotes
		\[
		\left( n \right)_{ L }
		\]
	\end{BnotationunderBnotation}
\end{Bnotation}

\begin{Bnotation}\label{BdefConcatenation}
	Let $ \mathrm{\textbf{L}}_{\mathbf{U}} $ denote a recursive binary self-delimiting universal programming language for a universal Turing machine $\mathbf{U}$ such that there is a concatenation of strings $w_1, \dots , w_k$ in the language $ \mathrm{\textbf{L}}_{\mathbf{U}} $, which preserves\footnote{ For example, by adding a prefix to the entire concatenated string $ w_1 w_2  \dots  w_k $ that encodes the number of concatenations. Note that each string was already self-delimiting. See also \cite{Abrahao2016,Abrahao2016c}. } the self-delimiting (or prefix-free) property of the resulting string, denoted by 
	\[
	w_1 \circ  \dots  \circ w_k \in \mathrm{\textbf{L}}_{\mathbf{U}}
	\]
	In addition, $  \mathbf{L_U} $ is a complete binary code with
	\[
	\sum\limits_{ p \in \mathbf{L_U} } \frac{1}{ 2^{ | p | } } = 1
	\]
	\begin{BnoteunderBnotation}
		The reader may also note that this self-delimiting-preserving concatenation ``$ \circ $'' is just one example of recursive pairing bijective function $ \left< \cdot \, , \, \cdot \right> $, as in \cite{Li1997,Downey2010}. We know that this pairing function can be extended to $ \left<   \, \cdot \, , \, \left< \, \cdot \, , \, \cdot \, \right> \right> $ and, then, to an ordered triple $ \left< \, \cdot \, , \, \cdot \,  \, , \, \cdot \,\right> $. This way, this procedure can be recursively applied with the purpose of defining finite ordered tuples $ \left< \cdot \, , \, \dots \, , \, \cdot   \right> $. In addition, choosing between two distinct recursive pairing bijective functions $ \left< \cdot \, , \, \cdot \right>_1 $ and $ \left< \cdot \, , \, \cdot \right>_2 $, can only affect the algorithmic complexity\footnote{ See Definition~\ref{BdefAlgComp}. } by
		\[
		A( \left< w_1 \, , \, w_2 \right>_1 ) = A( \left< w_1 \, , \, w_2 \right>_2 ) \pm \mathbf{O}(1)
		\] 
		\noindent Therefore, the reader may equivalently replace\footnote{ Along with the appropriate re-interpretation of what is prefixes or suffixes in language $ \mathbf{L_U} $. }
		\[
		w_1 \circ \dots \circ w_k
		\]
		\noindent with
		\[
		\left< w_1 \, , \, \dots \, , \, w_k \right>
		\]
		\noindent in the present article without affecting the final result. 
	\end{BnoteunderBnotation}

	\begin{BnoteunderBnotation}
		In fact, the self-delimiting-preserving concatenation ``$ \circ $'' may be seen as a ``marked concatenation'' version for self-delimiting binary languages.\footnote{ See \cite{Li1997} for unmarked concatenation. } It ensures some little advantages that were exploited in \cite{Abrahao2016,Abrahao2016c}, like the properties
		\[
		\left|w_i  \right| < \left| w_1 \circ \dots \circ w_k \right| \text{ where $ k \geq i \geq 1 $}
		\]
		\noindent and
		\[
		\left| w_1 \circ \dots \circ w_k \right| \leq  \mathbf{O}( \lg( k ) )  + \left| w_1 \right| + \dots +  \left| w_k \right| \leq  \mathbf{O}(k)  + \left| w_1 \right| + \dots +  \left| w_k \right|
		\]
		\noindent where, for every $ i \in \mathbb{N} $, $ w_i $ is self-delimited. In any case, these are not relevant to the present article.
	\end{BnoteunderBnotation}
\end{Bnotation}

\begin{Bnotation}\label{BdefAlgComp}
	The (prefix) \emph{algorithmic complexity} (Kolmogorov complexity, program-size complexity or Solomonoff-Komogorov-Chaitin complexity) of a string $ w \in \mathrm{\textbf{L}}_\mathbf{U} $, denoted by $A(w)$, is the length of the shortest program $w^* \in \mathrm{\textbf{L}}_\mathbf{U}$ such that $ \mathbf{U}(w^*) = w $.
	
	\begin{BnoteunderBnotation}
		The reader may also find in the literature the prefix algorithmic complexity denoted by $H(w)$ or --- more frequently used --- $K(w)$. As introduced in Section \ref{sectionIntro} and presented in Section \ref{sectionRelatedandFuture}, this work might have several intersections with other fields in future work. Thus, we choose a self-explaining approach on notation in order to avoid ambiguity and notation conflicts in future work, such that we would choose to denote the (prefix) algorithmic complexity/information by $ I_A( w ) $. However, for the sake of simplifying our notation, we choose to denote it only by $ A( w ) $ in the present article.  In fact, $ I_A( w ) $ (i.e., $A(w)$) may be also interpreted as the algorithmic information contained in a object about itself \cite{Li1997}. To this end, note that we have from \cite{Li1997,Downey2010} that, if $ K(w) $ denotes the prefix algorithmic complexity of $w$, $ K( w' | w ) $ denotes the conditional prefix algorithmic complexity of $w'$ given $w$, $ I_K( w : w' ) = K(w') - K( w' | w ) $ denotes the K-complexity of information in $w$ about $w'$, and $ I_A( w ; w' ) = K(w') - K( w' | w^* )  $ denotes the mutual algorithmic information of two objects $w$ and $w'$, then $ I_K( w : w ) = K( w ) - \mathbf{O}(1) $ and $ I_A( w ; w ) = K(w) - \mathbf{O}(1)  $. Thus, the reader may also choose to define $ I_A( w ) $ as an equivalence class in which
		\begin{gather*}
			\left| I_A( w ) -  I_K( w : w ) \right|=\mathbf{ O }(1) \\
			\text{and} \\
			\left| I_A( w ) -  I_A( w ; w ) \right|=\mathbf{ O }(1) \\
		\end{gather*}
		\noindent This way, we will have that $ \left| I_A( w ) -  K(w) \right|=\mathbf{ O }(1)  $ and, therefore, the results of the present article hold anyway.
	\end{BnoteunderBnotation}

\end{Bnotation}


\begin{Bdefinition}\label{BdefU'}
	Let $ \mathbf{L_U} $ be a recursive binary self-delimiting universal programming language $ \mathrm{\textbf{L}}_{\mathbf{U}} $ (as in Notation~\ref{BdefConcatenation}) for a universal Turing machine $\mathbf{U}$, where there is a constant $ \epsilon \in \mathbb{R} $, with $ 0 < \epsilon \leq 1 $, and a constant $  C_{L} \in \mathbb{N} $ such that, for every $ N \in \mathbb{N} $, 
	\[
	A(N) \leq \lg(N) + ( 1+\epsilon )\lg(\lg(N)) + C_{L}
	\] 
	\noindent We define an oracle\footnote{ Or any hypercomputer with a respective Turing degree higher than or equal to $\mathbf{1} $.} Turing machine $\mathbf{U'}$ such that, for every $w \in \mathrm{\textbf{L}}_{\mathbf{U}}$,
	\[ \mathbf{U'}(w) =
	\begin{cases}
	{\mathbf{U}}(w) \text{``} + 1\text{''} & \quad \text{  if  } \mathbf{U} \text{ halts on } w \\
	\text{``} 0 \text{''} & \quad \text{  if  } \mathbf{U} \text{ does not halt on } w \\
	\end{cases}
	\]
	
	\begin{Bremarknote}
		Since $ \mathbf{L_U} $ is self-delimiting and $ {\mathbf{U}}(w) \in \mathbf{L_U} $, we have that the operator $ \text{``}+ 1\text{''} $ actually means the successor operator in an arbitrary recursive enumeration of language $ \mathbf{L_U} $. In the same manner, we have that $ \text{``} 0 \text{''} $ actually means $ \left( 0 \right)_{ \mathbf{L_U} } $.
	\end{Bremarknote}

	\begin{Bremarknote}
		The oracle Turing machine is basically (except for a trivial bijection) the same as the chosen universal Turing machine. The oracle is only triggered to know whether the program halts or not in first place. Also note that $ \mathbf{U'}( w ) $ is a \emph{total} function, and not a partial function\footnote{ See Definition \ref{BdefFunctionU}.} as $ \mathbf{U}(w) $ --- this property will be important in Definition~\ref{BdefEAC}. That is, from Notations~\ref{BdefFunctionU} and~\ref{BdefFunctionphi}, we will have that
		\[
		\myfunc{ \varphi_{\mathbf{U'}} }{ \mathbf{L_U} }{ \mathbf{L_U} }{ x }{ y = \varphi_{\mathbf{U'}}(x) }
		\] 
		\noindent is a \emph{total} (non-recursive or \emph{hypercomputable}) function.
	\end{Bremarknote}
	
	\begin{Bremarknote}
		Note that, from algorithmic information theory (AIT), we know that the algorithmic complexity\footnote{ See Definition \ref{BdefAlgComp}. }$ A(\mathbf{U'}(w)) $ only differs from $ A(\mathbf{U}(w)) $ by a constant, if $\mathbf{U}$ halts on $w$. This constant is always limited by the length of the shortest program that always performs the operation $ \text{``}+1\text{''} $ (or ``subtracts'' $ 1 $, whichever is larger) to any other halting computation. Therefore, both machines belong to an algorithmic complexity equivalence class (the modulus of the subtraction upper bounded by a constant\footnote{ See also the invariance theorem in~\cite{Li1997}. }) everytime $w$ is a halting program. 
		This is the reason why the algorithmic complexity of the final outputs of nodes/programs in $ \mathfrak{N}_{BB} (N, f, t, \tau, j) $ only differ by a constant, should nodes be halting programs during its respective cycles. Also note that, since a non-halting program $ w' $ gives an output always equal to zero when running on machine $ \mathbf{U'} $,  the algorithmic complexity $ A(\mathbf{U'}(w')) $ --- which is defined for machine $ \mathbf{ U } $ in Definition~\ref{BdefAlgComp} --- of the output of $w'$ on $ \mathbf{U'} $ is always equal to a constant (see Lemma~\ref{BlemmaComplexityonBarHalt}). Then, these make Lemma~\ref{BlemmaComplexityp_i} and the Definition~\ref{BdefEAC} sound.
	\end{Bremarknote}

\end{Bdefinition}


\subsection{Definitions on the populations of algorithmic networks}

\begin{Bdefinition}\label{BdefRandompopulation}
	We say a population $ \mathfrak{ P } $ defined on language $ L \subseteq \mathbf{L_U} $ is \emph{randomly generated} \textit{iff} $ \mathfrak{ P } $ is a sequence of events generated by $ | \mathfrak{ P } | $  independent and identically distributed (i.i.d.) random trials accordingly to a \emph{program-size probability distribution} on the set $ L \subseteq \mathbf{L_U} $ such that $ X_{ \mathfrak{ P } } \subseteq L $, where $ X_{ \mathfrak{ P } } $ denotes the set (i.e., language) that is the support set of the multiset (i.e., population) $ \mathfrak{ P } $, and there is a fixed constant $ C \in \mathbb{R}$ and a probability measure $ \mathbf{P} \left[ \cdot \right] $ such that, for every $ p \in L $, 
	\[
	\mathbf{P} \left[ \text{ ``program } p \text{ occur''} \right] = C \frac{1}{2^{ | p | }}
	\]
	
	\begin{Bremarknote}
		Let $ X_{ \mathfrak{ P } } \subseteq \mathbf{L_U} $ denote the set (i.e., language) that is the support of the multiset\footnote{ If one is assuming a population as a sequence, just replace $ X_{ \mathfrak{ P } } $ with $ \mathfrak{ P } = \{ o_i \mid f_o( o_i ) = w \in \mathbf{L_U} $ and replace $p$ with $ f_o( o_i ) $ in $ \mathbf{P} $ and $ \mathbf{P'} $ below. } (i.e., population) $ \mathfrak{ P } $.
		In other words, a randomly generated population $ \mathfrak{ P } $ from language $ L $ is a single outcome of a finite discrete-time i.i.d. stochastic process (i.e., a sequence of i.i.d. random variables), where language $ L $ is the state space from which the random variables take values and 
		\[
		\myfunc{ \mathbf{P} }{ L }{ \left[ 0 , 1 \right] \subset \mathbb{R} }{ p }{ C \frac{1}{2^{ | p | }} }
		\]
		\noindent is the discrete \emph{probability measure}.
		Thus, since language $ \mathbf{L_U} $ is self-delimiting and it is a complete code, we will have that (in general, when $ C \neq 1 $) the function
		\[
		\myfunc{ \mathbf{P'} }{ X_{ \mathfrak{ P } } }{ \left[ 0 , 1 \right] \subset \mathbb{R} }{ p }{ \frac{1}{2^{ | p | }} }
		\]
		\noindent is a discrete \emph{probability semimeasure}.
	\end{Bremarknote}
	
	\begin{Bremarknote}
		The constant $C$ is important for us because it allows us to characterise population $ \mathfrak{P}_{BB}(N) $ in Definition \ref{BdefP_BB} as randomly generated, taking into account that there are information-sharing protocols that could not be previously determined. 
		Otherwise, if one is not considering any information-sharing protocol, one could just obtain a discrete probability measure from
		\[
		\myfunc{ \mathbf{P} }{ L= \mathbf{L_U} }{ \left[ 0 , 1 \right] \subset \mathbb{R} }{ p }{ \frac{1}{2^{ | p | }} } \text{ ,}
		\]
		\noindent which is equivalent to a classical \emph{Lebesgue measure} in algorithmic information theory \cite{Downey2010,Calude2002,Li1997}.
		\noindent However, our forthcoming proofs stems from the idea that only the suffixes---note that the elements of the population $ \mathfrak{P}_{BB}(N) $ are in the form ``$ P_{prot} \circ p $''---were randomly generated, and that the global information-sharing protocol in Definition~\ref{BdefIFP} was previously fixed (i.e., determined) as an assumption in our Lemmas, Theorems and Corollaries. 
		Hence, this constant $C$ is not taken into account in the present work. 
		Therefore, since language $ \mathbf{L_U} $ is a complete binary code, we assume
		\[
		\myfunc{ \mathbf{P} }{  \mathbf{L}_{BB}  = \{  P_{prot} \circ p \mid p \in \mathbf{L_U} \} \subseteq \mathbf{L_U} }{ \left[ 0 , 1 \right] \subset \mathbb{R} }{ P_{prot} \circ p }{ \frac{1}{2^{ | p | }} }
		\]
		\noindent which is a discrete probability measure.
	\end{Bremarknote}
	
	\begin{Bremarknote}
		Nevertheless, even if language $ \mathbf{L_U} $ is \emph{not} a complete code such that strict inequality holds in Definition~\ref{BdefConcatenation} where
		\[
		\sum\limits_{ p \in \mathbf{L_U} } \frac{1}{ 2^{ | p | } } < 1 \text{ ,}
		\]
		\noindent our main results also hold for the probability measure:
		\[
		\myfunc{ \mathbf{P} }{  \mathbf{L}_{BB} = \{  P_{prot} \circ p \mid p \in \mathbf{L_U} \} \subseteq \mathbf{L_U} }{ \left[ 0 , 1 \right] \subset \mathbb{R} }{ P_{prot} \circ p }{ C \, \frac{1}{2^{ | p | }} } \text{ .}
		\]
		To this end, the reader is invited to note that this constant $C$ would only affect Lemma~\ref{BlemmaSLLNandAIT} as an additive constant and the other lemmas, theorems and corollaries as a multiplicative constant wherever $ \Omega( w, c(x) ) $ already appears as multiplicative constant. Therefore, since we are investigating asymptotic behaviors as the population size grows toward infinity, our final results hold in the case $ C \neq 1 $.
	\end{Bremarknote}

\end{Bdefinition}

\begin{Bdefinition}\label{Bdefsensitivetooracles}
	We say a population $\mathfrak{P}$ is \emph{sensitive to oracles} \textit{iff} whenever an oracle is triggered during any cycle in order to return a partial output the final output of the respective node/program is also\footnote{ Since we are assuming $0$ as the assigned non-halting output for $\mathbf{U'}$ in relation to the machine $\mathbf{U}$.} $0$. That is, more formally, a population $\mathfrak{P}$ is \emph{sensitive to oracles} \textit{iff}:
	
	\begin{enumerate}[label=\upshape(\Roman*),ref=\theBdefinition (\Roman*)]
		 \item Let $ p_{net_{ \mathbf{U} } } $ be a program such that $ p_{net_{ \mathbf{U} } } \circ o_i \circ c $ computes on machine $ \mathbf{U'} $ cycle-by-cycle what a node/program $o_i \in \mathfrak{P} $ does on machine $ \mathbf{U} $ until cycle $c$ when networked. Let $ p_{iso_{ \mathbf{U} } } $ be a program such that $ p_{iso_{ \mathbf{U} } } \circ o_i \circ c $ computes on machine $ \mathbf{U'} $ cycle-by-cycle what a node/program $o_i \in \mathfrak{P} $ does on machine $ \mathbf{U} $ until cycle $c$ when isolated. Let $p_{o_i,c}$ be the \emph{partial output} sent by node/program $o_i$ at the end of cycle $c$. Also, $p_{o_i,max \{ c \mid c \in \mathfrak{C} \}}$ denotes the final output of the node/program $o_i$. Then, for every $ o_i \in \mathfrak{P} $, if there is $ c \in \mathfrak{C} $ such that $ \mathbf{U'}\left( p_{net_{ \mathbf{U} } } \circ o_i \circ c \right) = 0 $ (or $ \mathbf{U'}\left( p_{iso_{ \mathbf{U} } } \circ o_i \circ c  \right) = 0 $), then the networked (or, respectively, the isolated) final output $ p_{o_i,max \{ c \mid c \in \mathfrak{C} \}} = 0 $. \label{itemSensitivetooracles:1}  
	\end{enumerate}

	\begin{Bremarknote}\label{BnoteMachineU''}
		Another alternative is to embed the procedure that guarantees the property of being sensitive to oracles in an underlying oracle machine $ \mathbf{U''} $ that is simulating the entire algorithmic network $ \mathfrak{N}_{BB} (N, f, t, \tau, j)   $. This way, the property of being sensitive to oracles may be also understood as the extension of machine $ \mathbf{U'} $. However, in this article, since we are studying an average phenomenon in the population $ \mathfrak{P}_{BB} (N)  $ (whether networked or isolated), simulating the entire algorithmic network $ \mathfrak{N}_{BB} (N, f, t, \tau, j)   $ is not of relevancy to our results. Therefore, we choose to omit the definition of machine $ \mathbf{U''} $, which would follow directly from the procedure described in Definition~\ref{itemSensitivetooracles:1}).
	\end{Bremarknote}

	\begin{Bremarknote}
		Here, a straightforward interpretation is that nodes that eventually do not halt in a cycle are ``killed'', so that their final outputs have the ``worst'' fitness---see also Section~\ref{subsectionDiscussiononPPBBIG}. 
		Thus, these nodes are programs that need to be ultimately run on an (first-order) oracle Turing machine $ \mathbf{U'} $.
		This requirement is also analogous to the one in \cite{Hernandez-Orozco2018,Chaitin2012,Chaitin2013,Chaitin2014}, which deal with a sole program at the time under an evolutionary perspective---so, note that those are different from the studied model in this article, which deals with population of programs under a distributed-system non-evolutionary emergent perspective (see also Section~\ref{subsectionOE}). 
		
		However, the oracle is only necessary to deal with the non-halting computations.  That is, $ \mathbf{U'} $ in Definition~\ref{BdefU'} behaves like an universal Turing machine $ \mathbf{U} $ except that it returns zero (which was the chosen fixed output label for non-halting computation) whenever a non-halting computation occur. 
		From algorithmic information theory (AIT), we know that the algorithmic complexity $ A(\mathbf{U'}(w)) $ only differs from $ A(\mathbf{U}(w)) $ by a constant, if $\mathbf{U}$ halts on $w$. 
			This is because, from Definition~\ref{BdefU'}, this constant is always limited by the length of the shortest program that always performs the operation $ \text{``}+1\text{''} $ (or ``subtracts'' $ 1$, whichever is larger) to any other halting computation. 
			As a consequence, both machines belong to an algorithmic complexity equivalence class\footnote{ See also the invariance theorem in~\cite{Li1997}. } determined by
			\[
			\left| A(\mathbf{U'}(w)) - A(\mathbf{U}(w))  \right| = \mathbf{O}(1)
			\]
			\noindent everytime $w$ is a halting program. 
			This is the reason why the algorithmic complexity of the final outputs of nodes/programs in $ \mathfrak{N}_{BB} (N, f, t, \tau, j) $ only differ by a constant, should nodes be halting programs during its respective cycles. 
			On the other hand, since a non-halting program $ w' \in \mathbf{L_U} $ gives an output always equal to zero when running on machine $ \mathbf{U'} $,  the algorithmic complexity $ A(\mathbf{U'}(w')) $---which is defined for machine $ \mathbf{ U } $ in Definition~\ref{BdefAlgComp}---of the output of $w'$ on $ \mathbf{U'} $ is always equal to a constant (see Lemma~\ref{BlemmaComplexityonBarHalt}). 
		
		Therefore, \emph{oracle-sensitiveness} is not only sound within the context of populations of Turing machines with an associated fitness function which strictly depends on their final outputs, but also preserves an algorithmic complexity measure for nodes/programs' behaviors through an equivalence class of algorithmic complexities in the form:
		\[
		\mathbf{O}(1)
		=
		\begin{cases}
		\left| A(\mathbf{U'}(w)) - A(\mathbf{U}(w))  \right|  \quad \text{if $\mathbf{U}$ halts on $w$ } \\
		\left| A(\mathbf{U'}(w)) - A(\text{``label for non-halting computation''})  \right|  \quad \text{otherwise}
		\end{cases}
		\] 
		This way, Lemmas~\ref{BlemmaComplexityp_i} and~\ref{BlemmaComplexityonBarHalt} and Definition~\ref{BdefEAC} also become sound, and our results in this article hold for taking both machines $\mathbf{U}$ or $\mathbf{U'}$ as references for a complexity measure.
	\end{Bremarknote}

\end{Bdefinition}

\begin{Bdefinition}\label{BdefSynchronous}
	In a \emph{synchronous} population of a dynamic algorithmic network each node is 
	\emph{only} 
	allowed to receive inputs from its incoming neighbors and to send information to its 
	outgoing neighbors at the end of each node cycle (or communication round). Each node 
	cycle always begins and ends at the same time even if the computation time of the 
	nodes/programs is arbitrarily different. More formally:

\begin{enumerate}[label=\upshape(\Roman*),ref=\theBdefinition (\Roman*)]
	\item If the population $ \mathfrak{P} $ is \emph{networked} in an algorithmic network $ 
	\mathfrak{N} = \left( G_t, \mathfrak{ P }, b \right) $, then there is a \emph{partial} 
	function 
	$f$ such that, for every $ c \in \mathfrak{C} $, there is a constant $ t \in \mathrm{T}(G_t) 
	$ such that, for every $ o_i \in \mathfrak{P} $ with $ c(o_i) = c $,
	\[ 
	b( v , t, \mathbf{ \overline{ x } } ) = ( o_i , c , b_{ dim(Y) - 2 }( \mathbf{ \overline{ x } } ) )  
	\iff 
	\myfunc{f}{ \mathfrak{C}(o_i) } { 
		\mathrm{T}(G_t) } { c } { f(c)=t }  
	\text{ ,}
	\]
	\noindent where $\mathfrak{C}(o_i)$ is the set of node cycles of node $o_i$ as in 
	Definition~\ref{BdefCycle}.\footnote{ In other words, function $b$ is injective with 
	respect 
		to the second term $t$ for which there is $c$ with $ b( v , t, \mathbf{ \overline{ x } } ) = ( 
		o_i , c , b_{ dim(Y) - 2 }( \mathbf{ \overline{ x } } ) )   $, and this value 
		of $c$ does not depend on the choice of $v$. } \label{itemBdefSynchronous:1}
	
	\item If the population $ \mathfrak{ P } $ is \emph{isolated}, the same property in 
	Definition~\ref{itemBdefSynchronous:1} applies, except for taking another algorithmic 
	network $ \mathfrak{N'} = \left( G'_t, \mathfrak{ P }, b \right) $ in which $ \mathscr{E}( 
	G'_t ) $ only contains one-step vertex self-loops (see \cite{Wehmuth2017}).\footnote{ 
	Note that, in this case, function $b$ is irrelevant. } \label{itemBdefSynchronous:2}
\end{enumerate}
	
	\begin{Bremarknote}
		If the population is \emph{isolated}, then only each partial output is taken into account in the respective self-loop.
	\end{Bremarknote}

	\begin{Bremarknote}
		 Note that the procedure responsible for performing the synchronization may be abstract-hypothetical or defined, as in Note~\ref{BnoteMachineU''}, on an underlying oracle machine $ \mathbf{U''} $ that makes each individual node cycle start at the same time (or after every node/program returns its partial output in the respective cycle). Since machine $ \mathbf{U'} $ computes a total function, the reader is also invited to note that any algorithm of synchronization in distributed computing \cite{Barbosa1996,Kshemkalyani2008,Merideth2007a} could be embedded as subroutine into this machine $ \mathbf{U''} $. However, as in Note~\ref{BnoteMachineU''}, since we are studying an average phenomenon in the population $ \mathfrak{P}_{BB} (N)  $ (whether networked or isolated), simulating the entire algorithmic network $ \mathfrak{N}_{BB} (N, f, t, \tau, j)   $ is not of relevancy to our results. Therefore, we choose to omit the definition of machine $ \mathbf{U''} $.  
	\end{Bremarknote}
	
\end{Bdefinition}


\begin{Bdefinition}\label{BdefIFP}
	We say a population $ \mathfrak{P} $ follows an \emph{imitation-of-the-fittest protocol} (IFP) \textit{iff} each node/program always obeys protocols defined in \ref{BdefBBcontagion} , \ref{BdefMaxCoop} and \ref{BdefITFOprot} \emph{when networked}. Or more formally: 
	
	\begin{enumerate}[label=\upshape(\Roman*),ref=\theBdefinition (\Roman*)] 
		\item \label{procedureIFP} Let $\mathbf{X}_{neighbors}(o_j,c)$ be the set of incoming neighbors of node/program $o_j$ that have sent partial outputs to it at the end of the cycle $c$.
		Let $ \{ p_{o_i,c} \mid o_i \in \mathbf{X}_{neighbors}(o_j, c) \land i \in \mathbb{N} \land c \in \mathfrak{C} \} $ be the set of partial outputs relative to $\mathbf{X}_{neighbors}(o_j,c)$.
		Let $ w $ be the network input (as in Definition \ref{BdefNetworkinput}). Let $\circ$ denote a recursively determined concatenation of finite strings as in Definition~\ref{BdefConcatenation}.
		Then, for every $ o_j,o_i \in \mathfrak{ P } $ and $ c,c-1 \in \mathfrak{C} $,
		\begin{enumerate}
			\item  if $ max \{ c \mid c \in \mathfrak{C} \} = 1 $, then
			\[
			p_{o_j,c} = \mathbf{U'}( o_j \circ w )
			\]
			
			\item if $c=1$ and $ c \neq max \{ c \mid c \in \mathfrak{C} \} $, then
			\[
			p_{o_j,c} = w \circ o_j \circ \mathbf{U'}( o_j \circ w )
			\]
			
			\item \label{clauseIFPmain} if $ c \neq 1 $ and $ c \neq max \{ c \mid c \in \mathfrak{C} \} $, then 
			
			\noindent $ p_{o_j,c} = w \circ o_i \circ max \{ x \mid p_{o_j,c-1}= w \circ o_i \circ x \, \lor \, w \circ o_i \circ x \in  \{ p_{o_i,c-1} \mid o_i \in \mathbf{X}_{neighbors}(o_j, c-1) \land i \in \mathbb{N} \land c-1 \in \mathfrak{C} \}   \} $
			
			\item if $ c = max \{ c \mid c \in \mathfrak{C} \} $ and $ p_{o_j,c-1} = w \circ o_i \circ x $, then
			\begin{center}
				\noindent $ p_{o_j,c} = x $ 
			\end{center}
		\end{enumerate}
	\end{enumerate}
\end{Bdefinition}

\begin{Bremarknote}
	Since we will be dealing only with synchronous algorithmic networks in this article (see Section~\ref{sectionBusyBeaverGame}), these global sharing protocols will apply at the end of each cycle (or communication round) --- see Definition \ref{BdefCycle} . So, after the first cycle the diffusion of the biggest partial output will work like a \emph{spreading} in time-varying networks \cite{Guimaraes2013,Costa2015a}. And the last cycle (or more cycles\footnote{ See Definition \ref{BdefN_BB}. }) is spent in order to cause each node/program to return a number --- from which we measure the complexity of the respective node/program as discussed in Section~\ref{subsectionDiscussiononBBIM...}.	
\end{Bremarknote} 

\begin{Bremarknote}
	In order to simplify our notation, we let $ w \circ \mathbf{U'}(x ) $ denote the prefix preserving concatenation $\circ$ (see Notation \ref{BdefConcatenation}) of the string $ w \in \mathbf{L_U} $ with the string $ y \in \mathbf{L_U} $ such that $y$ represents the number $\mathbf{U'}(x)$ in the language $\mathbf{L_U}$.
\end{Bremarknote}

\begin{Bremarknote}
	Note that, if one enforces that the number of cycles needs to be informed to the this global information-sharing protocol, then the expected algorithmic complexity of a node/program in the networked population will be even larger in respect to the expected algorithmic complexity in the isolated population in Corollary~\ref{BcorMain}. That is, the additional input of the number of cycles in the networked case may ``cancel'' the one in the isolated case. Therefore, our final results on the lower bound for the expected emergent algorithmic complexity of a node (EEAC) can even be increased. This is expected to happen for example in the case it was possible to simulate the entire algorithmic network. In this way, it will be important for migrating our results to resource-bounded algorithmic networks for future research.\footnote{ As done in~\cite{Abrahao2015,Abrahao2016,Abrahao2016c} for metabiology~\cite{Chaitin2012,Chaitin2014,Chaitin2013}. } 
\end{Bremarknote}

\begin{Bsubdefinition}\label{BdefBBcontagion}
	We call a \emph{Busy Beaver contagion protocol} as a global information-sharing protocol in which every node/program runs the node/program of the neighbor that have output --- a partial output --- the largest integer instead of its own program \textit{iff} the partial output of this neighbor is bigger than the receiver's own partial output.

\end{Bsubdefinition}

\begin{Bsubremarknote}
	Note that a node/program only needs to take into account the biggest partial output that any of its neighbors have sent. If more than one sends the largest integer as partial output, the receiver node/program choose one of these respective neighbors accordingly to an arbitrary rule. Then, this partial output is the one that will be compared to the partial output from the receiver node/program.
\end{Bsubremarknote}

\begin{Bsubdefinition}\label{BdefMaxCoop}
	In a \emph{maximally cooperative protocol} every node/program shares its own program and its latest partial output with all its neighbors at the end of each cycle and before the next cycle begins.
	
	\begin{Bsubremarknote}
		In the model defined in \ref{BdefP_BB}, sharing only the last partial output turns out to be equivalent\footnote{ However, at the expense of using more computation time.} to the definition of maximally cooperative protocol.	
	\end{Bsubremarknote} 
	
	\begin{Bsubremarknote}
		Remember that $ \mathfrak{N}_{BB} (N, f, t, \tau, j)  $ only lets its nodes/programs perform computation in the first cycle (see Definition \ref{BdefITFOprot}). This is the reason why only the network input $w$ matters in its respective maximally cooperative protocol.
	\end{Bsubremarknote}
	
	\begin{Bsubremarknote}
		As in \cite{Costa2015a}, this diffusion process may be interpreted as following a Breadth-First Search (BFS), in which each node starts a diffusion by sending the specified information in Definition \ref{BdefIFP} to all of its adjacent
		nodes. Then, these adjacent nodes relay information for their own adjacent nodes in the next time instant, and so on.
		\begin{Bsubremarknote2}
			Typically, one can set a collection of time instants in which the diffusion of information is limited to 1 step. However, the presented final results in Theorems \ref{BthmMain} and  \ref{BthmMainCentralTime} and Corollary \ref{BcorMain} also hold if more than one step per time instant of diffusion (in our case, the application of the global sharing protocols) is allowed at the end of each cycle. In fact, it can only make $ { \tau_{\mathbf{E}(max)}( N,f,t_i,\tau ) }|_{t}^{t'} $ larger --- see Lemma \ref{BlemmaMinComplexityonDiffusion}.
		\end{Bsubremarknote2}
	\end{Bsubremarknote}
	
\end{Bsubdefinition}

\begin{Bsubdefinition}\label{BdefITFOprot}
	We call a \emph{contagion-only protocol} as a global information-sharing protocol in which every node/program only plays the Busy Beaver contagion and does not perform any other computation after the first cycle \emph{when networked} in some $ \mathfrak{N} = (G, \mathfrak{P}, b)$.

	\begin{Bsubremarknote}
		This is the condition that allows us to investigate the ``worst'' (see discussion \ref{subsectionDiscussiononBBIM...}) case in which no node/program spends computational resources other than playing its global sharing protocols.\footnote{ In our case, generating expected emergent algorithmic complexity of a node.} In other words, it forces the algorithmic network to rely on a diffusion process only.
	\end{Bsubremarknote} 
	
	\begin{Bsubremarknote}
		In the main model presented in this article, the first time instant occurs after the first cycle --- see \ref{BdefN_BB} . 
	\end{Bsubremarknote}
\end{Bsubdefinition}
\noindent \\



\begin{Bdefinition}\label{BdefL_BB}
	Let $ {\textbf{L}}_{BB} \subset \mathrm{\textbf{L}}_{\mathbf{U}} $  be a language of programs 
	in the form $P_{prot} \circ p$ where $p \in \textbf{L}_{\textbf{U}}$. The prefix $P_{prot}$ is any 
	program that always ensures that, if the node/program $ P_{prot} \circ p $ \emph{is 
	networked} and running on $ \mathbf{U'} $, then $P_{prot} \circ p$ obeys the 
	\emph{imitation-of-the-fittest protocol}\footnote{ See Definition~\ref{BdefIFP} in 
	Section~\ref{sectionBusyBeaverGame} }. Otherwise, if  the node/program $ P_{prot} \circ p $ is 
	\emph{isolated} and running on $ \mathbf{U'} $, then, for every $ w \in \mathbf{L_U} $, 
	$\mathbf{U'}(P_{prot} \circ p \circ w ) = \mathbf{U'}(p \circ w)$ and every subsequent node cycle 
	works like a reiteration of partial outputs as immediate next input for the same program $p$. 
	\footnote{ Thus, the isolated case may be represented (and is equivalent to) by an algorithmic 
	network built on a population in language $ \mathbf{L_U} $ that does not follow any 
	information-sharing protocol and the topology of the MultiAspect Graph (MAG) is composed 
	by one-step self-loops on each node/program only.} 
\end{Bdefinition}

\begin{Bdefinition}\label{BdefP_BB}
	Let $\mathfrak{P}_{BB}(N) \text{``} \subseteq \text{''} \mathbf{L}_{BB} $ be a randomly generated\footnote{ As in Definition~\ref{BdefRandompopulation}.} population of $N$ elements from\footnote{ That is, the support set of $ \mathfrak{P}_{BB}(N)  $ is contained in $ \mathbf{L}_{BB} $.} $ \mathbf{L}_{BB} $ that is \emph{synchronous}\footnote{ The procedure responsible for performing the synchronization may be abstract-hypothetical or defined on an underlying oracle machine that makes each individual cycle start at the same time (or after every node/program returns its partial output in the respective cycle). See Definition~\ref{BdefSynchronous} and Note~\ref{BnoteMachineU''}. }, \emph{sensitive to oracles}\footnote{ See Definition~\ref{Bdefsensitivetooracles}. } and with \emph{randomly generated}\footnote{ Note that we are dealing with self-delimiting languages, so that one can always define probability measures. See Definition~\ref{BdefRandompopulation}. } $ p \in \mathbf{L_U} $ such that\footnote{ Hence, $ \mathbf{L}_{ \mathfrak{P}_{BB}(N) } $ is defined as the population of suffix nodes/programs that were randomly generated in order to constitute $ \mathfrak{P}_{BB}(N) $. Therefore, $ \mathbf{L}_{ \mathfrak{P}_{BB}(N) } $ is a population and not a language (see note \ref{notePopulationnotlanguage} and Definition~\ref{BdefPopulation}), so that there may be repetitions within $ \mathbf{L}_{ \mathfrak{P}_{BB}(N) } $. It is important to note this since the letter $L$ is used to denote languages in other parts in this paper.} 
	
	\[
	p \in \mathbf{L}_{ \mathfrak{P}_{BB}(N) }
	\]
	\[
	iff 
	\]
	\[
	P_{prot} \circ p \in \mathfrak{P}_{BB}(N) 
	\]
	
	\noindent where $ \mathbf{L}_{ \mathfrak{P}_{BB}(N) } $ is a population (of suffixes\footnote{ See Definition~\ref{BdefConcatenation}. } ).
	\begin{Bremarknote}
		Note that all conditions and protocols in Definition~\ref{BdefP_BB} define the set of properties $Pr(\mathfrak{P}_{BB}(N))$ of the population $\mathfrak{P}_{BB}(N)$.
	\end{Bremarknote}
	
	\begin{Bremarknote}\label{notePopulationnotlanguage}
		There may be a misplaced usage of the operator ``$ \subseteq $'' in $ \mathfrak{P}_{BB}(N) \subseteq \mathrm{\textbf{L}}_{BB} $ here. Since $ \mathfrak{P}_{BB}(N) $ is a population, it may contain repeated elements of $ \textbf{L}_{BB} $. However, for the sake of simplicity, we assume a population as an ordered set of labels defined upon a sequence (see Definition~\ref{BdefPopulationasaset} and Note~\ref{BnoteFunctionBinAN}) and we say a population $ \mathfrak{ P } $ is contained in a language $ L $ \textit{iff} \[  \forall o_i, 1 \leq i \leq | \mathfrak{ P } | \, \big( \, o_i \in \mathfrak{ P }  \implies f_o( o_i ) = w \in L \, \big) \]
		\noindent Otherwise, in the interpretation of a population as a multiset, if $ Support( \mathfrak{ P } ) $ denotes the support set of the multiset/population $ \mathfrak{ P } $, then we let $ \mathfrak{P} \subseteq L $ in fact denote $ Support( \mathfrak{ P } ) \subseteq L  $.
	\end{Bremarknote}
	
\end{Bdefinition}

\begin{Bdefinition}\label{BdefL_U}
	For the sake of simplifying our notation, we denote the language\footnote{ Note that, since it is a language and not a population, no repetitions are allowed in $\mathbf{L_U}(N) $. } of the size-ordered\footnote{ That is, an ordering from the smallest to the largest size. Also note that ordering members with the same size may follow a chosen arbitrary rule (e.g., a lexicographical one).} shortest
	\[ p \in  \mathbf{L}_{ \mathfrak{P}_{BB}(N) }	\]
	\noindent as
	\[ 
	\mathbf{L_U}(N) 
	\]
	
	\begin{Bremarknote}
		Note that 
		\[
		\left| \mathbf{L_U}(N)  \right| \leq \left|  \mathbf{L}_{ \mathfrak{P}_{BB}(N) } \right| = N \in \mathbb{N}
		\]
		\noindent This will be important in Lemma~\ref{BlemmaGibbsandalgorithmicentropy}.
	\end{Bremarknote}
\end{Bdefinition}	

\noindent \\


\subsection{Definitions on the studied algorithmic network model}


\begin{Bdefinition}\label{BdefN_BB}
	Let
	\[
	\mathfrak{N}_{BB} (N, f, t, \tau, j)=(G_t, \mathfrak{P}_{BB} (N),b_j)
	\]
	\noindent be an algorithmic network where $f$ is an arbitrary well-defined function such that 
	\[
	\myfunc{f}{ \mathbb{N^*} \times X \subseteq \mathrm{T}(G_t) \times Y \subseteq ]0,1] } { \mathbb{N} } { (x,t,\tau) } { y }
	\]
	\noindent and $ G_t \in \mathbb{G}(f, t, \tau)$, $ | \mathrm{V}(G_t) | = N $, $ | \mathrm{T}(G_t) | > 0 $, and there are arbitrarily chosen\footnote{ Since they are arbitrarily chosen, one may choose to take them as minimum as possible in order to minimize the number of cycles for example. That is, $ c_0 = 0 $ and $ n = | \mathrm{T}(G_t) | + 1 $ for example. } $ c_0, n \in \mathbb{N} $ where $ c_0 + | \mathrm{T}(G_t) | + 1 \leq n \in \mathbb{N} $ such that $b_j$ is a function
	\[
	\myfunc{b_j} { \mathrm{V}(G_t) \times \mathrm{T}(G_t) } { Support\left( \mathfrak{P}_{BB}(N) \right) \times \mathbb{N}|_1^{ n }} { (v,t_{c-1}) } { b_j(v,t_{c-1})=( p, c_0 + c ) }  
	\]
	\noindent such that\footnote{ See Definition \ref{BdefFunctionbinAN} .} , since one has fixed the values of $c_0$ and $n$,
	\[
	\left| \left\{ 
	b_j \, \Bigg| \, 
	\myfunc{b_j} { \mathrm{V}(G_t) \times \mathrm{T}(G_t) } { Support\left( \mathfrak{P}_{BB}(N) \right) \times \mathbb{N}|_1^{ n }} { (v,t_{c-1}) } { b_j(v,t_{c-1})=( p, c_0 + c ) } 
	\right\} \right| \leq N^{ N }
	\] 
	
	\begin{Bremarknote}
		In the case population $ \mathfrak{P}_{BB}(N) $ is an ordered set of labels from a sequence, we have function $b_j$ as an injective function, where
		\[
		\myfunc{b_j} { \mathrm{V}(G_t) \times \mathrm{T}(G_t) } { \mathfrak{P}_{BB}(N) \times \mathbb{N}|_1^{ n }} { (v,t_{c-1}) } { b_j(v,t_{c-1})=( o_i , c_0 + c ) }  
		\]
		\noindent such that\footnote{ See Definition \ref{BdefFunctionbinAN} .} , since one has fixed the values of $c_0$ and $n$,
		\[
		\left| \left\{ 
		b_j \, \Bigg| \, 
		\myfunc{b_j} { \mathrm{V}(G_t) \times \mathrm{T}(G_t) } { \mathfrak{P}_{BB}(N)  \times \mathbb{N}|_1^{ n }} { (v,t_{c-1}) } { b_j(v,t_{c-1})=( o_i , c_0 + c ) } 
		\right\} \right| \leq N^{ N }
		\] 
	\end{Bremarknote}
	
	\begin{Bremarknote}
		Since the way time instants are mapped into cycles is fixed given values of $c_0$ and $n$, we may equivalently denote function $ b_j $ as
		\[
		\myfunc{b_j} { \mathrm{V}(G_t) } { \mathfrak{P}_{BB}(N) } { v } { b_j( v )=( o_i ) }  
		\]
	\end{Bremarknote}
	
	\begin{Bremarknote}
		In summary, $ \mathfrak{N}_{BB} (N, f, t, \tau, j) $ is an algorithmic network populated by $N$ nodes/programs from $ \mathfrak{P}_{BB} (N) $ such that, after the first (or $c_0$ cycles) cycle, it starts a diffusion\footnote{ There are other diffusions too. However, only the one from the biggest partial output is independent of neighbor's partial outputs.} process of the biggest partial output (given at the end of the first cycle) determined by a time-varying graph $G_t$ that belongs to a family of graphs $ \mathbb{G}( f, t, \tau ) $. Then, at the last time instant, diffusion stops and one cycle (or more) is spent\footnote{ This condition is necessary to make this algorithmic network defined even when $|\mathrm{T}(G_t)|=1$.} in order to make each node return a final output. Each node returns as final output its previous partial output determined at the last time instant --- see Definition \ref{BdefITFOprot} . 
	\end{Bremarknote}
	
	\begin{Bremarknote}
		Note that, in this model, the aspects of the graphs in the family $\mathbb{G}$ that are mapped into the properties of the population $ \mathfrak{P}_{BB}(N) $ by functions $b_j$ are nodes and time instants. 
	\end{Bremarknote}
	
	\begin{Bremarknote}\label{BnoteRestrictedfamilyGinN_BB}
		The reader is invited to note that the main results presented in this paper also hold for only one function $b_j$ per graph in the family $ \mathbb{G} $.
	\end{Bremarknote}
	
	\begin{Bsubdefinition}\label{BdefC_BB}
		We denote $ \mathbb{N}|_1^{ n } = \{ 1 , \dots , n \mid n \in \mathbb{N} \}  $ as $ \mathfrak{C_{BB}} $. 
	\end{Bsubdefinition}

\end{Bdefinition}

\noindent \\


\section{Average Emergent Open-Endedness (AEOE)}\label{sectionAEOE} 

\subsection{Discussion on emergence and complexity}\label{Bdiscussionemergenceandcomplexity}

In this section, we will introduce a formal definition of average emergent open-endedness (AEOE). Some of the current discussion is already presented in Section~\ref{subsectionDiscussiononBBIM...} \footnote{ And we recommend its reading before the current subsection.} . Thus, we will abstain from addressing the problem of how AEOE relates to the BBIG. Instead, we will discuss general aspects and the idea that the former mathematically grasps.

As discussed in Section \ref{sectionAN}, one can understand the theory of algorithmic networks as a mathematical framework for studying interacting complex systems. However, we are limiting ourselves to deal with computable complex systems that send and receive information from neighbors (during cycles) in order to return final outputs. This way, as mentioned in discussion \ref{subsectionDiscussiononBBIM...}, we can apply the tools from theoretical computer science and information theory --- specially, algorithmic information theory --- to the study of how a systemic feature like complexity ``behaves'' as this system of interacting subsystems runs (see discussion \ref{subsectionDiscussiononBBIM...}). So, our central inquiry falls under investigating how much complexity emerges when subsystems are interacting compared with the case in which they are running completely isolated from each other.

We follow a consensual abstract notion of \emph{emergence}~\cite{DOttaviano2004,Nicolis2009,Prokopenko2014,Hernandez-Orozco2016,Prokopenko2009,Boschetti2005,Fernandez2013,Gershenson2007,Bar-Yam1997,Boschetti2005} as a systemic feature or property that appears only when the system is analyzed (theoretically or empirically) as a ``whole''. Thus, the algorithmic complexity (i.e., an irreducible number of bits of information) of a node/program's final output when networked\footnote{ That is, interacting with other parts of the system.} minus the algorithmic complexity of a node/program's final output when isolated formally defines an irreducible quantity\footnote{ Note that this quantity of bits may be $0$ or negative. Therefore, this measure of emergent algorithmic complexity may also be suitable for measuring the cases where algorithmic complexity was ``lost'' when the system is networked. We leave for future research the study of such degenerate cases. } of information that \emph{emerges} in respect to a node/program\footnote{ More precisely, in respect to a node's final output, i.e., the complexity of its fitness/payoff as discussed in Section~\ref{subsectionDiscussiononPPBBIG}.} that belongs to an algorithmic network. We call it as \emph{emergent algorithmic complexity (EAC)} of a node/program. Consequentially, note that if a system is analyzed as a separated\footnote{ The subparts do not need to be necessarily apart from each other, but each part in this case would be taken as an object of investigation where no information enters or exits anyway.} collection of ``subparts'', the EAC of a node/program will be always $0$.

An important distinction is crucial: the EAC of a node/program must be not confused with EAC of the \emph{entire} algorithmic network. Measuring the emergent algorithmic complexity of the algorithmic network taking into account every node/program ``at the same time'' is --- as our intuition demands to be --- mathematically different from looking at each individual final output's algorithmic complexity. For example, one may consider the algorithmic information of each node/program combined (in a non-trivial way) with the algorithmic information of the network's topology \cite{Buhrman1999, Zenil2014,Mowshowitz2012}. Within the framework of algorithmic networks, this ``whole'' emergent algorithmic complexity can be formally captured by the \emph{joint} algorithmic complexity of each node/program's final output when networked minus the \emph{joint} algorithmic complexity of each node/program's final output when isolated. That is, the algorithmic complexity of the networked population's output as whole minus the algorithmic complexity of the isolated population's output as a whole. However, analyzing this systemic property is not part of the scope of the present article and it will be a necessary future research. Specially in investigating different emergent complexity phase transitions as discussed in Section~\ref{subsectionOE}.

Therefore, instead of investigating the \emph{joint} or \emph{global} EAC of an algorihtmic network, one may look for a mean value of EAC for all nodes/programs. That is, we are focusing the \emph{local} EAC. The \emph{average} (local) emergent algorithmic complexity of a node/program (AEAC) is defined by the mean on all nodes/programs' (and possible network's topologies) EAC (see Definition \ref{BdefAEAC} ). It gives the average complexity of the nodes/programs' respective fitnesses (or, in a game-theoretical interpretation, payoffs) in a networked population, once there is a fitness function that evaluates final outputs. For further discussions, see Section~\ref{subsectionDiscussiononBBIM...} . 

As the main object of mathematical investigation of this article is an algorithmic network $ \mathfrak{N}_{BB} (N, f, t, \tau, j) $ in which the population is randomly generated from a stochastic process of independent and identically distributed random variables under a self-delimiting program-size probability distribution, we refer to the average EAC in $ \mathfrak{N}_{BB} (N, f, t, \tau, j) $ as \emph{expected} emergent algorithmic complexity (EEAC), denoted here (see Definition~\ref{BdefEEACN_BB}) by
\[
\mathbf{E}_{ \mathfrak{N}_{BB}(N,f,t,\tau) } 
\left(
{ {\displaystyle{\myDelta_{iso}^{net}} A} (o_i,c)} 
\right)
\] 
\noindent Therefore, both terms, \emph{average or expected}, can be used interchangeably in the present article.

\subsection{Discussion on open-endedness}\label{subsectionOE}

Definition \ref{BdefAEAC} will tell us how much emergent algorithmic complexity emerges on the average in an algorithmic network $\mathfrak{N}$. The question is: which algorithmic networks produces more AEAC? And an immediate question would be: how does AEAC increase as the population grows in size? More AEAC means that a node/program needs more irreducible information on the average than it already contains, should it try to compute isolated what it does networked. A system with a larger AEAC ``informs'' or ``adds'' more information to its parts on the average. Thus, the investigation of these questions are directly related to important topics in complex systems, like integration \cite{Oizumi2014}, synergy (see Section \ref{sectionRelatedandFuture} ), and emergence (see Section~\ref{Bdiscussionemergenceandcomplexity}).  

Another important concept mentioned in discussion \ref{subsectionDiscussiononBBIM...} is open-endedness. It is commonly defined in evolutionary computation and evolutionary biology as the inherent potential of a evolutionary process to trigger an endless increase in complexity \cite{Abrahao2015,Hernandez-Orozco2016,Adams2017,Bedau1998,Channon2001,Maley1999,Standish2003,Ruiz-Mirazo2004}. That means that, in the long run, it will eventually appear an organism that is as complex as one may want.\footnote{ As a consequence, an infinite number of different organisms tends to appear in the evolutionary line after an infinite time, or after an infinite amount of successive mutations.} Given a certain complexity value as target, one would just need to wait a while in order to appear an organism with a larger complexity than (or equal to) the target value --- no matter how big this value is. In fact, within the framework of algorithmic information theory, as shown in \cite{Chaitin2012,Chaitin2013,Chaitin2014,Abrahao2015}, a cumulative\footnote{ Which allows organisms to recall its predecessors. } evolution reaches $N$ bits of algorithmic complexity after --- realistic fast --- $ \mathbf{ O }( N^2 ( \log(N) )^2 ) $ successive algorithmic mutations on one organism at the time --- whether your organisms are computable, sub-computable or hyper-computable (see \cite{Abrahao2015,Abrahao2016,Abrahao2016c}).  

What we have found is that, within the theory of algorithmic networks, open-endedness is a mathematical phenomenon closely related to the core questions in the first paragraph of this subsection. However, it emerges as an akin --- but different --- phenomenon to evolutionary open-endedness: instead of achieving an unbounded quantity of algorithmic complexity over time (or successive mutations), an unbounded quantity of emergent algorithmic complexity is achieved as the population size increases indefinitely. Since it is a property that emerges depending on the amount of parts of a system only when these parts are interacting somehow (e.g., exchanging information), it arises as an emergent property. So, we refer to it as \emph{emergent open-endedness}. As discussed in Section~\ref{Bdiscussionemergenceandcomplexity}, since we are dealing only with the local EAC and not with the global (or joint) EAC, then a more accurate term would be \emph{local emergent open-endedness}. For the sake of simplifying our nomenclature, we choose to omit the term ``local'' in this article.

The main results of this paper (see Theorems \ref{BthmMain} and \ref{BthmMainCentralTime} and Corollaries \ref{BcorMain} and \ref{BcorDiameter}) prove that there are conditions on the network and communication protocols of the population that trigger a type of emergent open-endedness as the randomly generated populations grows toward infinity. They cause the AEAC to increase as one may want, should the population size increases sufficiently.   
Since it is an increase in average complexity (see discussion \ref{Bdiscussionemergenceandcomplexity}), we refer to it as \emph{average (local) emergent open-endedness} (AEOE) (see Definition \ref{BdefAEOE}). Moreover, since the population is randomly genarated in a way such that AEAC turns to EEAC in the case of algorithmic networks $ \mathfrak{N}_{BB} (N, f, t, \tau, j) $, one may refer to AEOE as \emph{expected (local) emergent open-endedness} (EEOE) (see Definition \ref{BdefEEOE}). For example, even if the AEAC is $0$ or negative for some algorithmic networks with finite-size populations, the average emergent open-endedness phenomenon tells us that, for large enough population sizes, the probability that these algorithmic networks have a larger AEAC tends to $1$ --- see Note \ref{noteSLLN}. In other words, we will show that there are conditions for algorithmic networks such that a phase transition in AEAC eventually occurs, should the AEAC be non-positive for small populations. Thus, one may call this transition as \emph{average}\footnote{ Or \emph{expected} in the case the population is randomly generated. } \emph{emergent complexity phase transition}.

\subsection{Definitions on emergent algorithmic information}

\begin{Bdefinition}\label{BdefEAC}
	The \emph{emergent algorithmic complexity (EAC)} of a node/program $o_i$ in $c$ cycles is given in an algorithmic network that always produces partial and final outputs by

	\[
	{\displaystyle{\myDelta_{iso}^{net(b)}}} A (o_i,c) = 
	A (\mathbf{U}(p_{net}^{b} ( o_i , c )) - 
	A (\mathbf{U}(p_{iso} ( o_i , c) )
	\]
	
	\noindent where:
	\begin{enumerate}
		\item $ f_o( o_i ) \in L $;\footnote{ See Definition~\ref{BdefPopulationasaset}.}
		
		\item $ p_{net}^{b} $ is the program that computes cycle-per-cycle the partial outputs of $o_i$ when networked assuming the position $ v $, where $ b(v,\mathbf{\bar{x}}) = (o_i, b_{ dim( Y ) - 1 }( \mathbf{ \overline{x} } ) ) $, in the graph $G$ in the specified number of cycles $c$ with network input $w$. Thus, $ p_{net}^{b} ( o_i , c ) $ represents the program that returns the final output of $o_i$ when networked assuming the position $ v $, where $ b(v,\mathbf{\bar{x}}) = (o_i, b_{ dim( Y ) - 1 }( \mathbf{ \overline{x} } ) ) $, in the graph $G$ in the specified number of cycles $c$ with network input $w$; 
		
		\item $ p_{iso} $ is the program that computes cycle-per-cycle the partial outputs of $o_i$ when isolated in the specified number of cycles $ c $ with network input $w$. Thus, $ p_{iso} ( o_i , c)  $ represents the program that returns the final output of $o_i$ when isolated in the specified number of cycles $ c $ with network input $w$;
	\end{enumerate}
	
	\begin{Bremarknote}
		For the sake of simplifying notation, since the network input is fixed in our forthcoming results, we choose to omit $w$ in $ p_{net}^{b} ( o_i , c ) $ and $ p_{iso} ( o_i , c)  $;
	\end{Bremarknote}
	
	\begin{Bremarknote}
		Note that:
		\begin{enumerate}
			\item $ A (\mathbf{U}(p_{net}^{b} ( o_i , c) )) $ is the algorithmic complexity of what the node/program $o_i$ returns in the end of all cycles \emph{when networked};
			
			\item $ A (\mathbf{U}(p_{iso} ( o_i , c) ) $ is the algorithmic complexity of what the node/program $o_i$ returns in the end of all cycles \emph{when isolated};
		\end{enumerate}
	\end{Bremarknote}

	\begin{Bremarknote}
		Whereas program $p_{iso}$ may be very simple, since it is basically a program that reiterates partial outputs of $o_i$ as inputs to itself at the beginning of the next cycle (up to $c$ times), program  $p_{net}^{b}$ may also comprise giving the sent partial outputs (which in turn may contain arbitrary information) from node $o_i$'s incoming neighbors at the end of each cycle as inputs to $o_i$ at the beginning of the respective next cycle. Thus, $p_{net}^{b}$  may be only encoded by a much more complex program.  
	\end{Bremarknote}
	
	\begin{Bremarknote}
		Note that the algorithmic complexity of $p_{iso}$ or $p_{net}^{b}$ may be not directly linked to $ A (\mathbf{U}(p_{iso} ( o_i , c) ) $ or $ A (\mathbf{U}(p_{net}^{b} ( o_i , c) )) $ respectively, since $ A (\mathbf{U}(p_{iso} ( o_i , c) ) $ and $ A (\mathbf{U}(p_{net}^{b} ( o_i , c) )) $ are only related to the final outputs (if any) of each node's computation. 
	\end{Bremarknote}
	
	\begin{Bremarknote}
		Remember Definitions \ref{Bdefsensitivetooracles} and \ref{BdefU'} which states that even when a node/program does not halt in some cycle, machine $ \mathbf{U'} $ was defined in order to assure that there is always a partial output for every node/program and for every cycle. If population $ \mathfrak{ P } $ is defined in a way that eventually a partial or final output is not a defined value, when running on the respective theoretical machine (see Definition \ref{BdefPopulation} ), then Definition \ref{BdefEAC} would be inconsistent. This is the reason we stated in its formulation that there always is partial and final outputs. However, this is not necessary in the case of $ \mathfrak{N}_{BB} (N, f, t, \tau, j) $ (see Definition \ref{BdefP_BB} and \ref{BdefEAConN_BB} ).
	\end{Bremarknote}
	\begin{Bremarknote}
		If one defines the \emph{emergent algorithmic creativity} of a node/program $ o_i $ in $ c $ cycles as
		\[
		A ( \textbf{U}(p_{net}^{b} ( o_i , c )) | \textbf{U}(p_{iso} ( o_i , c) ) \text{ ,}
		\]
		\noindent then our results also hold for replacing the expected emergent algorithmic complexity (EEAC) with \emph{expected emergent algorithmic creativity} (EEACr). Since we are estimating lower bounds and $ A( y | x ) $ denotes the conditional prefix algorithmic complexity of string $y$ given string $x$ as input, note that from AIT we have that
		\[
		A ( \textbf{U}(p_{net}^{b} ( o_i , c )) | \textbf{U}(p_{iso} ( o_i , c) )
		\geq
		A (\mathbf{U}(p_{net}^{b} ( o_i , c )) - 
		A (\mathbf{U}(p_{iso} ( o_i , c) )
		-
		\mathbf{O}\big( 1 \big)
		\]
	\end{Bremarknote}

	\begin{Bsubdefinition} \label{BdefEAConN_BB}
		More specifically, one can denote the \emph{emergent algorithmic complexity} of a node/program $o_i$ \emph{in an algorithmic network} $ \mathfrak{N}_{BB} (N, f, t, \tau, j) $ ($ \mathbf{ EAC_{BB} } $) during $c$ cycles as

		\[
		{\displaystyle{\myDelta_{iso}^{net(b_j)}}} A (o_i,c) = 
		A (\mathbf{U}(p_{net}^{b_j} ( o_i ,  c ) )) - 
		A (\mathbf{U}(p_{iso} ( p_i ,  c )))
		\]
		
		\noindent where:
		
		\begin{enumerate}
			\item $ f_o( o_i ) = P_{prot} \circ p_i \in \mathbf{L}_{BB} $ ;
			
			\item $ p_{net}^{b_j} $ is the program that computes cycle-per-cycle the partial outputs of $o_i$ \emph{when networked} assuming the position $ v $, where $ b_j(v)=(o_i) $, in the graph $G_t$ in $c$ cycles with network input $w$. Thus, $ p_{net}^{b_j} ( o_i ,  c ) $ represents the program that returns the final output of $o_i$ \emph{when networked} assuming the position $ v $, where $ b_j(v)=(o_i) $, in the graph $G_t$ in $c$ cycles with network input $w$; 
			
			\item $ p_{iso} $ is the program that computes cycle-per-cycle the partial outputs of $p_i$ \emph{when isolated}\footnote{ Which is a redundancy since we are refering to $p_i$ instead of $o_i$ here.} in $c$ cycles with network input $w$. Thus, $ p_{iso} ( p_i ,  c ) $ represents the program that returns the final output of $p_i$ \emph{when isolated}\footnote{ Which is a redundancy since we are refering to $p_i$ instead of $o_i$ here.} in $c$ cycles with network input $w$; 
		\end{enumerate}

		\begin{Bsubremarknote}
			In order to check if Definition \ref{BdefEAConN_BB} is well-defined, remember Definitions \ref{BdefU'} and \ref{BdefL_BB} . 
		\end{Bsubremarknote}
		
	\end{Bsubdefinition}


\end{Bdefinition}

\begin{Bdefinition}[Notation]
	For the sake of simplifying our notation, let $ \{ b_j \} $ denote\footnote{ The same holds for $ \{ b \} $. } 
	\[ 
	\left\{ b_j \, \Bigg| \, \myfunc{b_j} { \mathrm{V}(G_t) \times \mathrm{T}(G_t) } { Support\left( \mathfrak{P}_{BB}(N) \right) \times \mathbb{N}|_1^{ n }} { (v,t_{c-1}) } { b_j(v,t_{c-1})=( p, x + c ) } \right\}
	\]
	and let	 $ \sum\limits_{ b_j } $ denote 
	\[ 
	\sum\limits_{
	\myfunc{b_j} { \mathrm{V}(G_t) \times \mathrm{T}(G_t) } { Support\left( \mathfrak{P}_{BB}(N) \right) \times \mathbb{N}|_1^{ n }} { (v,t_{c-1}) } { b_j(v,t_{c-1})=( p, x + c ) }
	}
	\] 
	
	\begin{Bremarknote}
		Or, analogously, let $ \{ b_j \} $ denote\footnote{ The same holds for $ \{ b \} $. } 
		\[ 
		\left\{ b_j \, \Bigg| \, \myfunc{b_j} { \mathrm{V}(G_t) \times \mathrm{T}(G_t) } {  \mathfrak{P}_{BB}(N)\times \mathbb{N}|_1^{ n }} { (v,t_{c-1}) } { b_j(v,t_{c-1})=( o_i , x + c ) } \right\}
		\]
		and let	 $ \sum\limits_{ b_j } $ denote 
		\[ 
		\sum\limits_{
			\myfunc{b_j} { \mathrm{V}(G_t) \times \mathrm{T}(G_t) } { \mathfrak{P}_{BB}(N)  \times \mathbb{N}|_1^{ n }} { (v,t_{c-1}) } { b_j(v,t_{c-1})=( o_ i , x + c ) }
		}
		\] 
		as in Definition~\ref{BdefPopulationasaset} and~Note~\ref{BnoteFunctionBinAN}, if the population is an ordered set of labels from a sequence.
	\end{Bremarknote}
	
\end{Bdefinition}

\begin{Bdefinition}\label{BdefAEAC}
	
	We denote the \emph{average emergent algorithmic complexity of a node/program (AEAC)} in an algorithmic network $ \mathfrak{N} = (G, \mathfrak{ P }, b) $ with network input $w$ and $ \left| \mathfrak{ P } \right| = N $, as 
	\begin{align*}
	&\mathbf{E}_{ \mathfrak{N} } 
	\left(
	{ {\displaystyle{\myDelta_{iso}^{net}} A} (o_i,c)} 
	\right)
	= \\
	&=
	{\tiny 
		\sum\limits_{ b }
	}
	\frac{
		\frac{
			\sum\limits_{ o_i \in \mathfrak{P} } { {\displaystyle{\myDelta_{iso}^{net(b)}} A} (o_i,c)}
		}
		{N} 
	}
	{ | \{ b \} | } 
	\end{align*}
	\begin{Bremarknote}\label{BnoteRestrictedfamilyGinAEAC}
		As in Note \ref{BnoteRestrictedfamilyGinN_BB}, if only one function $b$ exists per population, then there is only one possible network's topology linking each node/program in the population. So, in this case, 
		\[
		\mathbf{E}_{ \mathfrak{N} } 
		\left(
		{ {\displaystyle{\myDelta_{iso}^{net}} A} (o_i,c)} 
		\right)
		=
		\frac{
			\sum\limits_{ o_i \in \mathfrak{P} } { {\displaystyle{\myDelta_{iso}^{net(b)}} A} (o_i,c)}
		}
		{ N } 
		\]

	\end{Bremarknote}
	
\end{Bdefinition}

\begin{Bdefinition}\label{BdefEEACN_BB}
	
	We denote the \emph{expected emergent algorithmic complexity of a node/program} for algorithmic networks $ \mathfrak{N}_{BB}(N,f,t,\tau,j) $ ($ \mathbf{ EEAC_{BB} } $) with network input $w$, where $0<j \leq |\{ b_j \}|$, as 
	
	\begin{align*}
	&\mathbf{E}_{ \mathfrak{N}_{BB}(N,f,t,\tau) } 
	\left(
	{ {\displaystyle{\myDelta_{iso}^{net}} A} (o_i,c)} 
	\right)
	= \\
	&=
	{\tiny 
		\sum\limits_{ b_j  }
	}
	\frac{
		\frac{
			\sum\limits_{ o_i \in \mathfrak{P}_{BB}(N) } { {\displaystyle{\myDelta_{iso}^{net(b_j)}} A} (o_i,c)}
		}
		{N} 
	}
	{ | \{ b_j \} | }
	= \\ 
	&= 
	\frac{
		{
			\sum\limits_{ b_j  }
		}
		\frac{
			\sum\limits_{ o_i \in \mathfrak{P}_{BB}(N) } { A (\mathbf{U}(p_{net}^{b_j} ( o_i ,  c ) )) - 
				A (\mathbf{U}(p_{iso} ( p_i ,  c ))) }
		}
		{N}
	}
	{ | \{ b_j \} | } 
	\end{align*}
	
	\begin{Bremarknote}\label{BnoteRestrictedfamilyGinEEAC_BB}
		As in Notes \ref{BnoteRestrictedfamilyGinN_BB} and \ref{BnoteRestrictedfamilyGinAEAC}, if only one function $b_j$ exists per population, then only one possible network's topology will be linking each node/program in the population. So, in this case, 
		\[
		\mathbf{E}_{ \mathfrak{N}_{BB}(N,f,t,\tau) } 
		\left(
		{ {\displaystyle{\myDelta_{iso}^{net}} A} (o_i,c)} 
		\right)
		=
		\frac{
			\sum\limits_{ o_i \in \mathfrak{P}_{BB}(N) } { {\displaystyle{\myDelta_{iso}^{net(b_j)}} A} (o_i,c)}
		}
		{N} 
		\]

	\end{Bremarknote}
	
\end{Bdefinition}


\begin{Bdefinition}\label{BdefAEOE}
	We say that a class (or family) of algorithmic networks $ \mathfrak{N} $, where $N$ is the population size, has the property of \emph{average (local) emergent open-endedness (AEOE)} for a given network input $w$ in $c$ cycles \textit{iff} 
	\[
	\lim_{ N \to \infty } \mathbf{E}_{ \mathfrak{N} } 
	\left(
	{ {\displaystyle{\myDelta_{iso}^{net}} A} (o_i,c)} 
	\right)
	=
	\infty
	\]  
\end{Bdefinition}

\begin{Bsubdefinition}\label{BdefEEOE}
	We say algorithmic networks $ \mathfrak{N}_{BB}(N,f,t,\tau) $ have the property of \emph{expected (local) emergent open-endedness (EEOE)} for a given network input $w$ in $c$ cycles \textit{iff} 
	\[
	\lim_{ N \to \infty } \mathbf{E}_{ \mathfrak{N}_{BB}(N,f,t,\tau) } 
	\left(
	{ {\displaystyle{\myDelta_{iso}^{net}} A} (o_i,c)} 
	\right)
	=
	\infty
	\] 
\end{Bsubdefinition}



\section{Cycle-bounded conditional halting probability}\label{sectionOmega}

\subsection{Discussion on halting probabilities}

As presented in \cite{Chaitin1975,Li1997,Calude2002,Chaitin2004,Barmpalias2016,Downey2010,Calude2001,Calude2009}, the number $ \Omega $ is the probability of randomly picking a halting program given a program-size probability distribution for a self-delimiting language (i.e., an algorithmic-informational Lebesgue measure). So, a conditional probability $ \Omega(w) $ is the probability of randomly picking a halting program with $w$ as its input. A variation of halting probabilities on self-delimiting universal programming languages is the \emph{resource-bounded} halting probability: the probability of randomly picking a halting program on prefix universal machine with limited computation time or memory \cite{Li1997,Allender2011,Abrahao2015}. However, for isolated populations of arbitrary self-delimiting programs, another variation is needed. The ``scarce'' resource is not time or memory, but cycles (or, in our current model, communication rounds). Since isolated nodes/programs are computing during cycles (see Definition~\ref{BdefCycle} ) that determine the maximum number of times a node/program may reuse (through successive halting computations) its partial output before returning a final output, the \emph{cycle-bounded conditional halting probability} $ \Omega( w,c ) $ is the probability of randomly picking a program that halts for every cycle until $c$ with $w$ as its initial input.  

In Section~\ref{subsectionDiscussiononBBIM...} we have discussed that we want to compare the average algorithmic complexity of an algorithmic network playing the BBIG with its population when nodes/programs are completely isolated from each other. Thus, in order to prove the central Theorem \ref{BthmMain} one will need to calculate an upper bound for the expected algorithmic complexity of what a node/program does when isolated during the cycles. The cycle-bounded conditional halting probability $ \Omega( w,c ) $ is used with the purpose of normalizing the probabilities of halting nodes/programs in Lemmas \ref{BlemmaGibbsandalgorithmicentropy} and \ref{BlemmaComplexityonBarHalt}. Since we need to calculate probability in each case and the population is randomly generated, $ \Omega( w,c ) $ necessarily arises as an important constant in Theorem \ref{BthmMain}. 

\subsection{Definitions on cycle-bounded halting probability}

\begin{Bdefinition}\label{BdefHalt}
	We denote the population composed of $ p_i \in \mathfrak{P} $, where $ 1 \leq i \leq N $ and $ P_{prot} \circ p_i \in \mathbf{L}_{BB} $, such that $p_i$ always halts on network input $w$ in \emph{every} cycle until $c$ when $ P_{prot} \circ p_i $ is \emph{isolated} as
	\[
	Halt_{iso} ( \mathfrak{P}, w, c )
	\]
	
	\begin{Bsubdefinition}\label{BdefBarHalt}
		Analogously, we denote the population of $ p_i \in \mathfrak{P} $, where $ 1 \leq i \leq N $ and $ P_{prot} \circ p_i \in \mathbf{L}_{BB} $, such that $p_i$ does \emph{not} halt on network input $w$ in at least one cycle until $c$ when $ P_{prot} \circ p_i $ is isolated as
		
		\[
		\overline{ Halt }_{iso}( \mathfrak{P}, w, c ) 
		\] 
	\end{Bsubdefinition}
	
	\begin{Bremarknote}
		Note that, since a set is a special case of a multiset in which each element has multiplicity $1$, then both Definitions~\ref{BdefHalt} and~\ref{BdefBarHalt} are well-defined for an arbitrary language $L$ instead of a population $ \mathfrak{P} $, so that one can denote it as $ Halt_{iso}( L, w, c ) $ and $ \overline{ Halt }_{iso}( L, w, c ) $ respectively.
	\end{Bremarknote}
\end{Bdefinition}

\begin{Bdefinition}\label{BdefOmegawc}
	We denote the \emph{cycle-bounded conditional halting probability}\footnote{ That is, a conditional Chaitin's Omega number for isolated programs in a population of prefix Turing machines. } of an isolated node/program in a language $\mathbf{L_U}$ that always halts for an initial input $w$ in $c$ cycles as 
	\[
	\Omega( w,c )
	=
	\sum\limits_{ p_i \in Halt_{iso}( \mathbf{L_U}, w, c )  } 
	{ \frac {1} { 2^{ |p_i| } } }
	=
	\lim\limits_{ N \to \infty }
	\sum\limits_{ p_i \in Halt_{iso}( \mathbf{L_U}(N), w, c )  } 
	{ \frac {1} { 2^{ |p_i| } } }
	\]
	\begin{Bremarknote}
		Since $\mathbf{L_U}$ is self-delimiting and it is a complete code (see Definition~\ref{BdefConcatenation}), the probability of each program is well-defined (i.e., by classical Lebesgue measure in algorithmic information theory \cite{Li1997,Downey2010,Calude2002}). Hence, one can define the halting probability $\Omega$ \footnote{ Note that the Greek letter $\Omega$ here does not stand for an asymptotic notation opposed to the big O notation. } for $\mathbf{L_U}$. Further, the same holds for conditional halting probability $\Omega(w)$, i.e., the probability that a program halts when $w$ is given as input. Then, the set $ Halt_{iso}( \mathbf{L_U}, w, c ) $ of isolated nodes/programs that always halt on initial input $w$ in $ c \geq 1 $ cycles is a proper subset of the set of programs that halt. Therefore, 
		\[
		\Omega(w,c) \leq \Omega(w) = \Omega(w,1) < 1
		\]
		\noindent In fact, one can prove that
		\[
		\Omega(w,c') \leq \Omega(w,c)
		\]
		\noindent when $ c' \geq c > 0 $. On the other hand, one can also build programs that always halt for every input and for every number of cycles. So, for every $ c > 0 $
		\[
		0 < \Omega(w,c) < 1
		\] 
	\end{Bremarknote}

\end{Bdefinition}

\section{Average diffusion density}\label{sectionTau}

\subsection{Discussion on diffusion in the BBIG}
On the ``networked side of the equation'' in Theorem \ref{BthmMain}, as discussed in Section~\ref{subsectionDiscussiononBBIM...}, we are looking for a lower bound for the expected algorithmic complexity of what a node/program does when networked playing the Busy Beaver imitation game (BBIG). This is the value from which will be subtracted the upper bound for the expected algorithmic complexity of what a node/program does when isolated.

Remember that, in the BBIG, once nodes/programs start sharing their partial outputs, the only relevant feature that matters from then on is who has the largest integer. In other words, the neighbor with a bigger partial output will ``infect'' a node/program with a smaller partial output. Note that it may be analogous to a multi-``disease'' spreading, since there are several different graph arrangements of randomly generated nodes/programs in which there might be separated clusters with different respective ``diseases'' (in our case, different integers) spreading at the same time for a while. However, the BBIG rules cause one type of partial output to spread over any other partial output: the biggest partial output of the randomly generated population. As we will stress in proof of Theorem \ref{BthmMain}, this is one of the mathematical shortcuts to reach a desired proof of a lower bound for the expected emergent algorithmic complexity of what a node/program does when networked playing the BBIG. It is given by the average density of nodes/programs infected by a single source in Definition \ref{BdefAveragePartitionmaxN_BB} --- in the case, the source is the randomly generated node/program that returns the largest integer. Thus, we call it as \emph{average singleton diffusion density}. Note that, as there is a probability of two or more nodes/programs being generated which return the largest integer, the average singleton diffusion density is actually a lower bound for the \emph{average diffusion density} (which might be multisource). Therefore, although Lemma \ref{BlemmaMinComplexityonDiffusion} is proved using the average singleton diffusion density, it also holds for the average diffusion density --- since it works like a lower bound. The average diffusion density is related to prevalence in \cite{Pastor-Satorras2001a,Pastor-Satorras2002}, although there is no ``curing'' process in our case. And it might be also related to flooding time in \cite{Clementi2010} --- see also Section \ref{sectionRelatedandFuture}.

Indeed, this analogy with diffusion or spreading in complex networks is not only a coincidence but is also responsible for bringing the aspects and proprieties studied in complex networks to our current results in algorithmic networks. The EEAC in the BBIG is lower bounded by a value that depends on how powerful the diffusion is on a network. We address more specifically some of theses relations with complex networks in Section \ref{sectionRelatedandFuture}.

\subsection{Definitions on the average diffusion density}

\begin{Bdefinition}
	Let $ A_{max}( \mathfrak{N} , c ) $ denote the algorithmic complexity of the biggest final output returned by a member of the population $\mathfrak{P}$ during a maximum number of $c$ cycles, where $\mathfrak{N}=(G, \mathfrak{P},b)$.
	
	\begin{Bsubdefinition}\label{BdefA_max}
		In the case of $ \mathfrak{N}_{BB}(N,f,t,\tau, j) $ and $c=1$, for the sake of simplifying our notation, we will just denote it as $ A_{max} $. 
	\end{Bsubdefinition}

\end{Bdefinition}

\begin{Bdefinition}\label{BdefPartitionmax}
	Let $ t_i, t, t' \in \mathrm{T}(G_t) $ with $ t_i \leq t \leq t' $.
	We denote the fraction of nodes/programs to which the node/program with the ``best'' partial output at time instant $ t_i $ in the algorithmic network $ \mathfrak{N}=(G,\mathfrak{P},b) $ diffuses its partial output during time instants $ t $ until $ t' $, given that this diffusion started at time instant $t_i$, as
	\[ 
	{ \tau_{max}( \mathfrak{N}, t_i ) }|_{t}^{t'} 
	=
	\frac{ \left| \mathbf{X}_{ { \tau_{max}( \mathfrak{N}, t_i ) }|_{t}^{t'} } \right| }{ \left| \mathrm{V}( G )  \right| }
	\] 
	\noindent or
	\[
	{ \tau_{max}( G,\mathfrak{P},b, t_i ) }|_{t}^{t'}
	=
	\frac{ \left| \mathbf{X}_{ { \tau_{max}( \mathfrak{N}, t_i ) }|_{t}^{t'} } \right| }{ \left| \mathrm{V}( G )  \right| }
	\]
	
	\noindent and we call it as \emph{singleton diffusion density}.
	\begin{Bremarknote}\label{Bbestpartialoutput}
		The notion of what is the ``best'' partial output may vary on how the algorithmic network $\mathfrak{N}$ is defined --- in some algorithmic networks the notion of ``best'' partial output may even be not defined. We consider the ``best'' partial output as being the one that always affects the neighbors to which it is shared by making them to return a final result that is at least as ``good'' as the one that the node with the ``best'' partial output --- that is, the one that started this diffusion --- initially had. How good is a final result also depends on how $ \mathfrak{N} $ is defined and on how one defines what makes a result better than another (e.g., a fitness function). Whereas this general definition is not formally stated, these matters become formal and precise in Definition \ref{BdefPartitionmaxN_BB} in the particular case of the present paper --- see also~\cite{Chaitin2012,Chaitin2014,Hernandez-Orozco2016,Abrahao2015,Abrahao2016,Abrahao2016c,Hernandez-Orozco2018} for the Busy Beaver function as our complexity measure of fitness. 
	\end{Bremarknote}
	
	\begin{Bremarknote}
		Note that a node/program does only belong to $ \mathbf{X}_{ { \tau_{max}( \mathfrak{N}, t_i ) }|_{t}^{t'} } $ \textit{iff} its partial outputs between time instants $t$ and $t'$ historically depends on the information diffused by the node/program with the ``best'' partial output on time instant $ t_i $. 
	\end{Bremarknote}
	
	\begin{Bremarknote}
		\[
		0 \leq { \tau_{max}( \mathfrak{N}, t_i ) }|_{t}^{t'} \leq 1
		\]
	\end{Bremarknote}

	\begin{Bsubdefinition}\label{BdefPartitionmaxN_BB}
		We denote the fraction of nodes/programs to which the node/program with the biggest partial output at the first cycle in the algorithmic network $ \mathfrak{N}_{BB}(N,f,t_i,\tau, j) $ diffuses this partial output during time instants $t$ until\footnote{ Note that the nodes that got ``infected'' at the end of time interval between time instants $ t' - 1 $ and $ t' $ are not being taken into account in this fraction $ \tau_{max} $.} $t'$, given that this diffusion started at time instant $t_i$, as
		
		\[
		{ \tau_{max}( N,f,t_i,\tau, j ) }|_{t}^{t'}
		=
		\frac{ \left| \mathbf{X}_{ { \tau_{max}( N,f,t_i,\tau, j ) }|_{t}^{t'} } \right| }
		{ N }
		\] 
	\end{Bsubdefinition}
	
	\begin{Bsubsubdefinition}\label{BdefAveragePartitionmaxN_BB}
		In the average case for all possible respective node mappings into the population, we denote the \emph{average (or expected) singleton diffusion density} as
		
		\[
		{ \tau_{\mathbf{E}(max)}( N,f,t_i,\tau ) }|_{t}^{t'} 
		= 
		\sum_{\tiny 
			b_j
		} \frac{ { \tau_{max}( N,f,t_i,\tau, j ) }|_{t}^{t'}
		}
		{ | \{ b_j \} |
		}
		\] 
	\end{Bsubsubdefinition}
	
	\begin{Bsubsubremarknote}
		Note that this mean is being taken from a classical uniform distribution on the space of functions $ b_j $. An interesting future research will be to extend the results of this article to non-uniform cases on $ b_j $. 
	\end{Bsubsubremarknote}

\end{Bdefinition}

\section{Time centralities}\label{sectionTimecentralities}

\subsection{Discussion on how to trigger EEOE}\label{subsectionDiscussiontimecentralitiesandEEOE}

We have previously mentioned that our model of algorithmic networks playing the BBIG was built on time-varying graphs (see Section \ref{sectionBusyBeaverGame}), so that the static case become also covered as an instance of time-varying graphs --- see Definition~\ref{BdefStaticNetwork} and Note \ref{BnoteDiameterandstaticgraphs}. In fact, a preliminary result for static networks was presented in \cite{Abrahao2016b}. The reader is also invited to note that a model in which the population size varies during the cycles is a particular case of $ \mathfrak{N}_{BB}( N, f, t_i, \tau, j ) $ where $ N = \omega $ \footnote{ Note that $ \omega $ is the first infinite ordinal number and it is an ordered set with the same cardinality of the natural numbers $ \mathbb{N} $.} and each new edge is ``added'' in $ G_t $ in the exact cycle (or time instant) that an edge links the older network's population to a newborn node/program \footnote{ And analogously a similar procedure applies when nodes are being ``killed'' during the cycles.} --- and we leave for future research proving the present results for graphs with $ |\mathrm{V}(G_t)| = \omega $. Within the framework of algorithmic dynamic networks in which one is interested in investigating conditions that trigger EEOE, a natural question is when it is the best time instant (or cycle) to allow communication between nodes/programs with the purpose of generating more EEAC. Note that there may be cases where there is a cost in using a communication channel (i.e., an edge) for example. Since the information sharing relies only on diffusion processes in the BBIG played by algorithmic networks $ \mathfrak{N}_{BB}( N, f, t_i, \tau, j ) $, we tackle this problem following the approach and metrics developed in \cite{Guimaraes2013,Costa2015a}. The more powerful the diffusion, the more $ \mathfrak{N}_{BB}( N, f, t_i, \tau, j ) $ can be optimized (see discussion \ref{subsectionDiscussiononBBIM...}). So, in the present case, finding the central time to start a diffusion process is directly related to finding the central time to trigger expected emergent open-endedness.

We will present two related types of time centralities on $ \mathfrak{N}_{BB}( N, f, t_i, \tau, j ) $ (see \ref{BnoteunderdefRelatingtwotimecentralities}). Both relies on the chosen diffusion measure as discussed in Section~\ref{subsectionDiscussiononDGMD} and both depends on how time instants correspond to cycles (see function $c(x)$ in Definitions~\ref{BdefTimecentrality1} and \ref{BdefTimecentrality2}) . First, we define the central time to trigger expected emergent open-endedness in Definition~\ref{BdefTimecentrality1}. Second, we define the central time to trigger maximum expected emergent algorithmic complexity in Definition~\ref{BdefTimecentrality2} . One can think of these time centralities as the minimum value t for which the function $ f(N,t,\tau) $ minimizes the number of cycles, which depends on $ t + f(N,t,\tau) $, in order to generate increasing expected emergent algorithmic complexity of a node when $N$ tends to $\infty$, given a fraction $\tau$ of the nodes/programs as a diffusion ``yardstick''.

\subsection{Definitions of Time Centralities}\label{BdefTimecentralities}


\begin{Bdefinition}\label{BdefTimecentrality1}
	Let $ w \in \mathbf{L_U} $ be a network input.
	Let $ 0 < N \in \mathbb{N} $.
	Let $ c(x) $ be a non-decreasing total computable function where
	\[
	\myfunc{c}{ \mathbb{N} } { \mathbb{N}^{*} } { x } { c(x)=y }
	\]
	\noindent Let $ \mathfrak{N}_{BB}(N,f,t_z,\tau,j) = ( G_t, \mathfrak{ P }_{BB}(N), b_j ) $, where $ 0 \leq z \leq | \mathrm{T}(G_t) | -1 $, be well-defined, where there is $ t_{z_0} \in \mathrm{T}(G_t) $ such that\footnote{ This condition directly assures that this definition of time centrality is well-defined.} 
	\[
	\lim\limits_{ N \to \infty } 
	\mathbf{E}_{ \mathfrak{N}_{BB}(N,f,t_{z_0},\tau) } 
	\left(
	{ {\displaystyle{\myDelta_{iso}^{net}} A} (o_i, c( z_0 + f( N, t_{z_0}, \tau ) + 2 ) )} 
	\right)
	=
	\infty
	\]
	We define the central time $ t_{cen_1} $ in generating \emph{unlimited} expected emergent algorithmic complexity of a node on a diffusion process over a fraction $ \tau $ that minimizes the number of node cycles\footnote{ Or communication rounds.} as a function\footnote{ A computable function.} of the topological diffusion measure $ f( N, t, \tau ) $ as
	
	\begin{align*}
	& t_{cen_1}(c) = \min \Bigg\{
	t_j \, 
	\Bigg| \,
	j + f( N, t_{j}, \tau ) + 1 = \min \Big\{
	t_i + f( N, t_i, \tau ) \, 
	\Big| t_i + f( N, t_i, \tau ) \leq \\
	& \leq z_0 + f( N, t_{z_0}, \tau ) + 1 \, \land
	\lim\limits_{ N \to \infty } 
	\mathbf{E}_{ \mathfrak{N}_{BB}(N,f,t_{i},\tau) } 
	\Big(
	{\displaystyle{\myDelta_{iso}^{net}} A} (o_i, c( i + f( N, t_{i}, \tau ) + 2 ) ) 
	\Big)
	=
	\infty 
	\Big\} 
	\Bigg\} 
	\end{align*}

\end{Bdefinition}

\begin{Bdefinition}\label{BdefTimecentrality2}
	Let $ w \in \mathbf{L_U} $ be a network input.
	Let $ 0 < N \in \mathbb{N} $.
	Let $ c(x) $ be a non-decreasing total computable function where
	\[
	\myfunc{c}{ \mathbb{N} } { \mathbb{N}^{*} } { x } { c(x)=y }
	\]
	\noindent Let $ \mathfrak{N}_{BB}(N,f,t_z,\tau,j) = ( G_t, \mathfrak{ P }_{BB}(N), b_j ) $, where $ 0 \leq z \leq | \mathrm{T}(G_t) | -1 $, be well-defined, where there is $ t_{z_0} \in \mathrm{T}(G_t) $ such that\footnote{ This condition directly assures that this definition of time centrality is well-defined.} 
	
	\begin{align*}
	& \forall x \in { \mathbb{N} }|_{0}^{ |\mathrm{T}(G_t)| - 1 } \Bigg( \lim\limits_{ N \to \infty } 
	\mathbf{E}_{ \mathfrak{N}_{BB}(N,f,t_x,\tau) } 
	\left(
	{ {\displaystyle{\myDelta_{iso}^{net}} A} (o_i, c( x + f( N, t_x, \tau ) + 2 ) )} \right)
	\leq \\
	& \leq 
	\lim\limits_{ N \to \infty } 
	\mathbf{E}_{ \mathfrak{N}_{BB}(N,f,t_{z_0},\tau) } 
	\left(
	{ {\displaystyle{\myDelta_{iso}^{net}} A} (o_i, c( z_0 + f( N, t_{z_0}, \tau ) + 2 ) )} 
	\right)
	\Bigg)
	\end{align*}
	
	We define the central time $ t_{cen_2} $ in generating the \emph{maximum} expected emergent algorithmic complexity of a node on a diffusion process over a fraction $ \tau $ as the population grows that minimizes the number of cycles as a function of the topological diffusion measure $ f( N, t, \tau ) $ as
	
	\begin{align*}
	& t_{cen_2}(c) = \min \Bigg\{
	t_j \, 
	\Bigg| \,
	j + f( N, t_{j}, \tau ) + 1 = \min \Big\{
	t_i + f( N, t_i, \tau ) \, 
	\Big| \, 
	i + f( N, t_{i}, \tau ) + 1 \leq \\
	& \leq z_0 + f( N, t_{z_0}, \tau ) + 1 \, \land \forall x  \Big( 
	\lim\limits_{ N \to \infty } 
	\mathbf{E}_{ \mathfrak{N}_{BB}(N,f,t_x,\tau) } 
	\left(
	{ {\displaystyle{\myDelta_{iso}^{net}} A} (o_i, c( x + f( N, t_x, \tau ) + 2 ) )} \right)
	\leq \\
	& \leq 
	\lim\limits_{ N \to \infty } 
	\mathbf{E}_{ \mathfrak{N}_{BB}(N,f,t_{i},\tau) } 
	\left(
	{ {\displaystyle{\myDelta_{iso}^{net}} A} (o_i, c( i + f( N, t_{i}, \tau ) + 2 ) )} 
	\right)
	\Big)
	\Big\}
	\Bigg\}
	\end{align*}
	
	\begin{Bremarknote}\label{BnoteunderdefRelatingtwotimecentralities}
		Note that, by definition, since $ | \mathrm{T}(G_t) |  \leq \infty $, if $ t_{cen_1} $ is well-defined, then $ t_{cen_2} = t_{cen_1} $. 
	\end{Bremarknote}
	
\end{Bdefinition}


\section{An EEOE Phenomenon in the BBIG}\label{sectionProofs}

In this section we will prove lemmas, theorems and corollaries with the purpose of showing that $ \mathfrak{N}_{BB}( N, f, t, \tau, j ) = ( G_t, \mathfrak{P}_{BB}(N), b_j ) $ is an algorithmic network capable of EEOE (see Definition \ref{BdefEEOE} ) under certain topological conditions of the graph $G_t$ (see Corollary \ref{BcorMain} ). We will show that there is a trade-off between the average diffusion density and the number of cycles (see Theorem \ref{BthmMain}). Moreover, once these topological properties are met, the concept of central time to trigger EEAC within the minimum number of cycles becomes well-defined (see Theorem \ref{BthmMainCentralTime}).

We split the proof of our main Theorem \ref{BthmMain} in six Lemmas \ref{BlemmaSLLNandAIT}, \ref{BlemmaComplexityp_i}, \ref{BlemmaGibbsandalgorithmicentropy}, \ref{BlemmaComplexityonHaltp_i}, \ref{BlemmaComplexityonBarHalt} and \ref{BlemmaMinComplexityonDiffusion}. Corollary \ref{BcorMain} does not only enables one to apply Theorem \ref{BthmMain} to a network diffusion measure like cover time \ref{BdefCovertime} but also bridges EEAC to time centrality by linking Theorems \ref{BthmMain} and \ref{BthmMainCentralTime} . The main idea behind the construction of the proof of the Theorem \ref{BthmMain} comes from combining an estimation of a lower bound for the average algorithmic complexity of a \emph{networked} node/program and an estimation of an upper bound for the expected algorithmic complexity of an \emph{isolated} node/program. Whereas the estimation of the former comes from the very BBIG from the network dynamics, the estimation of the latter comes from the law of large numbers, Gibb's inequality and algorithmic information theory applied on the randomly generated population $ \mathfrak{P}_{BB}(N) $. Thus, calculating the former estimation minus the latter gives directly the EEAC (see Definition \ref{BdefEEACN_BB}). Corollary \ref{BcorMain} follows directly from Theorem \ref{BthmMain}. The proof of Theorem \ref{BthmMainCentralTime} is an algebraic adjustment of Definitions \ref{BdefTimecentrality1} and \ref{BdefTimecentrality2} into Corollary \ref{BcorMain}.

In appendix \ref{appendix} we present extended versions of these proofs, except for the Corollaries \ref{BcorMain} and \ref{BcorDiameter}.

\subsection{Lemmas}\label{subsectionLemmas}

\begin{Blemma}[or extended Lemma~\ref{lemmaSLLNandAIT}]\label{BlemmaSLLNandAIT}
	Let $ \mathfrak{N}_{BB}( N, f, t, \tau, j ) = ( G_t, \mathfrak{P}_{BB}(N), b_j ) $ be an algorithmic network.
	Let $N$ be the size of a randomly generated population $ \mathfrak{P}_{BB}(N) $. Then, on the average as $N$ grows, we will have that there is a constant $C_4$ such that\footnote{ This result --- on random sequences following algorithmic probabilities --- is inspired by a similar argument found in \cite{Chaitin2012} with the purpose of building a true open-ended evolutionary model for software --- see also \cite{Chaitin2014,Abrahao2015}.  }
	\[ A_{max} \geq \lg(N) - C_4 \]
	\begin{proof}[Short proof]
		Let $\phi_{N}(p)$ be the frequency that $p$ occurs in a random sequence of size $N$. From the \emph{strong law of large numbers} and \emph{algorithmic information theory (AIT)} we have that 
		\begin{equation}\label{BstepSLLN}
		\mathbf{P} \left[ \, \lim\limits_{ N \to \infty } \phi_{N}(p) = \frac{ 1 }{ 2^{ | p | } } \, \right]
		=
		1
		\end{equation}
		\noindent Let $BB(k)$ be the Busy Beaver value for an arbitrary $ k \in \mathbb{N} $ on machine $ \mathbf{U} $.
		From AIT, we know there are constants $C_{\Omega}, C_{BB} \geq 0 $ and a program $p_{BB(k)}$ such that, for every $ w \in \mathbf{L_U} $
		\[
		\mathbf{U}( p_{BB(k)} \circ w) = \mathbf{U}( p_{BB(k)} ) = BB(k)
		\]
		\noindent and
		\[
		k - C_{\Omega}  \leq A(BB(k)) \leq | p_{BB(k)} | \leq k + C_{BB} 
		\] 
		\noindent and
		\[
		\forall x \geq BB(k) \; \big( \, A(x) \geq k - C_{\Omega} \, \big)
		\]
		\noindent Since $k$ was arbitrary, let $ k = \lg(N) - C_{BB} $.
		From Step \eqref{BstepSLLN} we will have that, when $N$ is large enough, one should expect that $ p_{BB( \lg(N) - C_{BB} )} $ occurs 
		\[ 
		\frac{ N }{ 2^{ | p_{BB( \lg(N) - C_{BB} )} | } } 
		\geq
		\frac{ N }{ 2^{ \lg(N) } } 
		=
		1
		\]
		\noindent times within $N$ random tries.
		Let $ C_4 = 2 \, C_{\Omega} + 2 \, C_{BB} $.
		Thus, from Definitions \ref{BdefL_BB}, \ref{BdefA_max} and \ref{BdefN_BB}, since any node/program count as isolated from the network when $c=1$, we will have that, when $ N \to \infty $,
		\begin{align*}
		& A_{max} 
		\geq 
		A(\mathbf{U}( P_{prot} \circ p_{BB( \lg(N) - C_{BB} )} )) - C_{BB} - C_{\Omega}
		= \\
		&
		=
		A(\mathbf{U}(p_{BB( \lg(N) - C_{BB} )})) - C_{BB} - C_{\Omega}
		\geq \\ 
		& 
		\geq 
		\lg(N) - C_{BB} - C_{\Omega} - C_{BB} - C_{\Omega}
		= 
		\lg(N) - C_4  
		\end{align*}
	\end{proof}
\end{Blemma}

\begin{Blemma}[or extended Lemma~\ref{lemmaComplexityp_i}] \label{BlemmaComplexityp_i}
	Given a population $ \mathfrak{P}_{BB}(N) $ defined in \ref{BdefP_BB} , where $ p_i \in Halt_{iso}( \mathbf{L}_{ \mathfrak{P}_{BB}(N) }, w, c ) $ and $ P_{prot} \circ p_i \in \mathfrak{P}_{BB}(N) $ and $ N \in \mathbb{N^*} $ is arbitrary, there is a constant $C_1$ such  that 
	\begin{align*}
	A (\mathbf{U}(p_{iso} ( p_i , c )))
	\leq
	C_1 + |p_i| + A(w) + A(c)
	\end{align*} 
	
	\begin{proof}[Short proof]
		We have that there is at least one program $p'$ such that, for every $p_i \in Halt_{iso}( \mathbf{L}_{ \mathfrak{P}_{BB}(N) }, w, c ) $, 
		\[
		\mathbf{U}(p_{iso} ( p_i , c ))
		=			 \mathbf{U} (p' \circ p_i \circ w \circ c)
		\] 
		Take the shortest such program $p'$. Then, from AIT and Definition \ref{BdefU'} we will have that there is constant $ C_{1} $ such that
		\[
		A ( \mathbf{U}(p_{iso} ( p_i , c )) )			 =
		A ( \mathbf{U} (p' \circ p_i \circ w \circ c) )
		\leq
		C_1 + |p_i| + A(w) + A(c)
		\]
	\end{proof}
\end{Blemma}

\begin{Blemma}[or extended Lemma~\ref{lemmaGibbsandalgorithmicentropy}] \label{BlemmaGibbsandalgorithmicentropy}
	Given a population $ \mathfrak{P}_{BB}(N) $ defined in \ref{BdefP_BB} , where $ p_i \in \mathbf{L_U} $ and $ P_{prot} \circ p_i \in \mathfrak{P}_{BB}(N) $, we will have that
	\[
	\frac
	{1}
	{ \Omega(w,c) }
	\left( \lim_{ N \to \infty } \sum_{ p_i \in Halt_{iso}( \mathbf{L_U}(N), w, c ) } \frac
	{ |p_i| }
	{ 2^{ |p_i| } } \right)
	+ 
	\lg( \Omega(w,c) )
	\leq
	\lim_{ N \to \infty }
	\lg(\Omega(w,c) N) 
	\]
	
	\begin{proof}[Short proof]
		Since $\mathbf{L_U}$ is self-delimited, from AIT, Definitions~\ref{BdefOmegawc} and~\ref{BdefL_U}, Gibb's (or Jensen's) inequality \cite{Cover2005,MacKay2005} and Definition~\ref{BdefOmegawc} and from the law of large numbers, we will have that
		
		\begin{equation}\label{BstepNormalizedalgorithmicentropy}
		\lim_{ N \to \infty }
		\sum_{ p_i \in Halt_{iso}( \mathbf{L_U}(N), w, c ) } \frac
		{ |p_i| + \lg(\Omega(w,c)) }
		{ 2^{ |p_i| + \lg(\Omega(w,c)) } }
		= 
		\end{equation}
		\[
		= \frac
		{1}
		{ \Omega(w,c) }
		\left( \lim_{N \to \infty} \sum_{ p_i \in Halt_{iso}( \mathbf{L_U}(N), w, c ) } \frac
		{ |p_i| }
		{ 2^{ |p_i| } } \right)
		+ 
		\lg( \Omega(w,c) )
		\leq
		\]
		\begin{align*}\label{BstepGibbsinequality}
		& \leq
		\lim_{ N \to \infty }
		\lg( | Halt_{iso}( \mathbf{L_U}(N), w, c ) | )
		\leq
		\lim_{ N \to \infty }
		\lg( \Omega(w,c) N )
		\end{align*}
		
	\end{proof}
\end{Blemma}

\begin{Blemma}[or extended Lemma~\ref{lemmaComplexityonHaltp_i}]\label{BlemmaComplexityonHaltp_i}
	Given a population $ \mathfrak{P}_{BB}(N) $ defined in Definition~\ref{BdefP_BB} , where $ p_i \in \mathbf{L_U} $ and $ P_{prot} \circ p_i \in \mathfrak{P}_{BB}(N) $, we will have that
	
	\[
	\lim_{ N \to \infty }
	\frac{
		\sum\limits_{ b_j } 
		\frac{
			\sum\limits_{ p_i \in Halt_{iso}( \mathbf{L}_{\mathfrak{P}_{BB}(N)} , w, c )  } 
			{ | p_i | }
		}
		{N}
	}
	{ | \{ b_j \} | }
	\leq
	\lim_{ N \to \infty }
	\Omega(w,c) \lg(N)
	\] 
	
	\begin{proof}[Short proof]
		Note that, from the definition of language $ \mathbf{L}_{BB} $ in \ref{BdefL_BB}, we have that the choice of $ p_i $ is independent of topology. Therefore, from Definitions~\ref{BdefL_U} and~\ref{BdefHalt}, from the law of large numbers (as in Lemma~\ref{BlemmaSLLNandAIT}) and from Lemma~\ref{BlemmaGibbsandalgorithmicentropy}, we will have that 
		\begin{align*} \label{BstepTopologicalindependenceofp_i}
		& \lim_{ N \to \infty }
		\frac
		{
			\sum\limits_{\tiny b_j } 
			\frac{
				\sum\limits_{ p_i \in Halt_{iso}( \mathbf{L}_{\mathfrak{P}_{BB}(N)}, w, c )  } 
				{ | p_i | }
			}
			{N}
		}
		{ | \{ b_j \} | }
		= \\
		& =
		\lim_{ N \to \infty }
		\sum\limits_{ p_i \in Halt_{iso}( \mathbf{L_U}(N), w, c )  } 
		\frac
		{ |p_i| }
		{ 2^{ |p_i| }  } 
		\leq \\
		& \leq
		\lim_{ N \to \infty }
		\Omega(w,c) \lg(N)
		\end{align*}
	\end{proof}
\end{Blemma}

\begin{Blemma}[or extended Lemma~\ref{lemmaComplexityonBarHalt}]\label{BlemmaComplexityonBarHalt}
	Given a population $ \mathfrak{P}_{BB}(N) $ defined in \ref{BdefP_BB} , where $ p_i \in \mathbf{L_U} $ and $ P_{prot} \circ p_i = o_i \in \mathfrak{P}_{BB}(N) $, there is a constant $C_0$ such that
	\[
	\lim_{ N \to \infty }
	\frac{
		\sum\limits_{ b_j } 
		\frac{
			\sum\limits_{ p_i \in \overline{ Halt }_{iso} ( \mathbf{L}_{\mathfrak{P}_{BB}(N)}, w, c )  } 
			{ A ( \mathbf{U}(p_{iso} ( p_i , c )) ) }
		}
		{N}
	}
	{ | \{ b_j \} | }
	=
	C_0 ( 1 - \Omega(w,c) )
	\] 
	
	\begin{proof}[Short proof]
		Let $ A(0)=C_0 $. Remember Definitions~\ref{BdefP_BB}, \ref{BdefEAConN_BB}, \ref{BdefBarHalt}, and \ref{BdefOmegawc}. Since the population is sensitive to oracles, then
		\[
		A ( \mathbf{U}(p_{iso} ( p_i , c )) ) = A(0)=C_0
		\] 
		Therefore, we will have analogously to the proof of Lemma~\ref{BlemmaComplexityonHaltp_i} that
		\begin{align*} \label{BstepAlgorithmicentropyfromC_0}
		& 
		\lim_{ N \to \infty }
		\frac{
			\sum\limits_{ b_j } 
			\frac{
				\sum\limits_{ p_i \in \overline{ Halt }_{iso} ( \mathbf{L}_{\mathfrak{P}_{BB}(N)}, w, c )  } 
				{ A ( \mathbf{U}(p_{iso} ( p_i , c )) ) }
			}
			{N}
		}
		{ | \{ b_j \} | } 
		= \\
		& =
		\lim\limits_{ N \to \infty }
		\sum\limits_{ p_i \in \overline{ Halt }_{iso} ( \mathbf{L_U}(N), w, c )  } 
		{ \frac {A ( \mathbf{U}(p_{iso} ( p_i , c )) ) } { 2^{ |p_i| } } }
		= \\
		& =
		C_0 ( 1 - \Omega(w,c) )
		\end{align*}
	\end{proof}
\end{Blemma}

\begin{Blemma}[or extended Lemma~\ref{lemmaMinComplexityonDiffusion}]\label{BlemmaMinComplexityonDiffusion}
	Let $ \mathfrak{P}_{BB}(N) $ be a population in an arbitrary algorithmic network $ \mathfrak{N}_{BB} (N, f, t, \tau, j)=(G_t, \mathfrak{P}_{BB} (N),b_j) $ as defined in \ref{BdefN_BB} and \ref{BdefP_BB} .
	Let $ t_0 \leq t \leq t' \leq t_{ |\mathrm{T}(G_t)|-1 } $. 
	Let $ c \in \mathfrak{C_{BB}}$ be an arbitrary number of cycles where $ c_0 + t' + 1 \leq c $. 
	Then, there is a constant $C_2$ such that
	\[
	\frac{
		{
			\sum\limits_{ b_j  }
		}
		\frac{
			\sum\limits_{ p_i \in \mathfrak{P}_{BB}(N) } { A (\mathbf{U}(p_{net}^{b_j} ( o_i ,  c ) ))  }
		}
		{N}
	}
	{ | \{ b_j \} | }
	\geq
	( A_{max} - C_2 ) \,
	{ \tau_{\mathbf{E}(max)}( N,f,t,\tau ) }|_{t}^{t'}
	+ C_2
	\]
	
	\begin{proof}[Short proof]
		Let $ C_2 = \min \{ A(w) \mid \exists x \in \mathbf{L_U} ( \mathbf{U}( x ) = w ) \} $. 
		From the Definitions \ref{BdefPartitionmaxN_BB}, \ref{BdefPartitionmax} , \ref{BdefL_BB} and \ref{BdefA_max} we will have that
		\[
		\frac{
			{
				\sum\limits_{ b_j  }
			}
			\frac{
				\sum\limits_{ o_i \in \mathfrak{P}_{BB}(N) } { A (\mathbf{U}(p_{net}^{b_j} ( o_i ,  c ) ))  }
			}
			{N}
		}
		{ | \{ b_j \} | } =
		\]
		\begin{align} \label{BstepSumoffractionsmax}
		& =
		\frac
		{ \sum\limits_{ b_j  } 
			\left(
			\frac
			{ \sum\limits_{ o_i \in \mathbf{X}_{ { \tau_{max}( N,f,t,\tau, j ) }|_{t}^{t'} } } { A (\mathbf{U}(p_{net}^{b_j} ( o_i ,  c ) ))  }  }
			{ { \tau_{max}( N,f,t,\tau, j ) }|_{t}^{t'} N }
			{ \tau_{max}( N,f,t,\tau, j ) }|_{t}^{t'}
			\right)
		}
		{ | \{ b_j \} | } +
		\end{align}
		\[
		+ 
		\frac
		{ \sum\limits_{ b_j  } 
			\left(
			\frac
			{ \sum\limits_{ o_i \in \mathbf{X}_{ { \tau_{max}( N,f,t,\tau, j ) } |_{t'}^{ D(G_t, t) } } } { A (\mathbf{U}(p_{net}^{b_j} ( o_i ,  c ) ))  }  }
			{ { \tau_{max}( N,f,t,\tau, j ) } |_{t'}^{ D( G_t, t ) } N }
			{ \tau_{max}( N,f,t,\tau, j ) } |_{t'}^{ D(G_t, t) }
			\right)
		}
		{ | \{ b_j \} | }
		\geq
		\] 
		\[
		=
		\frac
		{ \sum\limits_{ b_j  } 
			\left( 
			( A_{max} - C_2 )
			{ \tau_{max}( N,f,t,\tau, j ) }|_{t}^{t'}
			+ 
			C_2
			\right)
		}
		{ | \{ b_j \} | }
		\] 
	\end{proof}
\end{Blemma}

\subsection{Theorem}

\begin{Bthm}[or extended Theorem~\ref{thmMain}]$\big($\textnormal{A lower bound for $ \mathrm{ EEAC_{BB} } $ from an arbitrary number of cycles}$\big)$\textbf{.} \label{BthmMain} 
	Let $ w \in \mathbf{L_U} $ be a network input. 	
	Let $ 0 < N \in \mathbb{N} $.
	Let $ \mathfrak{N}_{BB}(N,f,t,\tau,j) = ( G_t, \mathfrak{ P }_{BB}(N), b_j ) $ be well-defined.
	Let $ t_0 \leq t \leq t' \leq t_{ |\mathrm{T}(G_t)|-1 } $.
	Let 
	\[ \myfunc{c}{ \mathbb{N} } { \mathfrak{C_{BB}} } { x } { c(x)=y } \]
	\noindent be a total computable function where $ c(x) \geq c_0 + t'+1 $. 
	Then, we will have that:
	\begin{align*}
	& \lim\limits_{ N \to \infty }
	\mathbf{E}_{ \mathfrak{N}_{BB}(N,f,t,\tau) } 
	\left(
	{ {\displaystyle{\myDelta_{iso}^{net}} A} (o_i,c(x))} 
	\right)
	\geq
	\lim\limits_{ N \to \infty }
	\left( 
	{ \tau_{\mathbf{E}(max)}( N,f,t,\tau ) }|_{t}^{t'}
	-
	\Omega(w,c(x))
	\right)
	\lg(N) - \\
	& - \Omega(w,c(x)) \lg(x) - 2 \, \Omega(w,c(x))\lg(\lg(x)) - A(w) - C_5
	\end{align*}

	\begin{proof}[Short proof]
		We have from our hypothesis on function $c$ and from AIT that there is $C_c \in \mathbb{N}$ such that, for every $ x \in \mathbb{N} $, 
		\begin{equation}\label{BAITonfunctionc}
		A( c(x) ) \leq C_c + A(x)
		\end{equation}
		\noindent Let $ C_5 = C_c + C_L + C_1 + C_4 - C_0 $. 
		Therefore, from Definitions~\ref{BdefEEACN_BB}, \ref{BdefHalt} , \ref{BdefBarHalt} , \ref{BdefOmegawc} , \ref{BdefAveragePartitionmaxN_BB} , \ref{BdefU'} and Lemmas \ref{BlemmaComplexityp_i} , \ref{BlemmaComplexityonBarHalt} , \ref{BlemmaComplexityonHaltp_i} , \ref{BlemmaMinComplexityonDiffusion} , \ref{BlemmaSLLNandAIT} and Step \eqref{BAITonfunctionc} we will have that
		\begin{align}\label{BstepDefEEAC}
		&\mathbf{E}_{ \mathfrak{N}_{BB}(N,f,t,\tau) } 
		\left(
		{ {\displaystyle{\myDelta_{iso}^{net}} A} (o_i,c(x))} 
		\right)
		=  
		\frac{
			{
				\sum\limits_{ b_j  }
			}
			\frac{
				\sum\limits_{ o_i \in \mathfrak{P}_{BB}(N) } { A (\mathbf{U}(p_{net}^{b_j} ( o_i ,  c(x) ) )) - 
					A (\mathbf{U}(p_{iso} ( p_i ,  c(x) ))) }
			}
			{N}
		}
		{ | \{ b_j \} | }
		\geq
		\end{align}

		
		
		
		
		\[
		\geq
		\lim\limits_{ N \to \infty }
		\frac{
			{
				\sum\limits_{ b_j  }
			}
			\frac{
				\sum\limits_{ o_i \in \mathfrak{P}_{BB}(N) } { A (\mathbf{U}(p_{net}^{b_j} ( o_i ,  c(x)) )) }
			}
			{N}
		}
		{ | \{ b_j \} | }
		-
		\]
		\[
		-
		\left(
		\frac{ \sum\limits_{ b_j  }
			\left(
			\frac{
				\sum\limits_{p_i \in Halt_{iso}( L_{ \mathfrak{P}_{BB}(N) }, w, c(x)) } { C_1 + | p_i | + A(w) + A(c(x)) }
			}
			{N}
			+
			C_0 ( 1 - \Omega( w, c(x)) )
			\right)
		}
		{ | \{ b_j \} | }
		\right)	
		\geq
		\] 
		
		
		
		
		\[
		\geq
		\lim\limits_{ N \to \infty }
		\left( A_{max} - C_2 \right) 
		{ \tau_{\mathbf{E}(max)}( N,f,t,\tau ) }|_{t}^{t'}
		+
		C_2
		-
		\]
		\[
		-
		\left(
		\Omega(w,c(x)) \lg(N)
		+
		\Omega(w,c(x)) \big( C_1 + A(w) + A(c(x)) \big)
		+
		C_0 \big( 1 - \Omega( w, c(x)) \big)
		\right)
		\geq
		\]
		\[
		\geq
		\lim\limits_{ N \to \infty }
		\left( \lg(N) - C_4 - C_2 \right) 
		{ \tau_{\mathbf{E}(max)}( N,f,t,\tau ) }|_{t}^{t'}
		+
		C_2
		-
		\]
		\[
		-
		\left(
		\Omega(w,c(x)) \lg(N)
		+
		\Omega(w,c(x)) \big( C_1 + A(w) + A(c(x)) \big)
		+
		C_0 \big( 1 - \Omega( w, c(x)) \big)
		\right)
		\geq
		\]
		
		
		
		

		
		
		
		
		\[
		\geq
		\lim\limits_{ N \to \infty }
		\left( 
		{ \tau_{\mathbf{E}(max)}( N,f,t,\tau ) }|_{t}^{t'}
		-
		\Omega(w,c(x))
		\right)
		\lg(N)
		-
		\]
		\[
		- \Omega(w,c(x)) \lg(x) - 2 \, \Omega(w,c(x))\lg(\lg(x)) - C_5 - A(w)
		\]
	\end{proof}
	
	\subsubsection{Corollary}
	
	\begin{Bcorollary} $\big($\textnormal{A lower bound for $ \mathrm{ EEAC_{BB} } $ from a diffusion process as a function of cover time}$\big)$\textbf{.} \label{BcorMain} 
		Let $ w \in \mathbf{L_U} $ be a network input.
		Let $ 0 < N \in \mathbb{N} $.
		Let $ \mathfrak{N}_{BB}(N,f,t_z,\tau,j) = ( G_t, \mathfrak{ P }_{BB}(N), b_j ) $ be well-defined.
		Let $ t_z , \, t_{ z + f( N, t_z, \tau ) }  \in \mathrm{T}(G_t) $.
		Let $\tau \in \, ]0,1]$.
		Let $ \myfunc{c}{ \mathbb{N} } { \mathfrak{C_{BB}} } { x } { c(x)=y } $  be a total computable function where 
		\[ 
		c(z + f( N, t_z, \tau ) + 2) \geq c_0 + z + f( N, t_z, \tau ) + 2 
		\]
		Then, we will have that:
		%
		\begin{align*}
		& \lim\limits_{ N \to \infty } 
		\mathbf{E}_{ \mathfrak{N}_{BB}(N,f,t_z,\tau) } 
		\left(
		{ {\displaystyle{\myDelta_{iso}^{net}} A} (o_i, c( z + f( N, t_z, \tau ) + 2 ) )} 
		\right)
		\geq \\
		& \geq
		\lim\limits_{ N \to \infty }
		\left( 
		{ \tau_{\mathbf{E}(max)}( N,f,t_z,\tau ) }|_{ t_z }^{ t_{ z + f( N, t_z, \tau ) }  }
		-
		\Omega(w, c( z + f( N, t_z, \tau ) + 2 ))
		\right)
		\lg(N) - \\
		& - \Omega(w, c( z + f( N, t_z, \tau ) + 2 )) \lg( z + f( N, t_z, \tau ) + 2) - \\
		& - 2 \, \Omega(w, c( z + f( N, t_z, \tau ) + 2 ))\lg(\lg( z + f( N, t_z, \tau ) + 2)) - A(w) - C_5
		\end{align*}
		
		\begin{proof}
			Remember the Definition \ref{BdefN_BB} . We have that $ t_z \leq t_{ z + f( N, t_z, \tau ) } $. Note that $ z + f( N, t_z, \tau ) $ goes as an index of $ t_{z + f( N, t_z, \tau )} $, so that one must add $1$ on the right side of the condition $ c(x) \geq c_0 + t' +1 $ in Theorem~\ref{BthmMain} in order to make $ c_0 + t_{z + f( N, t_z, \tau )} +1 = c_0 + z + f( N, t_z, \tau ) + 2 $. Therefore, the conditions from Theorem~\ref{BthmMain} are satisfied. Then, the proof of Collorary~\ref{BcorMain} follows directly from replacing $t$ with $t_z$, $t'$ with $ t_{ z + f( N, t_z, \tau ) } $ and $x$ with  $z + f( N, t_z, \tau ) +2$ in Theorem~\ref{BthmMain} .
		\end{proof}
		
	\end{Bcorollary}
	
\end{Bthm}

\subsection{Theorem}

\begin{Bthm}[or extended Theorem~\ref{thmMainCentralTime}] $\big($\textnormal{Central time in reaching unlimited EEAC from the BBIG played by a randomly generated algorithmic population}$\big)$\textbf{.} \label{BthmMainCentralTime}
	Let $ w \in \mathbf{L_U} $ be a network input.
	Let $ 0 < N \in \mathbb{N} $.
	If there is $ 0 \leq z_0 \leq | \mathrm{T}(G_t) | -1 $ and $ \epsilon, \, \epsilon_2 > 0 $ such that 
	\[
	z_0 + f( N, t_{z_0}, \tau ) + 2 
	= 
	\mathbf{ O }
	\left( \frac
	{ N^{ C } }
	{ \lg(N) } 
	\right)
	\]
	\noindent where
	\[ 
	0
	<
	C = 
	\frac
	{
		{ \tau_{\mathbf{E}(max)}( N,f,t_{z_0},\tau ) }|_{ t_{z_0} }^{ t_{ z_0 + f( N, t_{z_0}, \tau ) }  }
		-
		\Omega(w, c_0 + z_0 + f( N, t_{z_0}, \tau ) + 2 )
		-
		\epsilon
	}
	{ \Omega(w,  c_0 + z_0 + f( N, t_{z_0}, \tau ) + 2  ) }
	\leq
	\frac{1}{ \epsilon_2 }
	\]
	Then, for every non-decreasing total computable function $ \myfunc{c}{ \mathbb{N} } { \mathfrak{C_{BB}} } { x } { c(x)=y } $, where $ t_{z_0}, \, t_{ z_0 + f( N, t_{z_0}, \tau ) } \in \mathrm{T}(G_t) $ and $ c(z_0 + f( N, t_{z_0}, \tau ) + 2) \geq c_0 + z_0 + f( N, t_{z_0}, \tau ) + 2 $ and $\mathfrak{N}_{BB}(N,f,t_{z_0},\tau,j) = ( G_t, \mathfrak{ P }_{BB}(N), b_j ) $ is well-defined, we will have that there are $ t_{cen_2}(c) $ and $ t_{cen_1}(c) $ such that
	\[
	t_{cen_2}(c) = t_{cen_1}(c) \leq t_{ z_0 }
	\]

	\begin{proof}[Short proof]
		Suppose that there is $ t_{z_0} \in \mathrm{T}(G_t) $, where $ 0 \leq {z_0} \leq | \mathrm{T}(G_t) | -1 $, and a small enough $\epsilon > 0$ such that
		\begin{equation}\label{BstepSuppositiononz_0}
		{z_0} + f( N, t_{z_0}, \tau ) + 2
		= 
		\mathbf{ O }
		\left( \frac
		{ N^{ C } }
		{ \lg(N) } 
		\right)
		\end{equation}
		\noindent where
		\[ 
		0
		<
		C = 
		\frac
		{
			{ \tau_{\mathbf{E}(max)}( N,f,t_{z_0},\tau ) }|_{ t_{z_0} }^{ t_{ z_0 + f( N, t_{z_0}, \tau ) }  }
			-
			\Omega(w,  c_0 + z_0 + f( N, t_{z_0}, \tau ) + 2 )
			-
			\epsilon
		}
		{ \Omega(w, c_0 + z_0 + f( N, t_{z_0}, \tau ) + 2 ) }
		\]
		\noindent Let \[ C' = \frac
		{
			{ \tau_{\mathbf{E}(max)}( N,f,t_{z_0},\tau ) }|_{ t_{z_0} }^{ t_{ z_0 + f( N, t_{z_0}, \tau ) }  }
			-
			\Omega(w, c( z_0 + f( N, t_{z_0}, \tau ) + 2 ))
			-
			\epsilon
		}
		{ \Omega(w, c( z_0 + f( N, t_{z_0}, \tau ) + 2 )) } \]
		\noindent Thus, since we are assuming $ c(z_0 + f( N, t_{z_0}, \tau ) + 2) \geq c_0 + z_0 + f( N, t_{z_0}, \tau ) + 2 $, we will have from Definitions~\ref{BdefOmegawc} and~\ref{BdefAveragePartitionmaxN_BB} that there is $ \epsilon_2 > 0  $ such that 
		\begin{equation} \label{BstepDominancebyC'}
		z_0 + f( N, t_{z_0}, \tau ) + 2
		= 
		\mathbf{ O }
		\left( \frac
		{ N^{ C' } }
		{ \lg(N) } 
		\right)
		\end{equation}
		\noindent where
		\[ 
		0
		\leq
		C' = 
		\frac
		{
			{ \tau_{\mathbf{E}(max)}( N,f,t_{z_0},\tau ) }|_{ t_{z_0} }^{ t_{ z_0 + f( N, t_{z_0}, \tau ) }  }
			-
			\Omega(w, c( z_0 + f( N, t_{z_0}, \tau ) + 2 ) )
			-
			\epsilon
		}
		{ \Omega(w, c( z_0 + f( N, t_{z_0}, \tau ) + 2 ) ) }
		\leq
		\frac{1}{ \epsilon_2 }
		\]		
		Therefore, from Corollary~\ref{BcorMain} and Step~\eqref{BstepDominancebyC'}, we will have that there is a constant $ C_6 $ such that
		\begin{equation}\label{BstepMainthmMainCentralTime}
		\lim\limits_{ N \to \infty } 
		\mathbf{E}_{ \mathfrak{N}_{BB}(N,f,t_{z_0},\tau) } 
		\left(
		{ {\displaystyle{\myDelta_{iso}^{net}} A} (o_i, c( z_0 + f( N, t_{z_0}, \tau ) + 2 ) ) } 
		\right)
		\geq
		\end{equation}
		\begin{align*}
		& \geq
		\lim\limits_{ N \to \infty }
		\left( 
		{ \tau_{\mathbf{E}(max)}( N,f,t_{z_0},\tau ) }|_{ t_{z_0} }^{ t_{ z_0 + f( N, t_{z_0}, \tau ) }  }
		-
		\Omega(w, c( z_0 + f( N, t_{z_0}, \tau ) + 2 ))
		\right)
		\lg(N) - \\
		& - \Omega(w, c( z_0 + f( N, t_{z_0}, \tau ) + 2 )) \lg( C_6 \, \frac
		{ N^{ C' } }
		{ \lg(N) } ) - \\
		& - 2 \, \Omega(w, c( z_0 + f( N, t_{z_0}, \tau ) + 2 ))\lg( \lg( C_6 \, \frac
		{ N^{ C' } }
		{ \lg(N) } ) ) - A(w) - C_5 
		\geq \\
		& \geq
		\lim\limits_{ N \to \infty }
		\left( 
		\epsilon
		\right)
		\lg(N) - 
		\lg( C_6 ) 
		- 2 \, \lg( \frac{1}{ \epsilon_2 } )
		- \lg(\lg( N ))
		- A(w) - C_5 
		=
		\infty
		\end{align*}	
		Then, directly from the Definitions \ref{BdefTimecentrality1} and \ref{BdefTimecentrality2} and Step \eqref{BstepMainthmMainCentralTime}, since $ t_{ z_0 } $ satisfies these definitions, we will have that
		\[
		t_{cen_2}(c) = t_{cen_1}(c) \leq t_{ z_0 }
		\] 
	\end{proof}
\end{Bthm}

\section{EEOE from a small diameter}\label{sectionDiameter}

We have previously established definitions and theorems for an abstract toy model of a simple optimization of average fitness through diffusion performed by a randomly generated population of interacting systems. It is shown in Theorem \ref{BthmMain} that there are conditions (e.g., in Theorem \ref{BthmMainCentralTime}) on the average diffusion density\footnote{ See Section \ref{sectionTau}. } and the number of cycles\footnote{ See Definition \ref{BdefCycle}. } that are sufficient to make the lower bound for the EEAC increase indefinitely when the population size grows toward infinity, giving rise to a phenomenon we call expected emergent open-endedness (EEOE). 

Following this purpose of theoretical investigation, we present in this section an application of our results using an important property in complex networks that fulfills these conditions mentioned to the above. Characterizing a small-diameter network, in which either the average shortest path length (also called as average geodesic distance) or the diameter (i.e., the maximum shortest path length) is ``small'' compared to the network size (i.e., when it is dominated by $\mathbf{O}(\lg(N))$ where $N$ is the network size), is a relevant topic in the current literature of complex networks \cite{Barabasi2003,Bollobas2004,Gershenson2011,Barabasi2009a}. Therefore, in this section, we aim to exemplify and reinforce the relevance of the study of models like the BBIG within the theory of algorithmic networks in relation to small-diameter networks (see Sections \ref{sectionIntro} and \ref{sectionAN} for more discussions). In fact, we will show that a small maximum shortest path length (that is, a small \emph{diameter}) gives rise to an emergent phenomenon: it increases the very potential\footnote{ In fact, a lower bound on the EEAC.} of average (or, in our case, expected) emergent open-endedness, should the respective network link a population of systems that mathematically behave like $ \mathfrak{N}_{BB}(N,f,t,\tau,j) = ( G_t, \mathfrak{ P }_{BB}(N), b_j ) $ (see Definition \ref{BdefN_BB} ). Thus, we not only present a model in which the \emph{small-diameter} phenomenon (SD) plays an central role in achieving the EEOE, but also suggest in a general way that EEOE in algorithmic networks might be related to why the SD would appear in real networks --- in which emergent complexity might be somehow advantageous to some degree.  

A family $ \mathbb{G}_{sm} $ of graphs will be defined in the same way as $ \mathbb{G} $ in Definition \ref{BdefFamilyG}. However, with two additional properties: first, there is a time instant on which every node can reach all nodes in a finite number of time intervals; second, taking this time instant as the starting point of diffusion, the graphs have a temporal diffusion diameter (see Definition \ref{BdefTemporaldiameter}) small compared to the size of the graph. The main idea is to ensure that a small-diameter phenomenon for dynamic networks, which relies on diffusion processes, eventually appears as the population size grows. Then, one can directly define an algorithmic network $ {\mathfrak{N}_{BB}}_{sm}(N,f,t,1,j) $ that is totally analogous to $ \mathfrak{N}_{BB}(N,f,t,\tau,j) $ except for replacing family $ \mathbb{G}(f, t, \tau) $ with $  \mathbb{G}_{sm}(f, t, 1) $.\footnote{ In Definition~\ref{BdefN_BB_sm}, as it is straightforward, we choose to omit a formal definition of $ {\mathfrak{N}_{BB}}_{sm}(N,f,t,1,j) $. } Thus, we will prove in Corollary \ref{BcorDiameter}, using in Theorem \ref{BthmMainCentralTime} the flexibility of Theorem \ref{BthmMain} and Corollary \ref{BcorMain}, that these conditions on $ \mathbb{G}_{sm}( f, t, 1 ) $ are sufficient for\footnote{ In a computably larger number of cycles compared to the temporal diffusion diameter.} existing a central time to trigger EEOE, which is determined by the smaller time instant from which the temporal diffusion diameter is minimal\footnote{ See Section~\ref{sectionTimecentralities}.}. 

\subsection{Definitions}

\begin{Bdefinition}\label{BdefFamilyG_sm}
	Let
	\begin{align*}
	& \mathbb{G}_{sm}(f,t,1) = \Big\{ G_t \mid  \, i = |\mathrm{V}(G_t)| \in \mathbb{N} \, \land \, f(i,t,1) = D(G_t,t)=\mathbf{O}\big( \lg( i ) \big) \, \land \\ 
	& \land \, \forall i \in \mathbb{N^*} \exists!G_t \in \mathbb{G}( f,t,\tau )  (\, |\mathrm{V}(G_t)|=i \,) \, \land \, \forall u \in \mathrm{V}(G_t) \exists x \in \mathbb{N} \big( x = d_t( G_t, t, u, 1 ) \big) \, \Big\}
	\end{align*}
	\noindent where 
	\[
	\myfunc{f}{ \mathbb{N^*} \times X \subseteq \mathrm{T}(G_t) \times Y \subseteq ]0,1] } { \mathbb{N} } { (x,t,\tau) } { y }
	\]

	\noindent be a \emph{family} of unique sized time-varying graphs which shares $ f(i,t,1) = D(G_t,t) = \mathbf{O}\big( \lg( i ) \big)$, where $i$ is the number of nodes, as a common property.
	
	\begin{Bremarknote}
		Note that the diameter might be much smaller indeed. For example, $ D(G_t,t)=\mathbf{O}\big( \frac{ \lg( N ) }{ \lg(\lg( N )) } \big) $ or $ D(G_t,t)=\mathbf{O}\big( \lg(\lg( N )) \big) $ are covered by Definition \ref{BdefFamilyG_sm} .
	\end{Bremarknote}	
	\begin{Bremarknote}
		As pointed in Notes~\ref{BnoteWeakeningfamilyG} and \ref{noteSLLN}, Corollary \ref{BcorDiameter} also holds with a weaker assumption on family $ \mathbb{G}_{sm}(f,t,1) $ in which
		\[
		\mathbf{P} \left[ \, \lim\limits_{ N \to \infty }
		\big( N = |\mathrm{V}(G_t)| \implies D(G_t,t)=\mathbf{O}\big( \lg( N ) \big) \, \big) 
		\, \right]
		=
		1
		\] 
	\end{Bremarknote}
\end{Bdefinition}

\begin{Bdefinition}\label{BdefN_BB_sm}
	Let the algorithmic network $ {\mathfrak{N}_{BB}}_{sm}(N,f,t,1,j) $ denote the same algorithmic network $ \mathfrak{N}_{BB}(N,f,t,1,j) $ defined in Definition~\ref{BdefN_BB}, except for replacing family $ \mathbb{G}(f, t, \tau) $ with family $ \mathbb{G}_{sm}(f,t,1) $.
\end{Bdefinition}

\subsection{Corollary}

\begin{Bcorollaryundersection}\label{BcorDiameter}
	Let $ w \in \mathbf{L_U} $ be a network input.
	Let $ 0 < N \in \mathbb{N} $. \\
	Let $ {\mathfrak{N}_{BB}}_{sm}(N,f,t_{z_0},1,j) = ( G_t, \mathfrak{ P }_{BB}(N), b_j ) $ be well-defined. Then, for every non-decreasing total computable function $ \myfunc{c}{ \mathbb{N} } { \mathfrak{C_{BB}} } { x } { c(x)=y } $ where $ t_{z_0} \in \mathrm{T}(G_t) $ and $ c(z_0 + f( N, t_{z_0}, 1 ) + 2) \geq c_0 + z_0 + f( N, t_{z_0}, 1 ) + 2 $, we will have that there are $ t_{cen_2}(c) $ and $ t_{cen_1}(c) $ such that
	\[
	t_{cen_2}(c) = t_{cen_1}(c) \leq t_{ z_0 }
	\] 
\end{Bcorollaryundersection}
	\begin{proof}[Main idea of the proof]
	The proof is analogous to Theorem \ref{BthmMainCentralTime}. However, the average singleton diffusion density assumes its maximum value $1$. And the small-diameter condition on family $ \mathbb{G}_{sm}(f,t,1) $ will assure that the conditions of Theorem \ref{BthmMainCentralTime} are met. First, remember that Corollary \ref{BcorMain} from Theorem \ref{BthmMain} only depends on function $ f( i, t, \tau) $ and not on it being equal to $CT(G_t, t, \tau)$. Second, the same is also valid for Theorem \ref{BthmMainCentralTime} that only depends on Corollary \ref{BcorMain} and also not on $ f( i, t, \tau) = CT(G_t, t, \tau) $.
	Therefore, we just replace $\tau$ with $1$ in Theorem \ref{BthmMainCentralTime} and the same proposition will also holds for
	\begin{align}\label{BstepMaincentraltimewithdiameter}
	f(i,t,1) = D(G_t,t)=\mathbf{O}\big( \lg( i ) \big) \end{align} 
	
	\end{proof}
	
	\begin{proof}
	Since $ G_t \in \mathbb{G}_{sm}(f,t,1) $, we will have from Definition \ref{BdefFamilyG_sm} that 
	\begin{align}\label{BstepFunctionfanddiameter}
	z_0 + f( N, t_{z_0}, 1 ) + 2
	=
	z_0 + D(G_t,t_{z_0}) + 2
	=
	z_0 + \mathbf{ O }\left( \lg(N) 
	\right) + 2
	=
	\mathbf{ O }\left( \lg(N) 
	\right)
	\end{align}
	\noindent where $ 0 \leq z_0 \leq | \mathrm{T}(G_t) | -1 $.
	As in Steps \eqref{BstepSuppositiononz_0} and \eqref{BstepDominancebyC'}, we will have from Definitions \ref{BdefTemporaldiameter}, \ref{BdefOmegawc}, \ref{BdefAveragePartitionmaxN_BB} and \ref{BdefFamilyG_sm} and Step \eqref{BstepFunctionfanddiameter} that there are constants $ \epsilon, \epsilon_2 > 0 $ such that
	\begin{align}\label{BstepCfromdiameter}
	& 
	0 
	<
	C = 
	\frac
	{
		{ \tau_{\mathbf{E}(max)}( N,f,t_{z_0},1 ) }|_{ t_{z_0} }^{ t_{ z_0 + f( N, t_{z_0}, 1 ) }  }
		-
		\Omega(w, c_0 + z_0 + f( N, t_{z_0}, 1 ) + 2 )
		-
		\epsilon
	}
	{ \Omega(w,  c_0 + z_0 + f( N, t_{z_0}, 1 ) + 2  ) }
	=
	\end{align}
	\[
	=
	\frac
	{
		1
		-
		\Omega(w, c_0 + z_0 + f( N, t_{z_0}, 1 ) + 2 )
		-
		\epsilon
	}
	{ \Omega(w,  c_0 + z_0 + f( N, t_{z_0}, 1 ) + 2  ) }
	\leq
	\frac{1}{ \epsilon_2 }
	\]
	We also have from Step \eqref{BstepFunctionfanddiameter} that if $\frac{1}{\epsilon_2}\geq C > 0$, then
	\begin{align}\label{BstepDiameterandSW}
	D(G_t,t_{z_0} ) = \mathbf{ O }
	\left( \lg(N) 
	\right) \implies z_0 + f( N, t_{z_0}, 1 ) + 2
	= 
	\mathbf{ O }
	\left(
	\frac
	{ N^{ C } }
	{ \lg(N) } 
	\right)
	\end{align}
	Then, from Steps \eqref{BstepCfromdiameter} and \eqref{BstepDiameterandSW}, we will have that the conditions in Theorem~\ref{BthmMainCentralTime} are satisfied. Therefore, there are $ t_{cen_2}(c) $ and $ t_{cen_1}(c) $ such that
	\[
	t_{cen_2}(c) = t_{cen_1}(c) \leq t_{ z_0 }
	\] 
\end{proof}
\begin{Bnote}\label{BnoteDiameterandstaticgraphs}
	As mentioned in discussions \ref{subsectionDiscussiononDGMD} and \ref{subsectionDiscussiontimecentralitiesandEEOE}, note that Corollary~\ref{BcorDiameter} also holds for static networks as in Definition~\ref{BdefStaticNetwork}. In this case, for strongly connected static graphs $G_s$ the usual diameter of graphs $diam(G_s)$ \cite{Bollobas2001} becomes equivalent to $D(G_s,t_0)$. Also, the notion of time centrality will be vacuous, since the graph's topology is always the same. Hence, for static graphs, the Corollary \ref{BcorDiameter} states only that the respective algorithmic network has the property of EEOE. 
\end{Bnote}

\section{Related Work and Future Research}\label{sectionRelatedandFuture}

The results presented in this article link algorithmic-informational properties with network-topological ones. Thus, we show that there may be graph  aspects that trigger EEOE (see Theorems \ref{BthmMain} and \ref{BthmMainCentralTime} and Corollaries \ref{BcorMain} and \ref{BcorDiameter}). These are embedded in the average singleton diffusion density (Definition \ref{BdefAveragePartitionmaxN_BB}) and function $c$ that takes a certain amount of cycles as input (see Corollaries \ref{BcorMain} and \ref{BcorDiameter} and Theorem \ref{BthmMainCentralTime}). In fact, our work suggests that the EEOE phenomenon benefits from topologies capable of faster spreading or flooding \cite{Clementi2010}. Although building variations of our models to fit real complex networks measurements for scale-free networks in \cite{Pastor-Satorras2001a,Pastor-Satorras2002} (through a Susceptible-Infected-Susceptible (SIS) variation of the BBIG model and the correspondence between prevalence and average diffusion density) will be a fruitful application of our models in network science, there would be no necessary advantage in some particular heterogeneity on the network structure.  Note that our result in Corollary~\ref{BcorDiameter} only uses the temporal diffusion diameter. We suggest that our mathematical results might indeed corroborate with computational experimental approaches in complex networks in \cite{Sole2002,Sole2004,Souza2012}. If an average optimization through diffusion corresponds somehow to expected emergent complexity, then graph sparseness (also discussed in \cite{Clementi2010}) and the small-diameter property may play a central role in achieving this optimization with minimum cost. A more accurate theoretical analysis of this potential corroborations is paramount for future research. 

With the same focus on complex networks analysis, an statistical or probabilistic information-theoretic counterpart of our work might be developed with further investigation on \cite{Sole2004, Dehmer2011, Mowshowitz2012}. The case more specific of an experimental approach centered on the relation between living systems and graph properties using statistical information is presented in \cite{Rashevsky1955,Walker2016,Kim2015}. 

Establishing relations between algorithmic information and statistical information like in \cite{Grunwald2004, Cover2005, Teixeira2011, Zenil2016} in the context of algorithmic networks is crucial to approximate, or find equivalences,  between measures from algorithmic information and measures from statistical information. In the particular case of an approach centered on emergence and self-organization, we suggest further investigation from \cite{Testa2000, Prokopenko2009, Fernandez2013}. Another interesting topic is whether the average (and the joint/global) emergent algorithmic complexity is sufficient to deal with synergistic and integrated information \cite{Williams2010, Griffith2012, Griffith2014, Oizumi2014, Wibral2015, Quax2017} or not,   as it directly gives a measure of irreducible information \cite{Li1997, Calude2002, Chaitin2004,Calude2009,Grunwald2008}. If one assumes the definition of synergy as the general phenomenon in which the whole system is irreducibly better in solving a common problem than the sum (or the union) of its parts taken separately, as is widely used in statistical information-theoretic approaches cited in the previous sentence, then there may be an immediate extension of (global or average local) emergent algorithmic complexity to (global or average local) synergistic algorithmic information. In order to prove the existence of such synergistic phenomenon directly from the present results, note that one may consider for example that the common goal of the algorithmic network is to increase the average fitness of each node/program (see Section~\ref{subsectionDiscussiononBBIM...}).  

The present model requires a hypercomputable function to deal with eventual nodes/programs that do not halt (see Definition \ref{BdefU'}). With the purpose of bringing this Busy Beaver imitation game to computer simulations or just to a fully computable model, the recursive relative incompressibility (or sub-incompressibility) as in \cite{Abrahao2015,Abrahao2016,Abrahao2016c} can be used to define a measure for resource-bounded algorithmic complexity \cite{Li1997,Allender2011}. Or one can use a measure through asymptotic approximations to algorithmic complexity as in \cite{Delahaye2012, Soler-Toscano2014, Zenil2016} in order to develop further investigations on algorithmic network complexity as an extension of an approach for graphs presented in \cite{Buhrman1999,Zenil2014,Zenil2017}. This way, one can aim at making the relative halting problem tractable while keeping the fitness function still relatively uncomputable.

There are further necessary investigations from a game-theoretical perspective. For example, in \cite{Mustafa2004}, it is shown a relation for reaching global consensus between a specific type of connectiveness of graphs and a local strategy (in our notation, protocols) for nodes. A transposition of this problem into an algorithmic network model would help to mathematically  understand global certification of best solutions (see Section~\ref{subsectionDiscussiononBBIM...}) from local ones for example. If fitness or payoff \cite{Hammerstein1994} is somehow connected to the complexity of the player's strategy, algorithmic networks would enable one to investigate game-theoretical consequences of randomly generated strategies for players without interaction in comparison to networked players \cite{Axelrod1997,Shoham2008} with global strategies, like the one in the BBIG. 

Another variation of the BBIG model in which the problem of consensus may be investigated is from synchronous distributed systems of failure-prone processors connected by small-diameter graphs as presented in \cite{Chlebus2006}. Adding the possibility of failure or crash to the BBIG model will require modifications on the information-sharing protocols (see Section \ref{sectionBusyBeaverGame} and Definition \ref{BdefIFP}) in order to solve the respective problem of consensus \cite{Chlebus2006} of choosing the fittest solution. A possible approach to this problem may be analogous to a BBIG variation under a SIS contagion model, as mentioned before in this section. A crash-prone system would introduce new difficulties on the plain diffusion that we have presented in this article; And a limit to the maximum expected number of failures in relation to the population size seems to be necessary to keep the phenomenon of expected emergent open-endedness in such algorithmic networks with crash-prone node/programs.   

We suggest that other relations to the fields we have mentioned in this Section and in Section \ref{sectionIntro} may configure a long-term horizon for future research. For example, one may relate the problem of evolvability versus robustness in evolutionary biology \cite{Wagner2008, Bateson2011} with intentional attack vulnerability and resilience to random failure in complex networks \cite{Greenbury2010}. To this end, one can investigate a variation of the BBIG model (see Section \ref{sectionBusyBeaverGame}) in which the diffused best solution may correspond to an increase in the distributed redundancy \cite{Randles2011}. Following this direction in the context of algorithmic networks, one may study why a network topology and an information-sharing protocol become relevant from an emergent open-ended evolutionary point of view that takes into account complexity and computational power in solving problems \cite{Hernandez-Orozco2016,Hernandez-Orozco2017}. While it is conjectured in complex networks theory that different types of real-world networks are indeed related by graph properties (e.g., the small-world phenomenon \cite{Kleinberg2000}), the theoretical approach we are developing may suggest (e.g., see Corollary \ref{BcorDiameter}) that they may be indeed \cite{Abrahao2016b}, when one assumes that problem solving, complexity increasing and networks are deeply connected through the concept of computation.

Regarding the general model of algorithmic networks (see Sections~\ref{subsectionDiscussionundersectionAN} and~\ref{subsectionDiscussiononGM} and Definition~\ref{BdefAN}), it stands as a computational model for networked distributed systems. As such, an important theoretical research would be studying equivalences and differences between algorithmic networks and other models, like Petri nets, population protocols, artificial neural networks, and cellular automata. For example, if an arbitrary model is Turing computable, then, by the universality of Turing machines, one can easily show that there is an algorithmic network that simulates\footnote{ In fact, with just a single node/program.} completely this other model. Furthermore, if every node/agent/neuron/cell in this model is Turing computable and the links between them can be represented univocally by a MultiAspect Graph (MAG)\footnote{ Or its respective variations, e.g., MultiAspect weighted graph, MultiAspect multigraph, MultiAspect hypergraph etc.}, one can also show that there is an algorithmic network that completely simulates this model within the same scale, i.e., with the same number of nodes/programs as the number of nodes/agents/neurons/cells. This may give rise to a definition of universality for distributed computable systems in which algorithmic networks are representative candidates. On the other hand, one may also investigate the opposite direction. That is, if another model can completely simulate algorithmic networks and if it can do within the same scale. Thus, the investigation of such theoretical models for distributed systems will be crucial to define a class of ``universal distributed systems'', i.e., the class of theoretical models for distributed computable systems that can simulate every other such model within the same scale. In addition, it may give rise to an hierarchy for networked distributed automata analogous to hierarchies (e.g., Chomsky–Schützenberger hierarchy) in automata theory and formal languages.

\section{Conclusion}\label{sectionConclusion}

In this article, we have presented definitions and theorems in order to investigate the problem of emergence of algorithmic complexity (or information) when a population of computable systems is networked. In particular, we have investigated conditions that enable the triggering of emergent open-endedness. That is, the conditions that trigger an unlimited increase of emergent complexity as the population size grows toward infinity.

First, we have introduced the general definition of an algorithmic network $ \mathfrak{N} $. This definition relies on a population of systems that runs on an arbitrarily chosen theoretical machine and relies on a MultiAspect Graph (MAG). This causes the population of systems to correspond to vertices in the respective MAG such that edges are communication channels that nodes/systems can use to send and receive information. 

Second, we have defined a particular mathematical model of algorithmic networks $ \mathfrak{N}_{BB} $ based on the simple imitation of the fittest neighbor: A type of algorithmic network that plays the Busy Beaver Imitation Game (BBIG). In this model, the randomly generated population of nodes/programs of a universal Turing machine is synchronous and non-halting nodes/programs always return the ``worst'' fitness/payoff. This population is networked by a time-varying graph. Hence, topological measures like cover time and temporal diffusion diameter can be promptly defined. Moreover, we have defined a graph-topological measure of diffusion (the average diffusion density) and the cycle-bounded conditional halting probability. These algorithmic networks $ \mathfrak{N}_{BB} $ can be seen as playing an optimization procedure where the whole is ``trying'' to increase the average fitness/payoff by diffusing through the network the best randomly generated solution.

Furthermore, we have presented the average (or expected) emergent algorithmic complexity of a node as the average (or expected) value obtained from the algorithmic complexity of what a node can do when running networked minus the algorithmic complexity of what the same node/program can do when running isolated. Thus, when this difference goes to infinity as the population size goes to infinity, this configures a phenomenon we have called as the \emph{average (or expected) emergent open-endedness}. This is a phenomenon akin to evolutionary open-endedness, but in the latter unlimited complexity occurs as successive mutations and natural selection apply over time. 

Since our network playing the BBIG is linked by a time-varying graph, a centrality for dynamic networks is defined for the optimum time instant to trigger expected emergent open-endedness in a way that it minimizes the amount of node cycles (i.e., communication rounds) needed to do this triggering. Moreover, we presented our main theorem proving that there is a lower bound for the expected emergent algorithmic complexity in algorithmic networks $ \mathfrak{N}_{BB}$ such that it depends on how much larger is the average diffusion density (in a given time interval) compared to the cycle-bounded conditional halting probability. This lower bound also depends on the parameter for which one is calculating the number of node cycles. In fact, we have proved a corollary showing that this parameter can be calculated for example from the cover time, so that our results hold even in the case of spending a computably larger number of node cycles compared to the cover time. Furthermore, from this corollary, we have proved that there are asymptotic conditions on the increasing diffusion power of the cover time as a function of the population size such that they ensure that there is a central time to trigger expected emergent open-endedness.

Additionally, we have made a small modification on the family of time-varying graphs of $ \mathfrak{N}_{BB} $ with the purpose of investigating what would happen if the networks were small-diameter. We replaced the cover time with the temporal diffusion diameter (i.e., the minimum number of time intervals to reach every node from any node) --- or the classical diameter (i.e., the maximum shortest path length) in the static case --- in the definition of this family of graphs. Indeed, we have shown that, in this case, the smallest temporal diffusion diameter (dominated by $\mathbf{O}(\lg(N))$, where $N$ is the network size) is sufficient for determining the central time to trigger expected emergent open-endedness.

Finally, we have suggested future research and discussed that our results might be related to problems in network science, information theory, computability theory, distributed computing, game theory, evolutionary biology, distributed systems, and automata theory.




\bibliographystyle{acm}
\bibliography{2.1-CompleteRefs-Felipe.bib}


\section{Appendix}\label{appendix}


\subsection{Lemma 1 Extended}

\begin{amslemma}[\ref{BlemmaSLLNandAIT}]\label{lemmaSLLNandAIT}
	\noindent \\
	Let $ \mathfrak{N}_{BB}( N, f, t, \tau, j ) = ( G_t, \mathfrak{P}_{BB}(N), b_j ) $ be an algorithmic network.\\
	Let $ N = | \mathbf{L}_{ \mathfrak{P}_{BB}(N) } | $.\\
	Then, on the average as $N$ grows, we will have that there is a constant $C_4$ such that
	\[ A_{max} \geq \lg(N) - C_4 \] \\
	
	\begin{noteunderlemma}\label{noteSLLN}
		Let $ \mathbf{P} \left[ \, X = a \, \right] $ be the usual notation for the probability of a random variable $X$ assuming value $a$. Or $ \mathbf{P} \left[ \, \mathit{statement \; S} \, \right] $ denote the probability of a $ statement \; S $ be true.
		Therefore, this theorem is formaly given by the strong law of large numbers \cite{Billingsley2012} as: there is a constant $C_4$ such that 
		
		\[
		\mathbf{P} \left[ \, \lim\limits_{ N \to \infty } A_{max} - \lg(N) + C_4  \geq 0 \, \right] = 1
		\] 
		
		In fact, the main results \ref{thmMain} and \ref{thmMainCentralTime} presented in this article can be translated into such probabilistic form by putting their last statements into the square brackets. For example, the reader is invited to check that, for finite subsets of $ \mathbf{L_U} $, the strong law of large numbers on a re-normalized probability distribution in the form $ \frac{1}{C \, 2^{|p|}} $ straightforwardly holds. Hence, in the limit as the size of this subset tends to $\infty$, a multiplicative term that tends to $1$ or an additive term that tends to $0$ would appear in Lemmas~\ref{lemmaSLLNandAIT}, \ref{lemmaComplexityonHaltp_i} and \ref{lemmaComplexityonBarHalt} and Theorem~\ref{thmMain}. For the sake of simplifying the notation and shortening the formulas, we have chosen to state our results without using this probabilistic form. \\
	\end{noteunderlemma}
	
	\begin{proof}
		\noindent \\
		From AIT, we know that the probability of occurring a program $ p \in \mathbf{L_U} $ is
		
		\begin{equation}
		\frac
		{ 1 }
		{ 2^{ | p | } }
		\end{equation} \\
		
		Let $\phi_{N}(p)$ be the frequency that $p$ occurs in a random sequence of size $N$. In the case, this random sequence is the randomly generated population $ \mathfrak{P} $.\\
		
		Define a Bernoulli trial on a random variable $ \mathbf{X} $ that assumes value $1$ \textit{iff} program $p$ occurs and assumes value $0$ \textit{iff} otherwise. Since this random sequence is identically distributed and/or define a binomial distribution where $ \sum\limits_{n=1}^{\infty} \frac{ Var[ X_n ] }{n^2} < \infty $, we will have from the \emph{strong law of large numbers} that 
		
		\begin{equation}\label{stepSLLN}
		\mathbf{P} \left[ \, \lim\limits_{ N \to \infty } \phi_{N}(p) = \frac{ 1 }{ 2^{ | p | } } \, \right]
		=
		1
		\end{equation} \\
		
		Therefore, when $N$ is large enough, one expects that $p$ occurs $ \frac{ N }{ 2^{ | p | } } $ times within $N$ random tries. That is, since $p$ was arbitrary, the probability distribution in $N$ random tries tends to match the program-size probability distribution on $\mathbf{L_U}$ when $N$ goes to $\infty$. \\
		
		Let $BB(k)$ be the Busy Beaver value for an arbitrary large enough $ k \in \mathbb{N} $ defined on machine $ \mathbf{U} $. We choose, for example, the definition of the Busy Beaver function in which $ BB(k) $ gives the largest value that a program $ p \in \mathbf{L_U} $, where $ |p| \leq k $, returns when running on machine $ \mathbf{U} $. \\
		
		From AIT~\cite{Chaitin2012,Chaitin2004,Li1997}, we know there are constants $C_{\Omega}, C_{BB} \geq 0 $ and a program $p_{BB(k)} \in \mathbf{L_U}$ such that, for every $ w \in \mathbf{L_U} $,
		
		\[
		\mathbf{U}( p_{BB(k)} \circ w) = \mathbf{U}( p_{BB(k)} ) = BB(k)
		\]
		
		\begin{equation}\label{stepAITmagic}
		\text{and}
		\end{equation}
		
		\[
		k - C_{\Omega}  \leq A(BB(k)) \leq | p_{BB(k)} | \leq k + C_{BB} 
		\]

		\begin{equation*}
		\text{and}
		\end{equation*}
		
		\[
		\forall x \geq BB(k) \; \Big( \, A(x)  \geq k - C_{\Omega} \, \Big)
		\] \\
		
		Since $k$ was arbitrary, let $ k = \lg(N) - C_{BB} $. \\
		
		From Step \eqref{stepSLLN} we will have that, when $N$ is large enough, one should expect that $ p_{BB( \lg(N) - C_{BB} )} $ occurs at least $ \frac{ N }{ 2^{ | p_{BB( \lg(N) - C_{BB} )} | } }  $ times where
		
		\[ 
		\frac{ N }{ 2^{ | p_{BB( \lg(N) - C_{BB} )} | } } 
		\geq
		\frac{ N }{ 2^{ \lg(N) } } 
		=
		1
		\]

		\noindent That is, from conditional probabilities,
		
		\begin{equation}
		\mathbf{P} \left[ \, \lim\limits_{ N \to \infty } \phi_{N}( p_{BB( \lg(N) - C_{BB} )} ) N \geq \frac{ N }{ 2^{ \lg(N) } } = 1 \, \right]
		\geq
		\end{equation}
		\[
		\geq \mathbf{P} \left[ \, \lim\limits_{ N \to \infty } \phi_{N}(p_{BB( \lg(N) - C_{BB} )})N = \frac{ N }{ 2^{ | p_{BB( \lg(N) - C_{BB} )} | } } \, \right]
		=
		1
		\] \\
		
		Let $ C_4 =  2 \, C_{\Omega} + 2 \, C_{BB} $.
		
		From Definitions \ref{BdefL_BB}, \ref{BdefA_max} and \ref{BdefN_BB} and Step \eqref{stepAITmagic} , since any node/program count as isolated from the network when $c=1$, we will have that, for large enough $ N $, \\
		\begin{center}if\end{center}
		
		\[
		\phi_{N}( p_{BB( \lg(N) - C_{BB} )} ) N \geq  1
		\]
		
		\begin{equation} \label{stepInequalitieschain}
		\text{then}
		\end{equation}
		
		\begin{align*}
		& A_{max} 
		\geq \\
		& \geq 
		\lg(N) - C_{BB} - C_{\Omega} 
		\geq \\
		& \geq
		\big| p_{BB( \lg(N) - C_{BB} )} \big| - C_{BB} - C_{\Omega}
		\geq \\
		& \geq 
		A(BB( \lg(N) - C_{BB} )) - C_{BB} - C_{\Omega}   = \\
		& =
		A(\mathbf{U}( P_{prot} \circ p_{BB( \lg(N) - C_{BB} )} )) - C_{BB} - C_{\Omega} 
		= \\
		&=
		A(\mathbf{U}(p_{BB( \lg(N) - C_{BB} )})) - C_{BB} - C_{\Omega}  = \\
		& =
		A(BB( \lg(N) - C_{BB} )) - C_{BB} - C_{\Omega}  \geq \\
		& \geq 
		\lg(N) - C_{BB} - C_{\Omega} - C_{BB} - C_{\Omega} 
		= \\
		& = 
		\lg(N) - C_4  
		\end{align*} \\
		
		Therefore, from conditional probabilities, we will have that
		
		\[
		\mathbf{P} \left[ \, \lim\limits_{ N \to \infty } A_{max} - \lg(N) + C_4  \geq 0 \, \right] \geq
		\]
		\[
		\geq
		\mathbf{P} \left[ \, \lim\limits_{ N \to \infty } \phi_{N}( p_{BB( \lg(N) - C_{BB} )} ) N \geq 1 \, \right]
		=
		1
		\]
	\end{proof}
\end{amslemma}

\noindent \\

\subsection{Lemma 2 Extended}

\begin{amslemma}[\ref{BlemmaComplexityp_i}] \label{lemmaComplexityp_i}
	Given a population $ \mathfrak{P}_{BB}(N) $ defined in \ref{BdefP_BB} , where $ p_i \in Halt_{iso}( \mathbf{L}_{ \mathfrak{P}_{BB}(N) }, w, c ) $ and $ P_{prot} \circ p_i \in \mathfrak{P}_{BB}(N) $ and $ N \in \mathbb{N^*} $ is arbitrary\footnote{ Where $ \mathbb{N^*} = \mathbb{N} \setminus \{ 0 \} $ is the set of strictly positive natural numbers.}, there is a constant $C_1$ such  that 
	
	\begin{align*}
	A (\mathbf{U}(p_{iso} ( p_i , c )))
	\leq
	C_1 + |p_i| + A(w) + A(c)
	\end{align*} \\
	
	\begin{noteunderlemma}
		Note that, from the Definition \ref{BdefEAConN_BB}, this result is independent of any topology in which $ \mathfrak{P}_{BB}(N) $ could be networked.
	\end{noteunderlemma}
	
	\begin{proof}
		\noindent \\
		Let $ N \in \mathbb{N^*} $ be arbitrary.
		Remember the definition of $ \mathbf{L}_{BB} $ in \ref{BdefL_BB} . And note that $p_i$ is a program in $ \mathbf{L_U} $. \\
		Then, from Definition~\ref{BdefHalt}, there is at least one program $p'$ such that 
		
		\[
		\mathbf{U}(p_{iso} ( p_i , c ))
		=			 \mathbf{U} (p' \circ p_i \circ w \circ c)
		\] 
		is a well-defined value for every $p_i \in Halt_{iso}( \mathbf{L}_{ \mathfrak{P}_{BB}(N) }, w, c ) $.
		
		Take the shortest such program $p'$ and let $ C_1 = |p'| + C_{\circ 4} $, where from AIT we know\footnote{ As in \cite{Li1997,Downey2010}, we have that, if $ \myfunc{ f_c }{ \{ 0 , 1 \}^* }{ \mathbb{N} }{ x }{ f_c( x ) } $ is a partial computable function and the value $ f_c( x ) $ is well-defined, then $ A( f_c( x ) ) \leq A( x ) + \mathbf{O}(1) $. } there is such constant $ C_{\circ 4} $ with
		\[
		A ( \mathbf{U} ( p' \circ p_i \circ w \circ c ) )
		\leq
		C_{\circ 4} + |p'| + |p_i| + A(w) + A(c)
		\]
		
		Then, we will have that
		
		\[
		A ( \mathbf{U}(p_{iso} ( p_i , c )) )			 =
		A ( \mathbf{U} (p' \circ p_i \circ w \circ c) )
		\leq
		C_1 + |p_i| + A(w) + A(c)
		\] \\
	\end{proof}
\end{amslemma}

\noindent \\

\subsection{Lemma 3 Extended}

\begin{amslemma}[\ref{BlemmaGibbsandalgorithmicentropy}] \label{lemmaGibbsandalgorithmicentropy}
	Given a population $ \mathfrak{P}_{BB}(N) $ defined in \ref{BdefP_BB} , where $ p_i \in \mathbf{L_U} $ and $ P_{prot} \circ p_i \in \mathfrak{P}_{BB}(N) $, then, from AIT and Gibb's (or Jensen's) inequality, we will have that
	
	\[
	\frac
	{1}
	{ \Omega(w,c) }
	\left( \lim_{ N \to \infty } \sum_{ p_i \in Halt_{iso}( \mathbf{L_U}(N), w, c ) } \frac
	{ |p_i| }
	{ 2^{ |p_i| } } \right)
	+ 
	\lg( \Omega(w,c) )
	\leq
	\lim_{ N \to \infty }
	\lg(\Omega(w,c) N) 
	\] \\
	
	\begin{proof}
		\noindent \\
		Since $\mathbf{L_U}$ is self-delimited, from AIT, the law of large numbers, and Definition \ref{BdefOmegawc}, we will have that
		\begin{equation}
		\lim_{ N \to \infty }
		\sum_{ p_i \in Halt_{iso}( \mathbf{L_U}(N), w, c ) } \frac
		{1}
		{ 2^{ |p_i| } }
		=
		\Omega(w,c)
		<
		1
		\end{equation} \\
		
		Hence,
		\begin{equation}\label{stepNormalizationofomegawc}
		\lim_{ N \to \infty }
		\sum_{ p_i \in Halt_{iso}( \mathbf{L_U}(N), w, c ) } \frac
		{1}
		{ \Omega(w,c) 2^{ |p_i| } }
		=
		\lim_{ N \to \infty }
		\sum_{ p_i \in Halt_{iso}( \mathbf{L_U}(N), w, c ) } \frac
		{1}
		{ 2^{ |p_i| + \lg( \Omega(w,c) ) } }
		=
		1
		\end{equation} \\
		
		Therefore, from Step \eqref{stepNormalizationofomegawc},

		\begin{equation}\label{stepNormalizedalgorithmicentropy}
		\lim_{ N \to \infty }
		\sum_{ p_i \in Halt_{iso}( \mathbf{L_U}(N), w, c ) } \frac
		{ |p_i| + \lg(\Omega(w,c)) }
		{ 2^{ |p_i| + \lg(\Omega(w,c)) } }
		= 
		\end{equation}
		
		\[
		= \lim_{ N \to \infty }
		\sum_{ p_i \in Halt_{iso}( \mathbf{L_U}(N), w, c ) } \frac
		{ |p_i| }
		{ 2^{ |p_i| } \Omega(w,c) }
		+
		\lim_{ N \to \infty }
		\sum_{ p_i \in Halt_{iso}( \mathbf{L_U}(N), w, c ) } \frac
		{ \lg(\Omega(w,c)) }
		{ 2^{ |p_i| } \Omega(w,c) }
		=
		\]
		
		\[
		= \frac
		{1}
		{ \Omega(w,c) }
		\left( \lim_{N \to \infty} \sum_{ p_i \in Halt_{iso}( \mathbf{L_U}(N), w, c ) } \frac
		{ |p_i| }
		{ 2^{ |p_i| } } \right)
		+ 
		\lg( \Omega(w,c) )
		\] \\
		
		Since Step \eqref{stepNormalizationofomegawc} holds, from Gibb's (or Jensen's) inequality \cite{Cover2005} \cite{MacKay2005}, we will have that
		
		\begin{equation}\label{stepGibbsinequality}
		\lim_{ N \to \infty }
		\sum_{ p_i \in Halt_{iso}( \mathbf{L_U}(N), w, c ) } \frac
		{ |p_i| + \lg(\Omega(w,c)) }
		{ 2^{ |p_i| + \lg(\Omega(w,c)) } }
		\leq
		\lim_{ N \to \infty }
		\lg( | Halt_{iso}( \mathbf{L_U}(N), w, c ) | ) 
		\end{equation} \\
		
		We also have by Definitions \ref{BdefOmegawc} and \ref{BdefL_U} and by the law of large numbers as in Lemma~\ref{lemmaSLLNandAIT} that\footnote{ In fact, due to the repetitions in the population, one may have that $ | \mathbf{L_U}(N) | < N $. Thus, our final results can be even improved in future work. }
		
		\begin{equation}\label{stepLimitofHalt}
		\lim_{ N \to \infty }
		\lg( | Halt_{iso}( \mathbf{L_U}(N), w, c ) | )
		\leq
		\lim_{ N \to \infty }
		\lg( \Omega(w,c) \left| \mathbf{L_U}(N)  \right| )
		\leq
		\lim_{ N \to \infty }
		\lg( \Omega(w,c) N )
		\end{equation} \\
		
		Therefore, from Steps \eqref{stepNormalizedalgorithmicentropy} , \eqref{stepGibbsinequality} and \eqref{stepLimitofHalt} we will have that
		
		\[
		\frac
		{1}
		{ \Omega(w,c) }
		\left( \lim_{ N \to \infty } \sum_{ p_i \in Halt_{iso}( \mathbf{L_U}(N), w, c ) } \frac
		{ |p_i| }
		{ 2^{ |p_i| } } \right)
		+ 
		\lg( \Omega(w,c) )
		\leq
		\lim_{ N \to \infty }
		\lg(\Omega(w,c) N)	
		\]
	\end{proof}
\end{amslemma}
\noindent \\

\subsection{Lemma 4 Extended}

\begin{amslemma}[\ref{BlemmaComplexityonHaltp_i}]\label{lemmaComplexityonHaltp_i}
	Given a population $ \mathfrak{P}_{BB}(N) $ defined in \ref{BdefP_BB} , where $ p_i \in \mathbf{L_U} $ and $ P_{prot} \circ p_i \in \mathfrak{P}_{BB}(N) $, then, from Lemma \ref{lemmaGibbsandalgorithmicentropy} we will have that
	
	\[
	\lim_{ N \to \infty }
	\frac{
		\sum\limits_{ b_j } 
		\frac{
			\sum\limits_{ p_i \in Halt_{iso}( \mathbf{L}_{\mathfrak{P}_{BB}(N)} , w, c )  } 
			{ | p_i | }
		}
		{N}
	}
	{ | \{ b_j \} | }
	\leq
	\lim_{ N \to \infty }
	\Omega(w,c) \lg(N)
	\] \\
	
	\begin{noteunderlemma}
		This theorem gives an upper bound on the algorithmic complexity of the randomly generated part of the elements of the population $ \mathfrak{ P }_{BB}(N) $. And it will be crucial to prove a lower bound for the emergent algorithmic complexity. However, the upper bound of Lemma \ref{lemmaComplexityonHaltp_i} is overestimated since a self-delimiting program-size probability distribution (i.e., a distribution defined by the algorithmic-informational Lebesgue measure\footnote{ See Definition~\ref{BdefRandompopulation}. }) is far from being classically uniform, which is the case where Gibb's equality applies on entropies.
	\end{noteunderlemma}

	\begin{proof}
		\noindent \\
		
		From the definition of language $ \mathbf{L}_{BB} $ in \ref{BdefL_BB} we have that $p_i$ is independent of any topology, so that
		
		\begin{equation} \label{stepTopologicalindependenceofp_i}
		\frac
		{
			\sum\limits_{\tiny b_j } 
			\frac{
				\sum\limits_{ p_i \in Halt_{iso}( \mathbf{L}_{\mathfrak{P}_{BB}(N)}, w, c )  } 
				{ | p_i | }
			}
			{N}
		}
		{ | \{ b_j \} | }
		=
		\end{equation}
		\[
		=
		\frac{
			\sum\limits_{ p_i \in Halt_{iso}( \mathbf{L}_{\mathfrak{P}_{BB}(N)}, w, c )  } 
			{ | p_i | }
		}
		{N}
		\] \\
		
		And from Definition~\ref{BdefL_U} we have that the ramdonly generated population $ \mathbf{L}_{\mathfrak{P}_{BB}(N)} $ tends to include all programs in $ \mathbf{L_U} $ in the limit.		
		Since one can define uniform probability measures in $\mathbf{L_U}$, then, by the Strong Law of Large Numbers, as in Lemma~\ref{lemmaSLLNandAIT}, we have that in the limit $\mathbf{L_U}(N)$ tends to follow the same distribution.
		
		Therefore, from Definition~\ref{BdefHalt}, we will have that\footnote{ Note that, since $ | Halt_{iso}( \mathbf{L}_{\mathfrak{P}_{BB}(N)} , w, c ) | < N $, then $ \sum\limits_{ p_i \in Halt_{iso}( \mathbf{L_U}, w, c )  } 
			\frac
			{ 1 }
			{ 2^{ |p_i| }  } = \Omega(w,c) < 1$. Then, $ \lim\limits_{ N \to \infty } \sum\limits_{ p_i \in Halt_{iso}( \mathbf{L_U}(N), w, c )  } 
			\frac
			{ |p_i| }
			{ 2^{ |p_i| }  } $ is not a proper Shannon entropy. Also note that $ \mathbf{L}_{\mathfrak{P}_{BB}(N)} $ may contain equal $ p_i $'s, since it is a population. However, in $ \mathbf{L_U}(N) $, each $p_i$ is unique, since it is a language and not a population. }
		
		\begin{equation} \label{stepSLLNonalgorithmicentropy}
		\lim_{ N \to \infty }
		{ \frac{1}{N} }
		\sum\limits_{ p_i \in Halt_{iso}( \mathbf{L}_{\mathfrak{P}_{BB}(N)} , w, c )  } { | p_i | }
		=
		\lim_{ N \to \infty }
		{ \frac{1}{N} }
		\sum\limits_{ p_i \in Halt_{iso}( \mathbf{L_U}(N), w, c )  } 
		\frac
		{ N \, |p_i| }
		{ 2^{ |p_i| }  }
		=
		\end{equation}
		\[
		=
		\lim_{ N \to \infty }
		\sum\limits_{ p_i \in Halt_{iso}( \mathbf{L_U}(N), w, c )  } 
		\frac
		{ |p_i| }
		{ 2^{ |p_i| }  } 
		\] \\
		
		From Lemma \ref{lemmaGibbsandalgorithmicentropy} we will have that
		
		\begin{equation} \label{stepLemmaGibbsandalgorithmicentropy}
		\lim_{ N \to \infty } \sum_{ p_i \in Halt_{iso}( \mathbf{L_U}(N), w, c ) } \frac
		{ |p_i| }
		{ 2^{ |p_i| } }
		\leq
		\lim_{ N \to \infty }
		\Omega(w,c) \lg( \Omega(w,c) N )		
		-
		\Omega(w,c) \lg( \Omega(w,c) )
		\end{equation} \\
		
		And
		
		\begin{equation} \label{stepSimplifyingLemmaGibbsandalgorithmicentropy}
		\Omega(w,c) \lg( \Omega(w,c) N )		
		-
		\Omega(w,c) \lg( \Omega(w,c) )
		=
		\end{equation}
		\[
		=
		\Omega(w,c) \lg( \Omega(w,c))
		+
		\Omega(w,c) \lg( N )
		-
		\Omega(w,c) \lg( \Omega(w,c) )=
		\]		
		\[
		=
		\Omega(w,c) \lg(N)
		\] \\
		
		Therefore, from Steps \eqref{stepTopologicalindependenceofp_i}, \eqref{stepSLLNonalgorithmicentropy} , \eqref{stepLemmaGibbsandalgorithmicentropy} and \eqref{stepSimplifyingLemmaGibbsandalgorithmicentropy}
		
		\[
		\lim_{ N \to \infty }
		\frac
		{
			\sum\limits_{\tiny b_j } 
			\frac{
				\sum\limits_{ p_i \in Halt_{iso}( \mathbf{L}_{\mathfrak{P}_{BB}(N)}, w, c )  } 
				{ | p_i | }
			}
			{N}
		}
		{ | \{ b_j \} | }
		=
		\]
		\[
		=
		\lim_{ N \to \infty } \sum_{ p_i \in Halt_{iso}( \mathbf{L_U}(N), w, c ) } \frac
		{ |p_i| }
		{ 2^{ |p_i| } }
		\leq
		\Omega(w,c) \lg(N)
		\]
	\end{proof}
\end{amslemma}

\noindent \\

\subsection{Lemma 5 Extended}

\begin{amslemma}[\ref{BlemmaComplexityonBarHalt}]\label{lemmaComplexityonBarHalt}
	Given a population $ \mathfrak{P}_{BB}(N) $ defined in \ref{BdefP_BB} , where $ p_i \in \mathbf{L_U} $ and $ P_{prot} \circ p_i = o_i \in \mathfrak{P}_{BB}(N) $, there is a constant $C_0$ such that
	
	\[
	\lim_{ N \to \infty }
	\frac{
		\sum\limits_{ b_j } 
		\frac{
			\sum\limits_{ p_i \in \overline{ Halt }_{iso} ( \mathbf{L}_{\mathfrak{P}_{BB}(N)}, w, c )  } 
			{ A ( \mathbf{U}(p_{iso} ( p_i , c )) ) }
		}
		{N}
	}
	{ | \{ b_j \} | }
	=
	C_0 ( 1 - \Omega(w,c) )
	\] \\
	
	\begin{noteunderlemma}
		Note that every lemma until here deals with the behavior of $ A ( \mathbf{U}(p_{iso} ( p_i , c )) ) $, so that they gave tools to obtain an upper bound for the expected algorithmic complexity of what each node can do when isolated. Since it is an upper bound for the algorithmic complexity of what each node can do when isolated and in the emergent algorithmic complexity it contributes negatively as defined in~\ref{BdefEAC}, then these results will help us to achieve a lower bound for the expected emergent algorithmic complexity. Furthermore, these results are independent of any topological feature that the algorithmic network $ \mathfrak{N}_{BB}( N, f, t, \tau, j ) $ might have. 
	\end{noteunderlemma}
	
	\begin{proof}
		\noindent \\
		As in Step \eqref{stepTopologicalindependenceofp_i}, $ A ( \mathbf{U}(p_{iso} ( p_i , c )) ) $ is independent of any topology and, therefore,
		
		\begin{equation} \label{stepTopologicalindependenceonBarHalt}
		\lim_{ N \to \infty }
		\frac{
			\sum\limits_{ b_j } 
			\frac{
				\sum\limits_{ p_i \in \overline{ Halt }_{iso} ( \mathbf{L}_{\mathfrak{P}_{BB}(N)}, w, c )  } 
				{ A ( \mathbf{U}(p_{iso} ( p_i , c )) ) }
			}
			{N}
		}
		{ | \{ b_j \} | }
		=
		\lim\limits_{ N \to \infty }
		\frac{
			\sum\limits_{ p_i \in \overline{ Halt }_{iso} ( \mathbf{L}_{\mathfrak{P}_{BB}(N)}, w, c )  } 
			{ A ( \mathbf{U}(p_{iso} ( p_i , c )) ) }
		}
		{N}
		\end{equation} \\
		
		From the Definition \ref{BdefBarHalt}, as in Step \eqref{stepSLLNonalgorithmicentropy}, we will have that
		
		\begin{equation} \label{stepSLLNonBarHalt}
		\lim\limits_{ N \to \infty }
		\frac{
			\sum\limits_{ p_i \in \overline{ Halt }_{iso} ( \mathbf{L}_{\mathfrak{P}_{BB}(N)}, w, c )  } 
			{ A ( \mathbf{U}(p_{iso} ( p_i , c )) ) }
		}
		{N}
		=
		\lim\limits_{ N \to \infty }
		\sum\limits_{ p_i \in \overline{ Halt }_{iso} ( \mathbf{L_U}(N), w, c )  } 
		{ \frac {A ( \mathbf{U}(p_{iso} ( p_i , c )) ) } { 2^{ |p_i| } } }
		\end{equation}	\\
		
		Now, let $ A(0)=C_0 $. \\
		Since $ p_i \in \overline{ Halt }_{iso} ( \mathbf{L_U}(N), w, c ) $ and $ \mathfrak{P}_{BB}(N) $ is sensitive to oracles as defined in \ref{BdefP_BB}, then, by the definition of $p_{iso}$ in \ref{BdefEAConN_BB}, we will have that, for every $p_i \in \overline{ Halt }_{iso} ( \mathbf{L_U}, w, c )$,
		
		\begin{equation}
		A ( \mathbf{U}(p_{iso} ( p_i , c )) ) = A(0) = C_0
		\end{equation} \\
		
		Therefore, by the definition of $\Omega(w,c)$ in \ref{BdefOmegawc},
		
		\begin{equation} \label{stepAlgorithmicentropyfromC_0}
		\lim\limits_{ N \to \infty }
		\sum\limits_{ p_i \in \overline{ Halt }_{iso} ( \mathbf{L_U}(N), w, c )  } 
		{ \frac {A ( \mathbf{U}(p_{iso} ( p_i , c )) ) } { 2^{ |p_i| } } }
		=
		\lim\limits_{ N \to \infty }
		\sum\limits_{ p_i \in \overline{ Halt }_{iso} ( \mathbf{L_U}(N), w, c )  } 
		{ \frac { C_0 } { 2^{ |p_i| } } }
		=
		\end{equation}
		\[ =
		\lim\limits_{ N \to \infty }
		C_0
		\sum\limits_{ p_i \in \overline{ Halt }_{iso} ( \mathbf{L_U}(N), w, c )  } 
		{ \frac {1} { 2^{ |p_i| } } }
		=
		C_0 ( 1 - \Omega(w,c) )
		\] \\
		
		And we conclude from Steps \eqref{stepTopologicalindependenceonBarHalt} , \eqref{stepSLLNonBarHalt} and \eqref{stepAlgorithmicentropyfromC_0} that
		
		\[
		\lim_{ N \to \infty }
		\frac{
			\sum\limits_{ b_j } 
			\frac{
				\sum\limits_{ p_i \in \overline{ Halt }_{iso} ( \mathbf{L}_{\mathfrak{P}_{BB}(N)}, w, c )  } 
				{ A ( \mathbf{U}(p_{iso} ( p_i , c )) ) }
			}
			{N}
		}
		{ | \{ b_j \} | }
		=
		\lim\limits_{ N \to \infty }
		\sum\limits_{ p_i \in \overline{ Halt }_{iso} ( \mathbf{L_U}(N), w, c )  } 
		{ \frac {A ( \mathbf{U}(p_{iso} ( p_i , c )) ) } { 2^{ |p_i| } } }
		= \]
		\[
		= C_0 ( 1 - \Omega(w,c) )
		\]
	\end{proof}
\end{amslemma}

\noindent \\

\subsection{Lemma 6 Extended}

\begin{amslemma}[\ref{BlemmaMinComplexityonDiffusion}]\label{lemmaMinComplexityonDiffusion}
	Let $ \mathfrak{P}_{BB}(N) $ be a population in an arbitrary algorithmic network $ \mathfrak{N}_{BB} (N, f, t, \tau, j)=(G_t, \mathfrak{P}_{BB} (N),b_j) $ as defined in \ref{BdefN_BB} and \ref{BdefP_BB}. \\
	Let $ t_0 \leq t \leq t' \leq t_{ |\mathrm{T}(G_t)|-1 } $. \\ 
	Let $ c \in \mathfrak{C_{BB}}$ be an arbitrary number of cycles where $ c_0 + t' + 1 \leq c $. \\ 
	Then, there is a constant $C_2$ such that
	
	\[
	\frac{
		{
			\sum\limits_{ b_j  }
		}
		\frac{
			\sum\limits_{ p_i \in \mathfrak{P}_{BB}(N) } { A (\mathbf{U}(p_{net}^{b_j} ( o_i ,  c ) ))  }
		}
		{N}
	}
	{ | \{ b_j \} | }
	\geq
	( A_{max} - C_2 ) \,
	{ \tau_{\mathbf{E}(max)}( N,f,t,\tau ) }|_{t}^{t'}
	+ C_2
	\] \\
	
	\begin{proof}
		\noindent \\
		
		Let  $ \mathbf{X}_{ { \tau_{max}( N,f,t,\tau, j ) }|_{t}^{t'} } $ denote the set of nodes/programs that belong to fraction $ { \tau_{max}( N,f,t,\tau, j ) }|_{t}^{t'} $ as defined in \ref{BdefPartitionmaxN_BB}. Hence, 
		\[ 
		| \mathbf{X}_{ { \tau_{max}( N,f,t,\tau, j ) }|_{t}^{t'} } | = N { \tau_{max}( N,f,t,\tau, j ) }|_{t}^{t'}
		\] \\
		
		Let $ C_2 = \min \{ A(w) \mid \exists x \in \mathbf{L_U} ( \mathbf{U}( x ) = w ) \} $. \footnote{ Note that depending on the choice of the programming language one may have $ C_2 \leq A(0) $.} \\
		
		From the Definition \ref{BdefPartitionmaxN_BB} of $ { \tau_{max}( N,f,t_i,\tau, j ) }|_{t}^{t'} $ we will have that
		
		\[
		\frac{
			{
				\sum\limits_{ b_j  }
			}
			\frac{
				\sum\limits_{ o_i \in \mathfrak{P}_{BB}(N) } { A (\mathbf{U}(p_{net}^{b_j} ( o_i ,  c ) ))  }
			}
			{N}
		}
		{ | \{ b_j \} | } =
		\]
		\begin{align} \label{stepSumoffractionsmax}
		& =
		\frac
		{ \sum\limits_{ b_j  } 
			\left(
			\frac
			{ \sum\limits_{ o_i \in \mathbf{X}_{ { \tau_{max}( N,f,t,\tau, j ) }|_{t}^{t'} } } { A (\mathbf{U}(p_{net}^{b_j} ( o_i ,  c ) ))  }  }
			{ { \tau_{max}( N,f,t,\tau, j ) }|_{t}^{t'} N }
			{ \tau_{max}( N,f,t,\tau, j ) }|_{t}^{t'}
			\right)
		}
		{ | \{ b_j \} | } +
		\end{align}
		\[
		+ 
		\frac
		{ \sum\limits_{ b_j  } 
			\left(
			\frac
			{ \sum\limits_{ o_i \in \mathbf{X}_{ { \tau_{max}( N,f,t,\tau, j ) } |_{t'}^{ D(G_t, t) } } } { A (\mathbf{U}(p_{net}^{b_j} ( o_i ,  c ) ))  }  }
			{ { \tau_{max}( N,f,t,\tau, j ) } |_{t'}^{ D( G_t, t ) } N }
			{ \tau_{max}( N,f,t,\tau, j ) } |_{t'}^{ D(G_t, t) }
			\right)
		}
		{ | \{ b_j \} | }
		\] \\
		
		Note that in the case the temporal diameter $ D(G_t, t) $ is not well-defined one can replace $ { \tau_{max}( N,f,t,\tau, j ) } |_{t'}^{ D(G_t, t) } $ with fraction 
		
		\[
		\frac 
		{ \left| \mathfrak{P}_{BB}(N) \setminus \mathbf{X}_{ { \tau_{max}( N,f,t,\tau, j ) }|_{t}^{t'} } \right| }
		{N}
		\]
		
		\noindent in the rest of this proof below without loss of generality. \\
		
		Since fraction $ { \tau_{max}( N,f,t,\tau, j ) }|_{t}^{t'} $ is centered on a node from which $ A_{max} $ is diffused --- see Definition \ref{BdefPartitionmax} ---, by Definitions \ref{BdefL_BB} and \ref{BdefA_max} we will have that
		
		\begin{equation} \label{stepA_maxandfractionmax}
		\frac
		{ \sum\limits_{ o_i \in \mathbf{X}_{ { \tau_{max}( N,f,t,\tau, j ) }|_{t}^{t'} } } { A (\mathbf{U}(p_{net}^{b_j} ( o_i ,  c ) ))  }  }
		{ { \tau_{max}( N,f,t,\tau, j ) }|_{t}^{t'} N }
		{ \tau_{max}( N,f,t,\tau, j ) }|_{t}^{t'}
		\geq
		A_{max} \, { \tau_{max}( N,f,t,\tau, j ) }|_{t}^{t'}
		\end{equation} \\
		
		\noindent
		and, analogously, the following always holds despite on which node fraction  $ { \tau_{max}( N,f,t,\tau, j ) } |_{t'}^{ D(G_t, t) } $ is centered and whenever it starts to diffuse
		
		\begin{equation} \label{stepC_2andfractionmax}
		\frac
		{ \sum\limits_{ o_i \in \mathbf{X}_{ { \tau_{max}( N,f,t,\tau, j ) } |_{t'}^{D(G_t,t)} } } { A (\mathbf{U}(p_{net}^{b_j} ( o_i ,  c ) ))  }  }
		{ { \tau_{max}( N,f,t,\tau, j ) } |_{t'}^{D(G_t,t)} N }
		{ \tau_{max}( N,f,t,\tau, j ) } |_{t'}^{D(G_t,t)}
		\geq
		C_2 { \tau_{max}( N,f,t,\tau, j ) } |_{t'}^{D(G_t,t)}
		\end{equation} \\
		
		Thus, since we have that $ { \tau_{max}( N,f,t,\tau, j ) }|_{t}^{t'} + { \tau_{max}( N,f,t,\tau, j ) } |_{t'}^{ D(G_t, t) } = 1 $, then, by Steps \eqref{stepSumoffractionsmax}, \eqref{stepA_maxandfractionmax} and \eqref{stepC_2andfractionmax}, we will have that
		
		\begin{align}
		& \frac
		{ \sum\limits_{ b_j  } 
			\left(
			\frac
			{ \sum\limits_{ o_i \in \mathbf{X}_{ { \tau_{max}( N,f,t,\tau, j ) }|_{t}^{t'} } } { A (\mathbf{U}(p_{net}^{b_j} ( o_i ,  c ) ))  }  }
			{ { \tau_{max}( N,f,t,\tau, j ) }|_{t}^{t'} N }
			{ \tau_{max}( N,f,t,\tau, j ) }|_{t}^{t'}
			\right)
		}
		{ | \{ b_j \} | } +
		\end{align}
		\[
		+ 
		\frac
		{ \sum\limits_{ b_j  } 
			\left(
			\frac
			{ \sum\limits_{ o_i \in \mathbf{X}_{ { \tau_{max}( N,f,t,\tau, j ) } |_{t'}^{ D(G_t, t) } } } { A (\mathbf{U}(p_{net}^{b_j} ( o_i ,  c ) ))  }  }
			{ { \tau_{max}( N,f,t,\tau, j ) } |_{t'}^{ D( G_t, t ) } N }
			{ \tau_{max}( N,f,t,\tau, j ) } |_{t'}^{ D(G_t, t) }
			\right)
		}
		{ | \{ b_j \} | } \geq
		\]
		\[
		\geq
		\frac
		{ \sum\limits_{ b_j  } 
			\left( 
			A_{max}
			{ \tau_{max}( N,f,t,\tau, j ) }|_{t}^{t'}
			+ 
			C_2
			{ \tau_{max}( N,f,t,\tau, j ) } |_{t'}^{ D(G_t, t) }
			\right)
		}
		{ | \{ b_j \} | } =	
		\]
		\[
		=
		\frac
		{ \sum\limits_{ b_j  } 
			\left( 
			( A_{max} - C_2 )
			{ \tau_{max}( N,f,t,\tau, j ) }|_{t}^{t'}
			+ 
			C_2
			\right)
		}
		{ | \{ b_j \} | }
		\]
	\end{proof}
\end{amslemma}

\noindent \\

\noindent \\

\subsection{Theorem 1 Extended}


\begin{thm}[\ref{BthmMain}]\label{thmMain}
	\noindent \\
	
	Let $ w \in \mathbf{L_U} $ be a network input. \\
	
	Let $ 0 < N \in \mathbb{N} $.\\
	
	Let $ \mathfrak{N}_{BB}(N,f,t,\tau,j) = ( G_t, \mathfrak{ P }_{BB}(N), b_j ) $ be well-defined.\\
	
	Let $ t_0 \leq t \leq t' \leq t_{ |\mathrm{T}(G_t)|-1 } $.\\
	
	Let $ \myfunc{c}{ \mathbb{N} } { \mathfrak{C_{BB}} } { x } { c(x)=y } $  be a total computable function where $ c(x) \geq c_0 + t'+1 $. \\

	Then, we will have that:
	
	\begin{align*}
	& \lim\limits_{ N \to \infty }
	\mathbf{E}_{ \mathfrak{N}_{BB}(N,f,t,\tau) } 
	\left(
	{ {\displaystyle{\myDelta_{iso}^{net}} A} (o_i,c(x))} 
	\right)
	\geq
	\lim\limits_{ N \to \infty }
	\left( 
	{ \tau_{\mathbf{E}(max)}( N,f,t,\tau ) }|_{t}^{t'}
	-
	\Omega(w,c(x))
	\right)
	\lg(N) - \\
	& - \Omega(w,c(x)) \lg(x) - 2 \, \Omega(w,c(x))\lg(\lg(x)) - A(w) - C_5
	\end{align*}
	
	\noindent \\
	
	\begin{noteunderthm}
		Thus, note that, for example, for bigger enough values of $x$ compared to $N$ one can make this lower bound always negative. One of the main ideas behind forthcoming results in this paper is to find optimal conditions where this lower bound is not only positive, but also goes to $\infty$.
	\end{noteunderthm}
	
	\begin{noteunderthm}
		Note that this lower bound for the expected emergent algorithmic complexity is dependent on the value in the domain of the function $c$ and not on function $c$ itself, even if it grows fast. And it holds as long as c is a total computable function. In fact, one may want to obtain this theorem for fixed values of $c$ in which it is not a function but an arbitrary value. And the same result also holds in this case. The reader is invited to check that a simple substitution of $c(x)$ for $c$ inside $ \Omega(w,c(x)) $ and of $x$ inside the logarithms for $c$ is enough\footnote{ Besides a slightly different constant $C_5$.}.
	\end{noteunderthm}
	
	\begin{noteunderthm}
		The same result also holds if only one possible function $b_j$ is defined for each member of the family $\mathbb{G}(f,t,\tau)$. This way, only one function $b_j$ will be taken into account within the sum in order to give the mean. Thus, in this case one can replace $\tau_{\mathbf{E}(max)}$ with $\tau_{max}$ not only in Theorem \ref{BthmMain}, but also in \ref{BcorMain} and \ref{BthmMainCentralTime} . Such variation of these theorems becomes useful when one has algorithmic networks $\mathfrak{N}_{BB}(N,f,t,\tau,j)$ built upon a historical population-size growth in which each new node/program is linked (or not) to the previous existing algorithmic network. 
	\end{noteunderthm}
	
	\noindent \\
	
	\begin{proof}
		\noindent \\
		The proof will follow from Steps~\eqref{stepDefEEAC} and~\eqref{stepMainthm1} below.\\
		
		We have from our hypothesis on function $c$ and from AIT that there is $C_c \in \mathbb{N}$ such that, for every $ x \in \mathbb{N} $, 
		
		\begin{equation}\label{AITonfunctionc}
		A( c(x) ) \leq C_c + A(x)
		\end{equation} \\
		
		Let $ C_5 = C_c + C_L + C_1 + C_4 + C_0 $. \\
		
		Note that, as in Step \eqref{stepSLLNonalgorithmicentropy}, we will have from Definition~\ref{BdefOmegawc} that
		\begin{equation}\label{stepLLNinHaltisothm1}
		\lim\limits_{ N \to \infty }
		\frac{\left| Halt_{iso}( L_{ \mathfrak{P}_{BB}(N) }, w, c(x)) \right|}
		{ N }
		=
		\lim\limits_{ N \to \infty }
		\frac{1}{N}
		\sum\limits_{ p_i \in Halt_{iso}( \mathbf{L_U}(N), w, c(x)) }
		\frac{N}
		{2^{ | p_i | }}
		=
		\Omega( w, c(x) )
		\end{equation} \\
		
		From Definition \ref{BdefEEACN_BB}, we have that the expected emergent algorithmic complexity of a node/program for $ \mathfrak{N}_{BB}(N,f,t,\tau,j) = ( G_t, \mathfrak{ P }_{BB}(N), b_j ) $, where $ 0 < j \leq |\{ b_j \}| $ is given by
		
		\begin{align}\label{stepDefEEAC}
		&\mathbf{E}_{ \mathfrak{N}_{BB}(N,f,t,\tau) } 
		\left(
		{ {\displaystyle{\myDelta_{iso}^{net}} A} (o_i,c(x))} 
		\right)
		=  
		\frac{
			{
				\sum\limits_{ b_j  }
			}
			\frac{
				\sum\limits_{ o_i \in \mathfrak{P}_{BB}(N) } { A (\mathbf{U}(p_{net}^{b_j} ( o_i ,  c(x) ) )) - 
					A (\mathbf{U}(p_{iso} ( p_i ,  c(x) ))) }
			}
			{N}
		}
		{ | \{ b_j \} | }
		\end{align} \\
		
		And, from Definitions \ref{BdefHalt} , \ref{BdefBarHalt} , \ref{BdefOmegawc} , \ref{BdefAveragePartitionmaxN_BB} , \ref{BdefU'} and Lemmas \ref{lemmaComplexityp_i} , \ref{lemmaComplexityonBarHalt} , \ref{lemmaComplexityonHaltp_i} , \ref{lemmaMinComplexityonDiffusion} , \ref{lemmaSLLNandAIT} and Steps \eqref{AITonfunctionc}  and \eqref{stepLLNinHaltisothm1}, we will have that
		
		\begin{align}\label{stepMainthm1}
		\lim\limits_{ N \to \infty }
		\frac{
			{
				\sum\limits_{ b_j  }
			}
			\frac{
				\sum\limits_{ o_i \in \mathfrak{P}_{BB}(N) } { A (\mathbf{U}(p_{net}^{b_j} ( o_i ,  c(x)) )) - 
					A (\mathbf{U}(p_{iso} ( p_i ,  c(x)))) }
			}
			{N}
		}
		{ | \{ b_j \} | } 
		= 
		\end{align}
		\[
		=
		\lim\limits_{ N \to \infty }
		\frac{
			{
				\sum\limits_{ b_j  }
			}
			\frac{
				\sum\limits_{ o_i \in \mathfrak{P}_{BB}(N) } { A (\mathbf{U}(p_{net}^{b_j} ( o_i ,  c(x)) )) }
			}
			{N}
		}
		{ | \{ b_j \} | }
		-
		\]
		\[
		-
		\left(
		\frac{ \sum\limits_{ b_j  }
			\left(
			\frac{
				\sum\limits_{ p_i \in Halt_{iso}( L_{ \mathfrak{P}_{BB}(N) }, w, c(x)) } {  
					A (\mathbf{U}(p_{iso} ( p_i ,  c(x)))) }
			}
			{N}
			+
			\frac{
				\sum\limits_{ p_i \in \overline{Halt}_{iso}( L_{ \mathfrak{P}_{BB}(N) }, w, c(x)) } { 
					A (\mathbf{U}(p_{iso} ( p_i ,  c(x)))) }
			}
			{N}
			\right)
		}
		{ | \{ b_j \} | }
		\right)
		=
		\]
		\[
		=
		\lim\limits_{ N \to \infty }
		\frac{
			{
				\sum\limits_{ b_j  }
			}
			\frac{
				\sum\limits_{ o_i \in \mathfrak{P}_{BB}(N) } { A (\mathbf{U}(p_{net}^{b_j} ( o_i ,  c(x)) )) }
			}
			{N}
		}
		{ | \{ b_j \} | }
		-
		\]
		\[
		-
		\left(
		\frac{ \sum\limits_{ b_j  }
			\left(
			\frac{
				\sum\limits_{p_i \in Halt_{iso}( L_{ \mathfrak{P}_{BB}(N) }, w, c(x)) } {  
					A (\mathbf{U}(p_{iso} ( p_i ,  c(x)))) }
			}
			{N}
			\right)
		}
		{ | \{ b_j \} | }
		+
		C_0 ( 1 - \Omega( w, c(x)) )
		\right)	
		\geq
		\] 
		\[
		\geq
		\lim\limits_{ N \to \infty }
		\frac{
			{
				\sum\limits_{ b_j  }
			}
			\frac{
				\sum\limits_{ o_i \in \mathfrak{P}_{BB}(N) } { A (\mathbf{U}(p_{net}^{b_j} ( o_i ,  c(x)) )) }
			}
			{N}
		}
		{ | \{ b_j \} | }
		-
		\]
		\[
		-
		\left(
		\frac{ \sum\limits_{ b_j  }
			\left(
			\frac{
				\sum\limits_{p_i \in Halt_{iso}( L_{ \mathfrak{P}_{BB}(N) }, w, c(x)) } { C_1 + | p_i | + A(w) + A(c(x)) }
			}
			{N}
			+
			C_0 ( 1 - \Omega( w, c(x)) )
			\right)
		}
		{ | \{ b_j \} | }
		\right)	
		=
		\] 
		\[
		=
		\lim\limits_{ N \to \infty }
		\frac{
			{
				\sum\limits_{ b_j  }
			}
			\frac{
				\sum\limits_{ o_i \in \mathfrak{P}_{BB}(N) } { A (\mathbf{U}(p_{net}^{b_j} ( o_i ,  c(x)) )) }
			}
			{N}
		}
		{ | \{ b_j \} | }
		-
		\]
		\[
		-
		\left(
		\frac{ \sum\limits_{ b_j  }
			\left(
			\frac{
				\sum\limits_{p_i \in Halt_{iso}( L_{ \mathfrak{P}_{BB}(N) }, w, c(x)) } { | p_i | }
			}
			{N}
			\right)
		}
		{ | \{ b_j \} | }
		+
		\Omega(w,c(x)) \big( C_1 + A(w) + A(c(x)) \big)
		+
		C_0 \big( 1 - \Omega( w, c(x)) \big)
		\right)
		\geq
		\] 
		\[
		\geq
		\lim\limits_{ N \to \infty }
		\frac{
			{
				\sum\limits_{ b_j  }
			}
			\frac{
				\sum\limits_{ o_i \in \mathfrak{P}_{BB}(N) } { A (\mathbf{U}(p_{net}^{b_j} ( o_i ,  c(x)) )) }
			}
			{N}
		}
		{ | \{ b_j \} | }
		-
		\]
		\[
		-
		\bigg(
		\Omega(w,c(x)) \lg(N)
		+
		\Omega(w,c(x)) \big( C_1 + A(w) + A(c(x)) \big)
		+
		C_0 \big( 1 - \Omega( w, c(x)) \big)
		\bigg)
		\geq
		\] 
		\[
		\geq
		\lim\limits_{ N \to \infty }
		\left( A_{max} - C_2 \right) 
		{ \tau_{\mathbf{E}(max)}( N,f,t,\tau ) }|_{t}^{t'}
		+
		C_2
		-
		\]
		\[
		-
		\left(
		\Omega(w,c(x)) \lg(N)
		+
		\Omega(w,c(x)) \big( C_1 + A(w) + A(c(x)) \big)
		+
		C_0 \big( 1 - \Omega( w, c(x)) \big)
		\right)
		\geq
		\]
		\[
		\geq
		\lim\limits_{ N \to \infty }
		\left( \lg(N) - C_4 - C_2 \right) 
		{ \tau_{\mathbf{E}(max)}( N,f,t,\tau ) }|_{t}^{t'}
		+
		C_2
		-
		\]
		\[
		-
		\Big(
		\Omega(w,c(x)) \lg(N)
		+
		\Omega(w,c(x)) \big( C_1 + A(w) + A(c(x)) \big)
		+
		C_0 \big( 1 - \Omega( w, c(x)) \big)
		\Big)
		=
		\]
		\[
		=
		\lim\limits_{ N \to \infty }
		\left( \lg(N) - C_4 - C_2 \right) 
		{ \tau_{\mathbf{E}(max)}( N,f,t, \tau ) }|_{t}^{t'}
		+
		C_2
		-
		\]
		\[
		-
		\Omega(w,c(x)) \lg(N)
		-
		\Omega(w,c(x)) C_1 - \Omega(w,c(x)) A(w) - \Omega(w,c(x)) A(c(x))
		-
		C_0 + C_0 \Omega( w, c(x))	
		=
		\]
		\[
		=
		\lim\limits_{ N \to \infty }
		\left( 
		{ \tau_{\mathbf{E}(max)}( N,f,t,\tau ) }|_{t}^{t'}
		-
		\Omega(w,c(x))
		\right)
		\lg(N)
		-
		\Omega(w,c(x)) A( c(x))
		-
		\]
		\[
		- ( C_4 + C_2 ) { \tau_{\mathbf{E}(max)}( N,f,t,\tau ) }|_{t}^{t'}
		-
		\Omega(w,c(x)) C_1
		+ 
		\Omega(w,c(x)) C_0
		+ C_2 - C_0
		-
		\Omega(w,c(x)) A(w)
		\geq 	
		\]
		\[
		\geq
		\lim\limits_{ N \to \infty }
		\left( 
		{ \tau_{\mathbf{E}(max)}( N,f,t,\tau ) }|_{t}^{t'}
		-
		\Omega(w,c(x))
		\right)
		\lg(N)
		-
		\Omega(w,c(x)) A( c(x))
		-
		\]
		\[
		- ( C_4 + C_2 ) 
		-
		C_1
		+ C_2 - C_0
		-
		A(w)
		=	
		\]
		\[
		=
		\lim\limits_{ N \to \infty }
		\left( 
		{ \tau_{\mathbf{E}(max)}( N,f,t,\tau ) }|_{t}^{t'}
		-
		\Omega(w,c(x))
		\right)
		\lg(N)
		-
		\Omega(w,c(x)) A( c(x))
		-
		\]
		\[
		- C_4 
		-
		C_0
		-
		C_1
		-
		A(w)
		\geq	
		\]
		\[
		\geq
		\lim\limits_{ N \to \infty }
		\left( 
		{ \tau_{\mathbf{E}(max)}( N,f,t,\tau ) }|_{t}^{t'}
		-
		\Omega(w,c(x))
		\right)
		\lg(N)
		-
		\]
		\[
		- \Omega(w,c(x)) \lg(x) - \Omega(w,c(x))(1+\epsilon)\lg(\lg(x)) - \Omega(w,c(x)) C_L
		-
		\]
		\[
		- \Omega( w, c(x) ) C_c
		- C_4 
		-
		C_0
		-
		C_1
		-
		A(w)	
		\]
		\[
		\geq
		\lim\limits_{ N \to \infty }
		\left( 
		{ \tau_{\mathbf{E}(max)}( N,f,t,\tau ) }|_{t}^{t'}
		-
		\Omega(w,c(x))
		\right)
		\lg(N)
		-
		\]
		\[
		- \Omega(w,c(x)) \lg(x) 
		- 2 \, \Omega(w,c(x))\lg(\lg(x)) 
		- C_5 
		- A(w)
		\]
	\end{proof}

\end{thm}

\noindent \\

\subsection{Theorem 2 Extended}

\begin{thm}[\ref{BthmMainCentralTime}]\label{thmMainCentralTime}
	\noindent \\
	
	Let $ w \in \mathbf{L_U} $ be a network input.\\
	
	Let $ 0 < N \in \mathbb{N} $.\\
	
	If there is $ 0 \leq z_0 \leq | \mathrm{T}(G_t) | -1 $ and $ \epsilon, \, \epsilon_2 > 0 $ such that 
	
	\[
	z_0 + f( N, t_{z_0}, \tau )  + 2
	= 
	\mathbf{ O }
	\left( \frac
	{ N^{ C } }
	{ \lg(N) } 
	\right)
	\]
	
	\noindent where
	
	\[ 
	0
	<
	C = 
	\frac
	{
		{ \tau_{\mathbf{E}(max)}( N,f,t_{z_0},\tau ) }|_{ t_{z_0} }^{ t_{ z_0 + f( N, t_{z_0}, \tau ) }  }
		-
		\Omega(w, c_0 + z_0 + f( N, t_{z_0}, \tau ) + 2 )
		-
		\epsilon
	}
	{ \Omega(w,  c_0 + z_0 + f( N, t_{z_0}, \tau ) + 2  ) }
	\leq
	\frac{1}{ \epsilon_2 }
	\] \\
	
	Then, for every non-decreasing total computable function $ \myfunc{c}{ \mathbb{N} } { \mathfrak{C_{BB}} } { x } { c(x)=y } $ where $ t_{z_0}, \, t_{ z_0 + f( N, t_{z_0}, \tau ) } \in \mathrm{T}(G_t) $ and $ c(z_0 + f( N, t_{z_0}, \tau ) + 2) \geq c_0 + z_0 + f( N, t_{z_0}, \tau ) + 2 $ and $\mathfrak{N}_{BB}(N,f,t_{z_0},\tau,j) = ( G_t, \mathfrak{ P }_{BB}(N), b_j ) $ is well-defined, we will have that there are $ t_{cen_2}(c) $ and $ t_{cen_1}(c) $ such that 
	
	\[
	t_{cen_2}(c) = t_{cen_1}(c) \leq t_{ z_0 }
	\]

	\noindent \\
	
	\begin{proof}
		\noindent \\
		We know from Corollary \ref{BcorMain} that
		
		\begin{equation}\label{stepFromcorMain}
		\lim\limits_{ N \to \infty } 
		\mathbf{E}_{ \mathfrak{N}_{BB}(N,f,t_z,\tau) } 
		\left(
		{ {\displaystyle{\myDelta_{iso}^{net}} A} (o_i, c( z + f( N, t_z, \tau ) + 2 ) )} 
		\right)
		\geq
		\end{equation}
		
		\begin{align*}
		& \geq
		\lim\limits_{ N \to \infty }
		\left( 
		{ \tau_{\mathbf{E}(max)}( N,f,t_{z},\tau ) }|_{ t_{z} }^{ t_{ z + f( N, t_{z}, \tau ) }  }
		-
		\Omega(w, c( z + f( N, t_z, \tau ) + 2 ))
		\right)
		\lg(N) - \\
		& - \Omega(w, c( z + f( N, t_z, \tau ) + 2 )) \lg( z + f( N, t_z, \tau ) + 2) - \\
		& - 2 \, \Omega(w, c( z + f( N, t_z, \tau ) + 2 ))\lg(\lg( z + f( N, t_z, \tau ) + 2)) - A(w) - C_5 \\
		\end{align*}
		
		Suppose that there is $ t_{z_0} \in \mathrm{T}(G_t) $, where $ 0 \leq {z_0} \leq | \mathrm{T}(G_t) | -1 $, and $\epsilon > 0$ such that
		
		\begin{equation}
		{z_0} + f( N, t_{z_0}, \tau ) + 2
		= 
		\mathbf{ O }
		\left( \frac
		{ N^{ C } }
		{ \lg(N) } 
		\right)
		\end{equation}
		
		\noindent where
		
		\[ 
		0 < C = 
		\frac
		{
			{ \tau_{\mathbf{E}(max)}( N,f,t_{z_0},\tau ) }|_{ t_{z_0} }^{ t_{ z_0 + f( N, t_{z_0}, \tau ) }  }
			-
			\Omega(w,  c_0 + z_0 + f( N, t_{z_0}, \tau ) + 2 )
			-
			\epsilon
		}
		{ \Omega(w, c_0 + z_0 + f( N, t_{z_0}, \tau ) + 2 ) }
		\] \\
		
		From Definition \ref{BdefOmegawc}, we have that, for every $y \in \mathbb{N}$, if $ y \geq c_0 + z_0 + f( N, t_{z_0}, \tau ) + 2 $, then
		
		\begin{equation}\label{stepNestingOmegas}
		\Omega(w,y) \leq \Omega( w, c_0 + z_0 + f( N, t_{z_0}, \tau ) + 2 )
		\end{equation} \\
		
		Thus, since we are assuming $ c(z_0 + f( N, t_{z_0}, \tau ) + 2) \geq c_0 + z_0 + f( N, t_{z_0}, \tau ) + 2 $, then, for fixed values of $ { \tau_{\mathbf{E}(max)}( N,f,t_{z_0},\tau ) }|_{ t_{z_0} }^{ t_{ z_0 + f( N, t_{z_0}, \tau ) }  } $ and $\epsilon$, we will have from Step \eqref{stepNestingOmegas} that

		\begin{equation}\label{stepCvsC'}
		\frac
		{
			{ \tau_{\mathbf{E}(max)}( N,f,t_{z_0},\tau ) }|_{ t_{z_0} }^{ t_{ z_0 + f( N, t_{z_0}, \tau ) }  }
			-
			\Omega(w, c( z_0 + f( N, t_{z_0}, \tau ) + 2 ))
			-
			\epsilon
		}
		{ \Omega(w, c( z_0 + f( N, t_{z_0}, \tau ) + 2 )) }
		\geq
		\end{equation}
		\[
		\geq
		\frac
		{
			{ \tau_{\mathbf{E}(max)}( N,f,t_{z_0},\tau ) }|_{ t_{z_0} }^{ t_{ z_0 + f( N, t_{z_0}, \tau ) }  }
			-
			\Omega(w,  c_0 + z_0 + f( N, t_{z_0}, \tau ) + 2 )
			-
			\epsilon
		}
		{ \Omega(w, c_0 + z_0 + f( N, t_{z_0}, \tau ) + 2 ) }
		=
		C
		\geq 0
		\] \\
		
		Let \[ C' = \frac
		{
			{ \tau_{\mathbf{E}(max)}( N,f,t_{z_0},\tau ) }|_{ t_{z_0} }^{ t_{ z_0 + f( N, t_{z_0}, \tau ) }  }
			-
			\Omega(w, c( z_0 + f( N, t_{z_0}, \tau ) + 2 ))
			-
			\epsilon
		}
		{ \Omega(w, c( z_0 + f( N, t_{z_0}, \tau ) + 2 )) } \]
		
		\noindent \\

		Remember that for every $x>0$ and $  t, t' \in \mathrm{T}(G_t) $ there is $\epsilon_2$ such that\footnote{ Remember that one can always have a program that halts for every input, so that it will also halts for every partial output and, hence, halt on every cycle --- see \ref{BdefOmegawc} . }
		
		\begin{equation}
		0 < \epsilon_2 \leq \Omega( w, x ) \leq 1
		\end{equation} 
		
		\noindent and, therefore, from the Definition \ref{BdefAveragePartitionmaxN_BB}, we will also have that
		
		\begin{equation}\label{stepBoundingC}
		\frac{-1 - \epsilon}{ \epsilon_2 }
		\leq
		\frac
		{
			{ \tau_{\mathbf{E}(max)}( N,f,t,\tau ) }|_{ t }^{ t' }
			-
			\Omega( w, x )
			- \epsilon
		}
		{ \Omega( w, x ) } \leq \frac{1}{ \epsilon_2 }
		\end{equation}
		
		Hence, from Steps \eqref{stepCvsC'} and \eqref{stepBoundingC} we will have that 
		
		\[
		z_0 + f( N, t_{z_0}, \tau )  + 2
		= 
		\mathbf{ O }
		\left( \frac
		{ N^{ C' } }
		{ \lg(N) } 
		\right)
		\]
		
		\noindent where
		
		\[ 
		0
		\leq
		C' = 
		\frac
		{
			{ \tau_{\mathbf{E}(max)}( N,f,t_{z_0},\tau ) }|_{ t_{z_0} }^{ t_{ z_0 + f( N, t_{z_0}, \tau ) }  }
			-
			\Omega(w, c( z_0 + f( N, t_{z_0}, \tau ) + 2 ) )
			-
			\epsilon
		}
		{ \Omega(w, c( z_0 + f( N, t_{z_0}, \tau ) + 2 ) ) }
		\leq
		\frac{1}{ \epsilon_2 }
		\] 
		
		\noindent \\
		
		And, since $ {z_0} + f( N, t_{z_0}, \tau ) + 2 $ is now assymptotically dominated by $ \frac
		{ N^{ C' } }
		{ \lg(N) } $, then by definition we will have that there is a constant $ C_6 $ such that
		
		\begin{equation}\label{stepMainthmMainCentralTime}
		\lim\limits_{ N \to \infty }
		\left( 
		{ \tau_{\mathbf{E}(max)}( N,f,t_{z_0},\tau ) }|_{ t_{z_0} }^{ t_{ z_0 + f( N, t_{z_0}, \tau ) }  }
		-
		\Omega(w, c( z_0 + f( N, t_{z_0}, \tau ) + 2 ))
		\right)
		\lg(N) - \\
		\end{equation}
		
		\begin{align*}
		& - \Omega(w, c( z_0 + f( N, t_{z_0}, \tau ) + 2 )) \lg( {z_0} + f( N, t_{z_0}, \tau ) + 2) - \\
		& - 2 \, \Omega(w, c( z_0 + f( N, t_{z_0}, \tau ) + 2 ))\lg(\lg( {z_0} + f( N, t_{z_0}, \tau ) + 2)) - A(w) - C_5 
		\geq \\
		& \geq
		\lim\limits_{ N \to \infty }
		\left( 
		{ \tau_{\mathbf{E}(max)}( N,f,t_{z_0},\tau ) }|_{ t_{z_0} }^{ t_{ z_0 + f( N, t_{z_0}, \tau ) }  }
		-
		\Omega(w, c( z_0 + f( N, t_{z_0}, \tau ) + 2 ))
		\right)
		\lg(N) - \\
		& - \Omega(w, c( z_0 + f( N, t_{z_0}, \tau ) + 2 )) \lg( C_6 \, \frac
		{ N^{ C' } }
		{ \lg(N) } ) - \\
		& - 2 \, \Omega(w, c( z_0 + f( N, t_{z_0}, \tau ) + 2 ))\lg( \lg( C_6 \, \frac
		{ N^{ C' } }
		{ \lg(N) } ) ) - A(w) - C_5 
		\geq \\
		& \geq
		\lim\limits_{ N \to \infty }
		\left( 
		{ \tau_{\mathbf{E}(max)}( N,f,t_{z_0},\tau ) }|_{ t_{z_0} }^{ t_{ z_0 + f( N, t_{z_0}, \tau ) }  }
		-
		\Omega(w, c( z_0 + f( N, t_{z_0}, \tau ) + 2 ))
		\right)
		\lg(N) - \\
		& - \Omega(w, c( z_0 + f( N, t_{z_0}, \tau ) + 2 )) \left( 
		\lg( C_6 ) + C' \, \lg( N ) - \lg(\lg(N))  )
		\right) - \\
		& - 2 \, \Omega(w, c( z_0 + f( N, t_{z_0}, \tau ) + 2 ))\lg\left( 
		\lg( C_6 ) + \lg( N^{ C' } ) - \lg(\lg(N))  ) 
		\right) - A(w) - C_5 
		\geq \\
		& \geq
		\lim\limits_{ N \to \infty }
		\left( 
		\epsilon
		\right)
		\lg(N) - \Omega(w, c( z_0 + f( N, t_{z_0}, \tau ) + 2 )) \left( 
		\lg( C_6 ) - \lg(\lg(N))  )
		\right) - \\
		& - 2 \,  \Omega(w, c( z_0 + f( N, t_{z_0}, \tau ) + 2 ))\lg\left( 
		\lg( C_6 ) + \lg( N^{ C' } ) - \lg(\lg(N))  ) 
		\right) - A(w) - C_5 
		\geq \\
		& \geq
		\lim\limits_{ N \to \infty }
		\left( 
		\epsilon
		\right)
		\lg(N) - \Big( 
		\lg( C_6 ) - \lg(\lg(N))  
		\Big) - 2 \, \lg\Big( 
		\lg( C_6 ) + \lg( N^{ C' } ) - \lg(\lg(N))  
		\Big) - \\
		& - A(w) - C_5 
		\geq \\
		& \geq
		\lim\limits_{ N \to \infty }
		\left( 
		\epsilon
		\right)
		\lg(N) - 
		\lg( C_6 ) + \lg(\lg(N)) - 2 \, \lg\left( \lg( N^{ \frac{1}{ \epsilon_2 } } ) \right) - A(w) - C_5 
		\geq \\
		& \geq
		\lim\limits_{ N \to \infty }
		\left( 
		\epsilon
		\right)
		\lg(N) - 
		\lg( C_6 ) + \lg(\lg(N)) 
		- 2 \, \lg( \frac{1}{ \epsilon_2 } \, \lg( N ) ) 
		- A(w) - C_5 
		\geq \\
		& \geq
		\lim\limits_{ N \to \infty }
		\left( 
		\epsilon
		\right)
		\lg(N) - 
		\lg( C_6 ) + \lg(\lg(N)) 
		- 2 \, \lg( \frac{1}{ \epsilon_2 } ) - 2 \, \lg( \lg( N ) )
		- A(w) - C_5 
		\geq \\
		& \geq
		\lim\limits_{ N \to \infty }
		\left( 
		\epsilon
		\right)
		\lg(N) - 
		\lg( C_6 ) 
		- 2 \, \lg( \frac{1}{ \epsilon_2 } )
		- \lg( \lg( N ) )
		- A(w) - C_5 
		=
		\infty \\
		\end{align*}

		Therefore, from Steps \eqref{stepFromcorMain} and \eqref{stepMainthmMainCentralTime}, we will have that

		\begin{equation}\label{stepPrevthmMainCentralTime}
		\lim\limits_{ N \to \infty } 
		\mathbf{E}_{ \mathfrak{N}_{BB}(N,f,t_{z_0},\tau) } 
		\left(
		{ {\displaystyle{\myDelta_{iso}^{net}} A} (o_i, c( z_0 + f( N, t_{z_0}, \tau ) + 2 ) ) } 
		\right)
		=
		\infty
		\end{equation} \\
		
		Therefore, directly from the Definitions \ref{BdefTimecentrality1} and \ref{BdefTimecentrality2} and Step \eqref{stepPrevthmMainCentralTime}, since $ t_{ z_0 } $ satisfies these definitions, we will have that
		
		\[
		t_{cen_2}(c) = t_{cen_1}(c) \leq t_{ z_0 }
		\]

	\end{proof}

\end{thm}

\end{document}